\def\draft{0}  

\documentclass[11pt,letterpaper]{article}
\usepackage[margin=1in]{geometry}

\ifnum\draft=1
\fi 

\usepackage[utf8]{inputenc}
\usepackage{complexity}
\usepackage{mathtools,amsmath,amssymb}
\usepackage{amsthm}
\usepackage{graphicx}
\usepackage{color}
\usepackage{multirow}
\usepackage{tabularx}
\usepackage[colorinlistoftodos]{todonotes}
\usepackage[colorlinks=true, allcolors=blue]{hyperref}
\usepackage{enumitem}
\usepackage{relsize}

\usepackage{soul}

\usepackage[colorlinks=true, allcolors=blue]{hyperref}
\usepackage[nameinlink,capitalise]{cleveref}
\hypersetup{
    citecolor={violet}
}

\ifnum\draft=1
\newcommand{\ameya}[1]{{\color{purple}[Ameya: #1]}}
\newcommand{\mnote}[1]{{\color{red} Madhu: #1}}
\newcommand{\cnote}[1]{{\color{blue} Chi-Ning: #1}}
\newcommand{\vnote}[1]{{\color{cyan}[Santhoshini: #1]}}
\newcommand{\snote}[1]{{\color{orange} Sasha: #1}}

\newcommand{\pnote}[1]{{\color{brown} $q$-ary version: #1}}

\else  \ifnum\draft=0
\newcommand{\ameya}[1]{}
\newcommand{\mnote}[1]{}
\newcommand{\cnote}[1]{}
\newcommand{\vnote}[1]{}
\newcommand{\snote}[1]{}
\newcommand{\pnote}[1]{}
\fi \fi 

\usepackage{thmtools}
\usepackage{thm-restate}
\RequirePackage[colorlinks=true]{hyperref}
\hypersetup{
  linkcolor=[rgb]{0,0,0.5},
  citecolor=[rgb]{0, 0.5, 0},
  urlcolor=[rgb]{0.5, 0, 0}
}

\allowdisplaybreaks

\usepackage[breakable]{tcolorbox}
\newtcolorbox{reduction}[2][]
{
  colframe = gray!50,
  colback  = gray!10,
  coltitle = gray!10!black,
  before skip = 10pt,
  after skip = 10pt,
  title    = \textbf{#2},
  #1,
  breakable,
}

\newtcolorbox{game}[2][]
{
  colframe = blue!50,
  colback  = blue!10,
  coltitle = blue!10!black,
  before skip = 10pt,
  after skip = 10pt,
  title    = \textbf{#2},
  #1,
  breakable,
}

\newtcolorbox{examplebox}[2][]
{
  breakable,
  colframe = gray!50,
  colback  = gray!10,
  coltitle = gray!10!black,
  before skip = 10pt,
  after skip = 10pt,
  title    = \textbf{#2},
  #1,
}

\usepackage{algorithmicx}
\usepackage{algorithm}
\usepackage[noend]{algpseudocode}

\algnewcommand\algorithmicinput{\textbf{Input:}}
\algnewcommand\Input{\item[\algorithmicinput]}

\algnewcommand\algorithmicoutput{\textbf{Output:}}
\algnewcommand\Output{\item[\algorithmicoutput]}

\usepackage{amsthm}
\usepackage{thmtools,thm-restate}

\numberwithin{equation}{section}
\declaretheoremstyle[bodyfont=\it,qed=\qedsymbol]{noproofstyle}

\declaretheorem[name=Observation,numbered=no]{observation*}

\declaretheorem[numberlike=equation]{theorem}

\declaretheorem[name=Theorem,numbered=no]{theorem*}

\declaretheorem[numberlike=equation]{lemma}
\declaretheorem[name=Lemma,numbered=no]{lemma*}

\declaretheorem[numberlike=equation]{corollary}
\declaretheorem[name=Corollary,numbered=no]{corollary*}

\declaretheorem[numberlike=equation]{proposition}
\declaretheorem[name=Proposition,numbered=no]{proposition*}

\declaretheorem[name=Proposition,numbered=no]{parameter*}

\declaretheorem[numberlike=equation]{claim}
\declaretheorem[name=Claim,numbered=no]{claim*}

\declaretheorem[name=Conjecture,numbered=no]{conjecture*}

\declaretheorem[name=Question,numbered=no]{question*}

\declaretheoremstyle[bodyfont=\it]{defstyle} 

\declaretheorem[numberlike=equation,style=defstyle]{definition}
\declaretheorem[unnumbered,name=Definition,style=defstyle]{definition*}

\declaretheorem[unnumbered,name=Example,style=defstyle]{example*}

\declaretheorem[unnumbered,name=Notation=defstyle]{notation*}

\declaretheorem[unnumbered,name=Construction,style=defstyle]{construction*}

\declaretheoremstyle[]{rmkstyle} 

\newtheorem*{remark}{Remark}

\newcommand{\wt}{\mathrm{wt}}

\newcommand{\Exp}{\mathop{\mathbb{E}}}

\newcommand{\val}{\textsf{val}}

\newcommand{\Unif}{\textsf{Unif}}

\newcommand{\supp}{\textsf{supp}}

\newcommand{\F}{\mathbb{F}}
\newcommand{\Z}{\mathbb{Z}}
\newcommand{\N}{\mathbb{N}}
\newcommand{\Comp}{\mathbb{C}}

\newcommand{\veca}{\mathbf{a}}
\newcommand{\vecb}{\mathbf{b}}
\newcommand{\vecc}{\mathbf{c}}
\newcommand{\vece}{\mathbf{e}}

\newcommand{\vecj}{\mathbf{j}}
\newcommand{\vecu}{\mathbf{u}}
\newcommand{\vecv}{\mathbf{v}}
\newcommand{\vecw}{\mathbf{w}}
\newcommand{\vecx}{\mathbf{x}}
\newcommand{\vecy}{\mathbf{y}}
\newcommand{\vecz}{\mathbf{z}}

\newcommand{\IRMD}{\textsf{IRMD}}
\newcommand{\IFRMD}{\textsf{IFRMD}}

\newcommand{\one}{\mathbf{1}}

\newcommand{\cD}{\mathcal{D}}
\newcommand{\cE}{\mathcal{E}}

\newcommand{\cF}{\mathcal{F}}

\newcommand{\maxF}{\textsf{Max-CSP}(\cF)}

\newcommand{\ALG}{\mathbf{ALG}}
\newcommand{\vecsigma}{\boldsymbol{\sigma}}
\newcommand{\vecnu}{\boldsymbol{\nu}}

\newcommand{\vecdelta}{\boldsymbol{\delta}}

\newcommand{\yes}{\textbf{YES}}
\newcommand{\no}{\textbf{NO}}

\renewcommand{\epsilon}{\varepsilon}
\newcommand{\eps}{\epsilon}

\title{Linear Space Streaming Lower Bounds for Approximating CSPs} 
\author{
Chi-Ning Chou\thanks{Center for Computational Neuroscience, Flatiron Institute, New York, New York, USA. Research supported in part by the Simons Foundation, and by NSF grants DMS-2134157 and CCF-1565264, DARPA grant W911NF2010021, DOE grant DE-SC0022199. Email: \texttt{cchou@flatironinstitute.org}.}
\and Alexander Golovnev\thanks{Department of Computer Science, Georgetown University. Supported in part by the NSF CAREER award (grant CCF2338730). Email: \texttt{alexgolovnev@gmail.com}.}
\and Madhu Sudan\thanks{School of Engineering and Applied Sciences, Harvard University, Cambridge, Massachusetts, USA. Supported in part by a~Simons Investigator Award and NSF Awards CCF 1715187 and CCF 2152413. Email: \texttt{madhu@cs.harvard.edu}.}
\and Ameya Velingker\thanks{This work was done while the author was at Google Research. Email: \texttt{ameyav@gmail.com}.}
\and Santhoshini Velusamy\thanks{Toyota Technological Institute at Chicago, Illinois, USA. Supported in part by a Google Ph.D. Fellowship, a Simons 
Investigator Award to Madhu Sudan, and NSF Awards CCF 1715187, CCF 2152413, and CCF 2348475. Email: \texttt{santhoshinivelusamy@gmail.com}.}
}
\begin{document}
\date{}
\maketitle

\begin{abstract}
    We consider the approximability of constraint satisfaction problems in the streaming setting. For every constraint satisfaction problem (CSP) on $n$ variables taking values in $\{0,\ldots,q-1\}$, we prove that improving over the trivial approximability by a factor of $q$ requires $\Omega(n)$ space even on instances with $O(n)$ constraints. 
    We also identify a broad subclass of problems for which any improvement over the trivial approximability requires $\Omega(n)$ space.
    The key technical core is an optimal, $q^{-(k-1)}$-inapproximability for the \textsf{Max $k$-LIN}-$\bmod\; q$ problem, which is the Max CSP problem where every constraint is given by a system of $k-1$ linear equations $\bmod\; q$ over $k$ variables. 
    
    Our work builds on and extends the breakthrough work of Kapralov and Krachun (Proc. STOC 2019) who showed a~linear lower bound on any non-trivial approximation of the MaxCut problem in graphs. MaxCut corresponds roughly to the case of  \textsf{Max $k$-LIN}-$\bmod\; q$ with ${k=q=2}$. For general CSPs in the streaming setting, prior results  only yielded $\Omega(\sqrt{n})$ space bounds. In particular no linear space lower bound was known for an approximation factor less than $1/2$ for {\em any} CSP. Extending the work of Kapralov and Krachun to \textsf{Max $k$-LIN}-$\bmod\; q$ to $k>2$ and  $q>2$ (while getting optimal hardness results) is the main technical contribution of this work. Each one of these extensions provides non-trivial technical challenges that we overcome in this work.
\end{abstract}

\newpage
\tableofcontents
\newpage

\section{Introduction}

In this work we consider the \emph{approximability of constraint satisfaction problems (CSPs)} by \emph{streaming algorithms} with sublinear space. We give tight inapproximability results for a broad class of CSPs, while giving somewhat weaker bounds on the approximability of every CSP. We introduce these terms below.

\subsection{Background} 
We consider the general class of constraint satisfaction problems with finite constraints over finite-valued variables. A \emph{problem} in this class, denoted $\maxF$, is given by positive integers $q$ and $k$ and a family of functions $\cF \subseteq \{f:\Z_q^k \to \{0,1\}\}$. An \emph{instance} of the problem consists of $m$ constraints placed on $n$ variables that take values in the set $\Z_q = \{0,\ldots,q-1\}$, where each constraint is given by a function $f \in \cF$ and $k$ distinct indices of variables $j_1,\ldots,j_k \in [n]$. Given an instance $\Psi$ of $\maxF$, the goal is to compute the \emph{value} $\val_\Psi$ defined to be the maximum, over all assignments to $n$ variables, of the fraction of constraints satisfied by the assignment. For $\alpha\in[0,1]$, the goal of the $\alpha$-approximate version of the problem is to compute an estimate $\eta$ such that $\alpha \cdot \val_\Psi \leq \eta \leq \val_\Psi$. 

In this work we consider the space complexity of approximating $\maxF$ by a single pass (potentially randomized) streaming algorithm that is presented the instance  $\Psi$ one constraint at a time. We consider ``non-trivial'' approximation algorithms for $\maxF$, where we first dismiss two notions of ``triviality''. First note that since we only consider space restrictions but not time restrictions, one can sample $O(n)$ constraints of $\Psi$ and solve the $\maxF$ problem on the sampled constraints optimally to get a $(1-\epsilon)$-approximation algorithm for every constant $\epsilon>0$ in $\widetilde{O}(n)$ space. Thus for this paper we view non-trivial algorithms to be those that run in $o(n)$ space.\footnote{We note that there is a gap between the $o(n)$ space we allow and the $O(n \log n)$ space that is trivial, but we are not able to get sharp enough lower bounds to address this gap.} The other form of ``triviality'' we dismiss is in the approximation factor.
Given a family $\cF$, let $\rho_{\min}(\cF)$ denote the infimum, over all instances $\Psi$ of $\maxF$, of the value $\val_\Psi$. Note that the algorithm that outputs the constant $\rho_{\min}(\cF)$ is a ($O(1)$-space!) $\rho_{\min}(\cF)$-approximation algorithm for $\maxF$. Thus we consider $\rho_{\min}(\cF)$ to be the ``trivial'' approximation factor for a family~$\cF$. With these two notions of ``triviality'' in mind, we define $\maxF$ to be {\em $\alpha$-approximable} (in the streaming setting) if $\alpha$ is the largest constant such that there exists an $\alpha$-approximation algorithm for $\maxF$ using $o(n)$ space. We simply say that $\maxF$ is {\em approximable} (in the streaming setting) if it is $\alpha$-approximable for some $\alpha > \rho_{\min}(\cF)$. We define a problem to be {\em approximation-resistant} (in the streaming setting) otherwise. 

\subsection{Results}\label{ssec:intro-results}

Our first main result in this paper gives a sufficient condition for a problem to be approximation resistant in the streaming setting.
We say that $f:\Z_q^k\to \{0,1\}$ is a {\em wide} constraint if there exists $\veca \in \Z_q^k$ such that for every $i \in \Z_q$ we have $f(\veca + i^k)=1$ where $i^k = (i,i,\ldots,i)$ and addition is performed in the group $\Z_q^k$. We say that a family $\cF$ is {\em wide} if every function $f \in \cF$ is wide. 

\begin{theorem}\label{thm:intro-approx-res}
	For every $q,k$ and every wide family $\cF$, $\maxF$ is approximation-resistant.
\end{theorem}

Many natural CSPs are wide, including Max $q$-colorability and Boolean problems such as Max $k$-SAT. Others, such as \textsf{Max $k$-LIN}-$\bmod\; q$ and the ``Unique Games'' problem,  contain wide subfamilies with the same ``trivial'' approximation factor, and thus \cref{thm:intro-approx-res} implies these are also approximation resistant. We elaborate on some of these examples in \cref{ssec:examples}. However, clearly wideness does not capture all CSPs. For general CSPs, while we do not pin down the approximability exactly, we do manage to pin it down up to a multiplicative factor of $q$. 

\begin{theorem}\label{thm:intro-q-factor}
	For every $q,k$ and every family $\cF$, if $\cF$ is $\alpha$-approximable then $\alpha \in [\rho_{\min}(\cF),q \cdot \rho_{\min}(\cF)]$.
\end{theorem}

Both \cref{thm:intro-approx-res,thm:intro-q-factor} follow from our more detailed \cref{thm:main-technical}. 
In \cref{ssec:examples} we give a few examples illustrating  how our theorems give tight lower bounds for some commonly studied CSPs including Max $q$-coloring, Unique Games, and Max Linear Systems.

Neither of the theorems above gives a complete classification of the approximability of CSPs in sublinear space. Contrasting with \cite{CGSV21-conference} one may have hoped that all lower bounds in \cite{CGSV21-conference} might simply extend, from ruling out $o(\sqrt{n})$-space sketching algorithms, to ruling out $o(n)$-space sketching algorithms. However subsequent work has shown that this hope is not realizable. Specifically Saxena, Singer, Sudan and Velusamy~\cite{SaxenaSSV23} have shown that the Max Dicut problem allows an $\widetilde{O}(\sqrt{n})$-space sketching algorithm that gets a $.485$ which beats the $4/9$-approximation upper bound for $o(\sqrt{n})$-space algorithms, from the work of Chou, Golovnev and Velusamy~\cite{CGV20}. Indeed there seems to be broader class of problems that might allow such improvements in $o(n)$-space. This is hinted at in the work of Singer~\cite{Singer23} who shows that for every $k\geq 2$, there is a $\widetilde{O}(n^{1-1/k})$-space algorithm for {\em bounded-degree} instances of the Max $k$-AND problem that beats the approximability upper bound given in Boyland, Hwang, Prasad, Singer and Velusamy~\cite{BoylandHPSV22} for $o(\sqrt{n})$-space sketching algorithms. (A CSP instance has bounded degree if each variable appears in $O(1)$ constraints. Note that all lower bounds in this paper and prior works are proven for bounded degree instances.)
And for the Max Dicut problem on bounded degree instances, Saxena, Singer, Sudan and Velusamy~\cite{SaxenaSSV25} gave $1/2 - \epsilon$ approximation algorithms, for every $\epsilon > 0$, using $o(n)$ space.
Their result was recently generalized to arbitrary instances by Azarmehr, Behnezhad, Ferante, and Sanneian~\cite{ABFS25}.
Thus the class of problems for which linear space upper bounds on the approximability match the performance of polylogarithmic space sketching algorithms is a strict subclass of all MaxCSPs. Finding where exactly this boundary lies remains a wide open question.

\subsection{Prior work}

There have been a number of works in the broad area of approximations for streaming constraint satisfaction problems and lower bound techniques for those~\cite{GKKRW,verbin2011streaming,KKS,assadi2016tight,KKSV17,GVV17,GT19,KK19,CGV20,assadi2020multi,assadi2021graph,CGSV21-conference,SSV21}.
Among these our work is the \emph{first work to aim to get tight inapproximability results for a broad class of CSPs for almost linear space single-pass streaming algorithms}.
Previous works either did not get tight approximation factors or were aimed at specific problems or only got $\Omega(\sqrt{n})$-space lower bounds, though some do target multi-pass streaming algorithms \cite{assadi2020multi,assadi2021graph} --- which we do not do here. We describe the state of the art prior to our work below. (More detailed descriptions of prior works can be found in \cite{CGSV21-conference}.)

On the front of general lower bounds, 
Chou, Golovnev, Sudan and Velusamy~\cite{CGSV21-conference} explored the same set of CSP problems as we do, i.e, $\maxF$ for arbitrary $q,k$ and $\cF$. Their focus is on looser space lower bounds: specifically, they focus on problems that require $n^{\Omega(1)}$ space vs. those where $n^{o(1)}$ space suffices. They give a complete dichotomy for sketching algorithms, a special class of streaming algorithms. They also give sufficient conditions for approximation resistance with respect to sub-polynomial space general streaming algorithms. Theorem 2.9 in their paper shows that families $\cF$ where the satisfying assignments of every function in the class support a one-wise independent distribution are approximation resistant. This theorem is incomparable with our \cref{thm:intro-approx-res} in that they give approximation resistance for a broader collection of problems (all wide families support one-wise independence) but the space lower bound is weaker --- they give an $\Omega(\sqrt{n})$ lower bound and we get $\Omega(n)$ lower bounds for wide families. 
\cite{CGSV21-conference} does not give an analogue of our \cref{thm:intro-q-factor}, though such a result (with the weaker $\Omega(\sqrt{n})$ space lower bound) can be derived from their theorems equally easily. Indeed, our \cref{sec:streaming} is based on their work. 

Turning to linear space lower bounds, the breakthrough work here is due to Kapralov and Krachun~\cite{KK19}, who show that approximating Max Cut (which translates in our setting to $\maxF$ for $\cF = \{\oplus_2\}$ where $\oplus_2:\{0,1\}^2 \to \{0,1\}$ is the binary XOR function) to within a factor $\frac12+\epsilon$ requires $\Omega(n)$ space for every $\epsilon >0$. Indeed, our work builds on their work and we compare our techniques later.
Prior to the work of Kapralov and Krachun, there was a weaker result due to Kapralov, Khanna, Sudan and Velingker~\cite{KKSV17} showing that there exists $\epsilon > 0$ such that  ($1-\epsilon$)-approximation for Max Cut requires linear space. Finally, Chou, Golovnev and Velusamy~\cite{CGV20} get a tight
inapproximability for Max Exact 2-SAT (corresponding to $\maxF$ for $\cF = \{\vee_2\}$, where $\vee_2:\{0,1\}^2\to\{0,1\}$ is the binary OR function) for linear space algorithms, by a reduction from Max Cut. 

Thus, prior to our work it was conceivable (though of course extremely unlikely) that every $\maxF$ allowed a $1/2$-approximating streaming algorithm using $o(n)$ space. Our work is the first to prove inapproximability $\alpha \leq 1/2$ for any $\maxF$. Indeed, we get inapproximabilities going to $0$ either as $q \to \infty$ (e.g., for the Unique Games problem) or as $k \to \infty$ (e.g., for the Max $k$-equality problem with $q = 2$ as defined later in \cref{ssec:intro-tech}).

The main contribution of our work is to extend the techniques of \cite{KK19} to problems beyond Max Cut.  Indeed the bulk of our proof takes the tour-de-force proof in \cite{KK19} and finds the correct replacements in our setting. In the process, we arguably even present cleaner abstractions of their work. We elaborate on this further in the next section but first comment on why we feel the extensions are not straightforward given \cite{KK19}. First we note that the exact class of problems we are able to deal with in \cref{thm:intro-approx-res} is not the fullest extension one may hope for. At the very least we have expected to cover the same set of problems as \cite[Theorem 2.9]{CGSV21-conference}, i.e., families supporting one-wise independent distributions, but this remains open. Indeed to get our extensions we have to formulate a new communication problem which generalizes the one in \cite{KK19} and is different from the many variations considered in \cite{CGSV21} and \cite{CGSV21-conference}. In particular we are forced to work with a less expressive set of communication problems that already forces a ``linear-algebraic'' restriction on the core problems we work with. (We do believe a slight extension of our results to ``families containing  one-wise independent cosets of $\Z_q^k$'' should be more feasible.) Having identified the right set of problems, carrying out the proof of Kapralov and Krachun is still non-trivial. In particular one has to be careful to ensure that the improvement in the exponent of the space bound (from $n^{1/2}$ to $n$) is by a full factor of $2$ and not a factor of $k/(k-1)$, which is what one natural extension would lead to! We comment on these improvements in greater detail in the following.

\subsection{Techniques and new contributions}\label{ssec:intro-tech}

There are two lines of previous work that seem relevant to this work and we discuss our technical contributions relative to those here. We start with quick comparison with the previous work~\cite{CGSV21-conference} that gives $\Omega(\sqrt{n})$ lower bounds for a broader subset of problems than those addressed in this paper. We then move on to the work \cite{KK19} which is much closer to our work and needs more detailed comparison.

\paragraph{Comparison with~\cite{CGSV21-conference}.}  While there is some obvious overlap in the set of problems considered in \cite{CGSV21-conference}  and this paper (and also in the set of authors) we claim that, beyond this aspect, the overlap in techniques is minimal. Both papers do use lower bounds on communication problems to establish lower bounds on streaming CSPs (which is standard in the context of streaming lower bounds). But the exact set of communication problems is different, and the tools used to establish the lower bounds are also different. In particular, \cite{CGSV21-conference} create roughly a new communication problem for every $\gamma,\beta$ and $\cF$ and the main technical contributions there are lower bounds for these problems achieved mainly through a rich set of reductions among these communication problems. In our work we essentially work with one communication problem (once we fix $k$ and $q$) and the core of our work is proving a lower bound for this problem. (This lower bound is based on extending \cite{KK19} and we will elaborate on this later.) We use this one problem to get hardness for many different $\gamma,\beta$ and $\cF$ --- this part is arguably related to the work of \cite{CGSV21-conference} but we feel this is the obvious part of their work as well as our work. Finally, turning to the communication problems, the natural communication problems used to analyze streaming complexity involves one way communication among a large constant number of players. The exact problem of this type that we focus on is different from the ones considered in \cite{CGSV21-conference} due to a concept we call ``folding''. Folding makes our problems too restrictive to work for \cite{CGSV21-conference} (i.e., would prevent them for addressing every $(\gamma,\beta)-\maxF$), whereas we do not know how to get our lower bounds without folding. We also note that \cite{CGSV21-conference} derive their multiplayer lower bounds from lower bounds for a corresponding 2-player game and all their reductions work only for these 2-player games, which are inherently limited to yielding $\Theta(\sqrt{n})$ space lower bounds.

We now turn to the more significant comparison, with \cite{KK19}. We start with a quick review of the main steps of \cite{KK19} and then describe our analysis and conclude with a summary of the differences/new contributions relative to \cite{KK19}. 

\paragraph{Summary of \cite{KK19}.}
Kapralov and Krachun~\cite{KK19} work with a distributional $T$-player one-way communication game for some constant $T$. The game also has a parameter $\alpha > 0$. In instances of length $n$ of this game,  $T$ players $P_1,\ldots,P_T$ get partial matchings $M_1,\ldots,M_T$ on the vertex set $[n]$ along with respective binary labels $\vecz_1,\ldots,\vecz_T$ on the edges of the matchings, i.e., player $t$ receives input $(M_t,\vecz_t)$. Each partial matching contains $\alpha n$ edges, while each corresponding label $\vecz_t$ is an element of $\{0,1\}^{\alpha n}$. In the communication game, the players sequentially broadcast messages as follows. Player $t \in [T-1]$ computes a small message $c_t$ which is a function of $M_t,\vecz_t$ and all ``previous messages'' $c_1,\ldots,c_{t-1}$,\footnote{For technical reasons the lower bounds are proved in the stronger model where player $t$ gets $M_1,\ldots,M_{t-1}$ as well, but this difference is not crucial for the current discussion.} after which the $T$-th player outputs a single $0/1$ bit that is said to be the output of the communication protocol. The complexity of the protocol is the maximum over $t\in [T]$ of the message length $c_t$, and the goal of the players is to distinguish input instances drawn according to a \yes\ distribution from those drawn according to a \no\ distribution, defined as follows.

In instances chosen from the \no\ distribution, the matchings $M_1,\ldots,M_T$ are chosen uniformly and independently from the set of matchings containing $\alpha n$ edges on the vertex set $[n]$. Furthermore, the vectors $\vecz_1,\ldots,\vecz_T$ are chosen uniformly and independently from $\{0,1\}^{\alpha n}$. In the \yes\ distribution, the matchings are chosen as in the \no\ distribution, but in order to generate $\vecz_1,\ldots,\vecz_T$, we choose a common hidden vector $\vecx^* \in \{0,1\}^n$ uniformly at random and set each $\vecz_t$ as $\vecz_t(e) = x^*_a \oplus_2 x^*_b$ for every edge $e = (a,b)$. Thus, the label $\vecz_t$ can be viewed as specifying which edges of the $i$-th matching cross the cut determined by $\vecx^*$. If $T \gg \frac1\alpha$ then it can be seen that the \yes\ and \no\ distributions are very far. The key theorem shows that for every $\alpha > 0$ and $T$, any protocol distinguishing \yes\ instances from \no\ instances with constant advantage requires $\Omega(n)$ space. With this lower bound a space lower bound on Max Cut is straightforward.

Turning to the communication lower bound, the focus of the analysis are the sets $B_1,\ldots,B_T \subseteq \{0,1\}^n$ corresponding to the purported hidden vector $\vecx^*$ that are consistent with the messages $c_1,\ldots,c_T$. Specifically for $t \in [T]$, $B_t$ is the set of all vectors $\vecx^*$ that are consistent with the first $t$ matchings $M_{1:t}$ and the first $t$ messages $c_{1:t}$. Kapralov and Krachun~\cite{KK19} argue that the sets $B_t$ are not shrinking too fast (in either the \yes\ case or the \no\ case) using a property that they term ``$C$-boundedness,'' defined by the Fourier spectrum of the indicator function of $B_t$ (the function from $\{0,1\}^n$ to $\{0,1\}$ that is $1$ on $B_t$). We do not give the exact definition of boundedness here but roughly describe it as follows: Given an arbitrary set $B$ of size $S$ and a Fourier weight~$w$, the total Fourier mass (strictly the $\ell_1$-mass)  of the $w$-th level Fourier coefficients of $B$ is well-known (by classical Fourier analysis) to be bounded by some amount $U(w) = U_{S,n}(w)$. For $C$-bounded sets, the corresponding Fourier mass is required to be at most $C^w U(w/2)$. The factor of two gained here in the argument of $U$ is the crux to improvement in the space lower bound from $\sqrt{n}$ to $n$. (If the right hand side had been of the form $C^w U(\alpha w)$ then the space lower bound would be $\Omega(n^{1/(2\alpha)})$.) This factor of two, in turn, is attributable to the fact that the $\vecz_t$ only contain information about pairs of bits of $\vecx^*$.
Their analysis shows that, for every $t$, $B_t$ is $C_t$-bounded for some constant $C_t$. (The proof is inductive on $t$ but the inductive hypothesis is complex and we won't reproduce it here.) They further show that if $B_T$ is $C$-bounded for some constant $C$, then the distinguishing probability is at most $o(1)$.

\paragraph{Our Analysis.}
The core of our paper essentially focuses on the setting posed by one problem for every given $q$ and $k$, which we call \textsf{Max $k$-LIN}-$\bmod\; q$. This is the MaxCSP problem where every constraint is a conjunction of $k-1$ linear equations on $k$ variables. Our main lower bound aims to prove a tight $q^{-(k-1)}+\epsilon$-inapproximability of this problem for every $q$, $k$ and $\epsilon > 0$. (See \cref{thm:main IFRMD} and the following remark.) 
We formally prove this in approximability in Example 4 in \cref{ssec:examples} where we consider an even broader set of problems $\textsf{Max-Lin}_{k,r,q}$ whose constraints are conjunctions of $r$ linear equations over $k$ variables and give a tight $q^{-r}+\epsilon$ inapproximability for this problem for every $1 \leq r \leq k-1$.

To study this problem we introduce a $T$-player communication problem that we call the ``Implicit Randomized Mask Detection Problem'' (IRMD) described as follows: There are $T$ players each of whom receives an $\alpha n$ $k$-hypermatching $M_t$ (i.e., a set of $\alpha n$ $k$-uniform hyperedges on $[n]$ that are pairwise disjoint). Additionally, the players receive a label in $\Z_q^{k}$ for every hyperedge they see. Thus the $i$-th player's input is $(M_t,\vecz_t)$ where $\vecz_t \in (\Z_q^{k})^{\alpha n}$. In the \no\ distribution the $\vecz_t$'s are drawn uniformly. In the \yes\ distribution a vector 
$\vecx^* \in [q]^n$ is drawn uniformly and the label associated with an edge $\vecj = (j_1,\ldots,j_k)$ is $(x^*_{j_1} + a_{\vecj},\ldots,x^*_{j_k} + a_{\vecj})$ where $a_{\vecj} \in [q]$ is chosen uniformly and independently for each edge in each matching. The goal of the players is to distinguish between the \yes\ and \no\ distributions with minimal communication (with ``one-way'' communication from the $t$-th player to all higher numbered players, as before).

To lower bound the communication complexity of IRMD we consider a  folded version of the problem we call IFRMD where the labels associated with an edge are from 
$\Z_q^{k-1}$ and obtained by mapping an IRMD label $\vecz = (z^{(1)},\ldots,z^{(k)}) \in \Z_q^k$ to the label $\tilde{\vecz} = (z^{(2)} - z^{(1)},\ldots,z^{(k)}-z^{(1)})$.
With this folding we recover the same communication problem as \cite{KK19} for the case of $k=q=2$ and the main focus of our work is proving lower bounds for higher $k$ and $q$. 

Our analysis of the communication complexity of IFRMD follows the same sequence of steps (with imitation even within the steps) as \cite{KK19}. In particular we also use the same sets $B_1,\ldots,B_T$ and use the same notion of boundedness. 

Turning to the induction and the analysis of boundedness of $B_t$ for general $t$,  we are able to extract a clean lemma (\cref{lem:induction step}) that makes the induction completely routine. To explain this contribution note that $B_t$ is the intersection of $B_{t-1}$ with a set say $A_t$ where $A_t$ is of the same type as $B_t$ (both are obtained by looking at the vector $\vecx^*$ projected to a matching followed by some folding). Thus both $B_{t-1}$ and $A_t$ are bounded sets. To complete the induction it would suffice to prove that the intersection of bounded sets is bounded, but alas this is not true! To get that $B_t$ is bounded, we need to use the fact that the matching $M_t$ is random and chosen independently of $B_{t-1}$ but it turns out that that is all that is needed. This is exactly what we show in \cref{lem:induction step} --- and of course this only happens with high probability over the choice of $M_t$.

\paragraph{Incremental contribution over~\cite{KK19}.}
Given that our result closely follows~\cite{KK19} we now focus on some key differences, and why these contributions are conceptually significant. 
\begin{enumerate}
    \item 
    The analysis of \cite{KK19} is intricate and it is not a~priori clear what problems it may extend~to. Our choice of \textsf{Max $k$-LIN}-$\bmod\; q$ is not the obvious choice, and was not our first choice. More natural choices would be to go for more general linear systems, or even functions supporting ``one-wise independence'', but we are unable to push the analysis to more general cases. Our choice reflects an adequate one to get coarse bounds on the approximability of every problem while getting tight ones for many natural ones.
    \item The choice of the communication problems to work with is also not obvious: Indeed working with both IRMD and IFRMD seems necessary for our approach --- the former is more useful for our final inapproximability results whereas the latter is the one we are able to analyze. 
    \item The exact notion of boundedness that is necessary and sufficient for our results is also not completely obvious. It is only in hindsight, after carrying out the entire analysis, does it become clear that the notion that works is exactly the same as the one in \cite{KK19}. Part of the challenge is that in the inductive proof of boundedness even the base case (which is quite simple in \cite{KK19}) is not obvious in our case, and nor is the inductive step. 
    \begin{itemize}
       \item With respect to the base case we note that if we had adopted a weaker notion of boundedness allowing $w$-th level Fourier mass to grow roughly as $U((k-1)w/k)$ boundedness would have been easier to prove but the result would not be optimal. Getting a bound of $U(w/2)$ is not technically hard, but involves a non-trivial randomization in the choice of folding purely for analysis purposes. (So there is an implicit passing back and forth between the IRMD and IFRMD problems in this technical step.)
        \item We also feel that it is important that we are able to extract an induction lemma (\cref{lem:induction step}) that clearly separates the (Fourier and combinatorial) analytic ingredients from the probabilistic setup. We believe the lemma is clarifying even when applied to the proof of \cite{KK19}. 
    \end{itemize}
    \item Finally we note that the underlying combinatorics are made significantly more intricate due to the need to work with $k > 2$. A conceptual difference from \cite{KK19} here is that whereas they explore the distribution of the number of edges in a random matching that intersect with a fixed set of vertices, we have to explore the distribution of edges that have an odd intersection (or non-zero mod $q$ intersection) with a random hypermatching. Indeed this part is clarifying the role of some of the quantities explored in the previous work. Additionally, we note that the number of parameters we have to track is much larger (and indeed it is fortunate that the number of parameters remains a constant independent of $k$), and managing these in our inequalities is a non-trivial technical challenge (even given the heavy lifting in \cite{KK19}). 
\end{enumerate}

\subsection{Organization of the rest of the paper} \label{ssec:intro-org}
We start with some background material in \cref{sec:prelim}. We introduce our communication problems (\IRMD\ and \IFRMD) in \cref{sec:comm} and state our lower bounds for these. We use these lower bounds on communication problems to prove our streaming lower bounds in \cref{sec:streaming}, and turn to proving the communication lower bounds in \cref{sec:proof overview}. To do so, \cref{sec:proof overview} introduces the notion of bounded sets, states three lemmas on the properties of bounded sets, and proves the lower bound  assuming these lemmas on the boundedness of sets encountered by the protocol.
Finally \cref{sec:reduced bounded} proves these lemmas on boundedness, concluding the proofs.


\section{Preliminaries} \label{sec:prelim}

We use the following notations throughout the paper. Let $\N=\{1,\dots\}$ denote the set of natural numbers and let $[n]=\{1,2,\dots,n\}$. For a discrete set $X$ and a function $f:X\rightarrow\mathbb{R}$, we denote $\|f\|_p=(\sum_{x\in X}|f(x)|^p)^{1/p}$ for every $p>0$ and $\|f\|_0=\sum_{x\in X}\mathbf{1}_{f(x)\neq 0}$. For a sequence of objects $O_1,O_2,\dots,O_T$, we define $O_{1:t}=\{O_1,O_2,\dots,O_t\}$ for every $t\in[T]$.

\subsection{Total variation distance}
In our analysis we will use the total variation distance between probability distributions, and several bounds on it presented in this section.

\begin{definition}[Total variation distance of discrete random variables]
	Let $\Omega$ be a finite probability space and $X,Y$ be random variables with support $\Omega$. The total variation distance between $X$ and $Y$ is defined as follows.
	\[
	\|X-Y\|_{tvd} :=\frac{1}{2}
	\sum_{\omega\in\Omega}\left|\Pr[X=\omega]-\Pr[Y=\omega]\right| \, .
	\]
\end{definition}
We will use the triangle and data processing inequalities for the total variation distance.
\begin{proposition}[E.g.,{\cite[Claim~6.5]{KKS}}]\label{prop:tvd properties}
	For random variables $X, Y$ and $W$:
	\begin{itemize}
		\item (Triangle inequality) $\|X-Y\|_{tvd}\geq\|X-W\|_{tvd}-\|Y-W\|_{tvd}$.
		\item (Data processing inequality) If $W$ is independent of both $X$ and $Y$, and $f$ is a function, then  $\|f(X,W)-f(Y,W)\|_{tvd}\leq\|X-Y\|_{tvd}$.
	\end{itemize}
\end{proposition}

\begin{lemma}\label{lem:statistical_test}
	Let $X,~Y,~W$ be random variables and let $f$ be a function. If there exists $\delta>0$ such that for every fixed $x$ in the support of $X$, we have
	\[\|f(x,Y)-f(x,W)\|_{tvd}\le \delta\, , \] then the following holds:
	\[\|(X,f(X,Y))-(X,f(X,W))\|_{tvd}\le \delta\, .\]
\end{lemma}

\begin{proof}
	Consider any statistical test\footnote{That is, $T(X,Z)$ is a Boolean function that aims to maximize $\Exp_{(X,Z) \sim (X,f(X,Y)}[T(X,Z)] - \Exp_{(X,Z) \sim (X,f(X,W)}[T(X,Z)]$. Note that $\|(X,f(X,Y))-(X,f(X,W))\|_{tvd} = \max_T \{\Exp_{(X,Z) \sim (X,f(X,Y)}[T(X,Z)] - \Exp_{(X,Z) \sim (X,f(X,W)}[T(X,Z)]\}$.} $T$ distinguishing the joint distributions $(X,f(X,Y))$ and $(X,f(X,W))$. It suffices to prove that
	\[\mathbb{E}_{X,Y}[T(X,f(X,Y))]-\mathbb{E}_{X,W}[T((X,f(X,W)))]\le \delta \, .\]
	We have 
	\begin{align*}
		&\mathbb{E}_{X,Y}[T(X,f(X,Y))]-\mathbb{E}_{X,W}[T((X,f(X,W)))] \\
		&=  \mathbb{E}_{x\sim X}\left[\mathbb{E}_{y\sim Y\mid X=x}[T(x,f(x,y))]\right] - \mathbb{E}_{x\sim X}\left[\mathbb{E}_{w\sim W\mid X=x}[T(x,f(x,w))]\right]\\
		&= \mathbb{E}_{x\sim X}\left[\mathbb{E}_{y\sim Y\mid X=x}[T(x,f(x,y))]-\mathbb{E}_{w\sim W\mid X=x}[T(x,f(x,w))]\right]\\
		&\le \mathbb{E}_{x\sim X}[\delta] = \delta \, ,
	\end{align*}
	where the last step follows from the hypothesis that for every fixed $x$, we have \[\|f(x,Y)-f(x,W)\|_{tvd}\le \delta\, . \]
\end{proof}
We will also need the following lemma from \cite{KK19}.

\begin{lemma}[{\cite[Lemma B.2]{KK19}}]\label{lem:KKsubstitutionlemma}
	Let $X^1,X^2$ be random variables taking values on finite sample space~$\Omega_1$. Let $Z^1,Z^2$ be random variables taking values on finite sample space $\Omega_2$, and suppose that $Z^2$ is independent of $X^1,X^2$. Let $f:\Omega_1\times\Omega_2 \rightarrow \Omega_3 $ be a function. Then
	\[
	\|(X^1,f(X^1,Z^1))-(X^2,f(X^2,Z^2)) \|_{tvd} \le \|(X^1,f(X^1,Z^1))-(X^1,f(X^1,Z^2)) \|_{tvd} + \|X^1-X^2\|_{tvd} \, .
	\]
	
\end{lemma}
\subsection{Concentration inequality}

We will use the following concentration inequality from \cite{KK19} which is essentially an Azuma-Hoeffding style inequality for submartingales.

\begin{lemma}[\protect{\cite[Lemma~2.5]{KK19}}]\label{lem:our-azuma}
	Let $X=\sum_{i\in[N]}X_i$ where $X_i$ are Bernoulli random variables such that for every $k\in[N]$, $\Exp[X_k \, |\, X_1,\dots,X_{k-1}]\leq p$ for some $p\in(0,1)$. Let $\mu=Np$. For every $\Delta>0$, we have:
	\[
	\Pr\left[X\geq\mu+\Delta\right]\leq\exp\left(-\frac{\Delta^2}{2\mu+2\Delta}\right) \, .
	\]
\end{lemma}

\subsection{Fourier analysis}\label{sec:fourier}
In this paper, we will use Fourier analysis over $\Z_q$ (see, for instance,~\cite{o2014analysis,GT19}).
For a function $f: \Z_q^n \to \Comp$, its Fourier coefficients are defined by $\widehat{f}(\vecu) = \frac{1}{q^n}\sum_{\veca\in\Z_q^n}f(\veca)\cdot\overline{\omega^{\vecu^\top\veca}}$, where $\vecu\in\Z_q^n$ and $\omega=e^{2\pi i/q}$ is the primitive $q$-th root of unity. In particular, for every $\veca$,  $f(\veca)=\sum_{\vecu\in\Z_q^n}\widehat{f}(\vecu)\cdot\omega^{\vecu^\top\veca}$.
Later we will use the three following important tools.
Note that here we define the $p$-norm of $f$ as $\|f\|_p^p=\sum_{\vecx\in\Z_q^n}|f(\vecx)|^p$ rather than the standard definition which uses expectation. This is for future notational convenience.

\begin{lemma}[Parseval's identity]\label{prop:parseval}
	For every function $f:\Z_q^n\to\Comp$, 
	\[
	\|f\|_2^2=\sum_{\veca\in\Z_q^n}f(\veca)^2=q^n\sum_{\vecu\in\Z_q^k}\widehat{f}(\vecu)^2 \, .
	\]
\end{lemma}

Note that for every distribution $f$ on $\Z_q^n$, $\widehat{f}(0^n)=q^{-n}$. For the uniform distribution $U$ on $\Z_q^n$, $\widehat{U}(\vecu)=0$ for every $\vecu\neq0^n$. Thus, by \autoref{prop:parseval}, for any distribution $f$ on $\Z_q^n$:
\begin{align}\label{eq:dist}
	\|f-U\|_2^2=q^n\sum_{\vecu\in\Z_q^n}\left(\widehat{f}(\vecu)-\widehat{U}(\vecu)\right)^2=q^n\sum_{\vecu\in\Z_q^n\backslash\{0^n\}}\widehat{f}(\vecu)^2 \, .
\end{align}

We now introduce some standard facts about how convolutions interact with the Fourier transform operation. For functions $f,g\colon \Z_q^n \to \Comp$, their convolution $f\star g\colon \Z_q^n \to \Comp$ is defined as $(f\star g)(\veca)=\sum_{\vecv\in\Z_q^n} f(\vecv) g(\veca-\vecv)$. The first lemma is the so-called ``convolution theorem,'' which essentially states that, up to normalization factors, the Fourier transform of the convolution of two functions is equal to the pointwise product of the individual Fourier transforms.
\begin{lemma}[Convolution Theorem]\label{lem:convthm}
	For $f,g: \Z_q^n \to \Comp$, we have
	\[
	\widehat{f\star g}(\vecu) = q^n \cdot \widehat{f}(\vecu)\cdot\widehat{g}(\vecu).
	\]
	for all $\vecu \in \Z_q^n$.
\end{lemma}
\begin{proof}
	For every $\vecu$,
	\begin{align*}
		\widehat{f\star g}(\vecu) &= \frac{1}{q^n} \sum_{\veca \in \Z_q^n} (f\star g)(\veca) \cdot \overline{\omega^{\vecu^\top \veca}}\\
		&= \frac{1}{q^n} \sum_{\veca \in \Z_q^n} \left(\sum_{\vecv\in\Z_q^n} f(\vecv) g(\veca-\vecv)\right)  \overline{\omega^{\vecu^\top \veca}}\\
		&= \frac{1}{q^n} \sum_{\veca \in \Z_q^n} \sum_{\vecv\in\Z_q^n} f(\vecv) \overline{\omega^{\vecu^\top \vecv}} \cdot g(\veca-\vecv) \overline{\omega^{\vecu^\top (\veca-\vecv)}} \\ 
		&= \frac{1}{q^n} \sum_{\vecv\in\Z_q^n} f(\vecv) \overline{\omega^{\vecu^\top \vecv}} \cdot \sum_{\veca \in \Z_q^n} g(\veca-\vecv) \overline{\omega^{\vecu^\top (\veca-\vecv)}}\\
		&= q^n \cdot \frac{1}{q^n}\sum_{\vecv\in\Z_q^n} f(\vecv) \overline{\omega^{\vecu^\top \vecv}} \cdot \frac{1}{q^n}\sum_{\veca \in \Z_q^n} g(\veca) \overline{\omega^{\vecu^\top \veca}}\\
		&= q^n \cdot \widehat{f}(\vecu) \cdot \widehat{g}(\vecu),
	\end{align*}
	as desired.
\end{proof}

We will also need the following lemma, which states that the Fourier transform of the \emph{product} of two functions is given by the convolution of the individual Fourier transforms.
\begin{lemma}[Fourier transform of product of functions]\label{lem:convolution}
	For every $f,g:\Z_q^n\to\Comp$, and $\vecu\in\Z_q^n$, we have
	\[
	\widehat{f\cdot g}(\vecu) = \sum_{\vecu'\in\Z_q^n}\widehat{f}(\vecu')\cdot\widehat{g}(\vecu-\vecu')\,.
	\]
	Furthermore, for every $h\in[n]$,
	\[
	\sum_{\substack{\vecu\in\Z_q^n\\ \|\vecu\|_0=h}}\widehat{f\cdot g}(\vecu) = \sum_{\vecu\in\Z_q^n}\sum_{\substack{\vecu'\in\Z_q^n\\ \|\vecu+\vecu'\|_0=h}}\widehat{f}(\vecu)\cdot\widehat{g}(\vecu')\,.
	\]
\end{lemma}
\begin{proof}
	For every $\vecu\in\Z_q^n$, we have
	\begin{align*}
		\widehat{f\cdot g}(\vecu) &= \frac{1}{q^n}\sum_{\veca\in\Z_q^n}f(\veca)\cdot g(\veca)\cdot\overline{\omega^{\vecu^\top\veca}} \\
		&= \frac{1}{q^n}\sum_{\veca\in\Z_q^n} \left(\sum_{\vecu'\in\Z_q^n}\widehat{f}(\vecu')\cdot\omega^{\vecu'^\top\veca}\right)\cdot g(\veca)\cdot\overline{\omega^{\vecu^\top\veca}} \\
		&= \sum_{\vecu'\in\Z_q^n}\widehat{f}(\vecu')\cdot \left(\frac{1}{q^n}\sum_{\veca\in\Z_q^n}g(\veca)\cdot \overline{\omega^{(\vecu-\vecu')^\top\veca}}\right) \\
		&= \sum_{\vecu'\in\Z_q^n}\widehat{f}(\vecu')\cdot\widehat{g}(\vecu-\vecu')\, .
	\end{align*}
	Next, for every $h\in[n]$,
	\begin{align}
		\sum_{\substack{\vecu\in\Z_q^n\\ \|\vecu\|_0=h}}\widehat{f\cdot g}(\vecu) &= \sum_{\substack{\vecu\in\Z_q^n\\ \|\vecu\|_0=h}}\sum_{\vecu'\in\Z_q^n}\widehat{f}(\vecu')\cdot\widehat{g}(\vecu-\vecu') \, .
        \label{eq:FourierProduct}
    \end{align}
    Letting $\vecw=\vecu-\vecu'$ and switching the order of the summations, the right-hand side of \eqref{eq:FourierProduct} becomes
    \begin{align*}
		\sum_{\vecu'\in\Z_q^n}\sum_{\substack{\vecw\in\Z_q^n\\ \|\vecw+\vecu'\|_0=h}}\widehat{f}(\vecu')\cdot\widehat{g}(\vecw) \, ,
	\end{align*}
	which, after renaming variables, proves the furthermore part of the lemma.
\end{proof}

{
Next we state a hypercontractivity statement from \cite{o2014analysis}. Let $(\Omega,\pi)$ be a finite probability space with $|\Omega|\ge 2$ and assume $\pi$ has full support. We denote by $L^2(\Omega,\pi)$ the inner product space of square-integrable functions $\Omega\to \mathbb{R}$ with inner product $\langle f,g\rangle = \mathbb{E}_{x\sim \pi}[f(x) g(x)]$.
 
\begin{lemma}\cite[Chapter 10, General Hypercontractivity Theorem, page 283]{o2014analysis}\label{lem:odonnell}~\\
Let $(\Omega_1,\pi_1),\dots,(\Omega_n,\pi_n)$ be finite probability spaces, in each of which every outcome has probability at least $\lambda$. Let $f\in L^2(\Omega_1\times \cdots \times \Omega_n,\pi_1\otimes\cdots\otimes\pi_n)$. Then for any $p'>2$, and $0\le \rho \le \frac{1}{\sqrt{p'-1}}\lambda^{1/2-1/{p'}}$, \[\Vert T_\rho f\Vert_2 \le \Vert f\Vert_{p}\, , \]where $p$ is the Hölder's conjugate of $p'$, and $T_\rho$ is the noise operator defined by \[T_\rho f(\vecx)=\sum_{\vecu\in\Z_q^n}\widehat{f}(\vecu)\rho^{\|\vecu\|_0}\omega^{\vecu^\top\vecx}\, .\]
\end{lemma}

We now state the following consequence in our language:

\begin{lemma}\label{lem:hypercontractivity}
Let $f:\Z_q^n\to\mathbb{R} \in L^2(\Z_q^n,\mathsf{Unif}(\Z_q^n))$. Then for any $1<p<2$, and  $0\le \rho \le {\sqrt{p-1}}\cdot (1/q)^{1/p-1/2}$,    \[\Vert T_\rho f\Vert_2 \le \Vert f\Vert_{p}\, .\]
\end{lemma}

\begin{proof}
The lemma follows from \cref{lem:odonnell} by letting $\Omega_i = \mathbb{Z}_q$ and $\pi_i$ be the uniform distribution on $\Omega_i$ yielding $\lambda=1/q$ and substituting 
$p' = p/(p-1)$.
\end{proof}
}

Next, we prove the following consequence of the hypercontractivity theorem.

\begin{lemma}\label{lem:hyper}
	For every $q\in \N$, there exists $\zeta_q$ such that for every $f:\Z_q^n\to[-1,1]$ and $B=\{\veca\in\Z_q^n\, |\, f(\veca)\neq0\}$, the following holds: If $|B|\geq q^{n-b}$ for some $b\in\N$, then for every $\vecv\in\Z_q^n$ and every $h\in\{1,\dots,4b\}$, we have
	\[
	\frac{q^{2n}}{|B|^2}\sum_{\substack{\vecu\in\Z_q^n\\\|\vecu+\vecv\|_0=h}}|\widehat{f}(\vecu)|^2\leq\left(\frac{\zeta_q\cdot b}{h}\right)^h \, .
	\]
\end{lemma}

\begin{proof}

    We will prove the lemma for $\zeta_q = 6 q^{4/3}$.
	Let $\vecv = 0^n$ and $f:\Z_q^n\to[-1,1]$. We choose $p = 1 + \frac{h}{6 b}$ and $\rho = {\sqrt{p-1}}\cdot (1/q)^{1/p-1/2}$. Assume $|B|\geq q^{n-b}$.
	
	The choices of $p$ and $\rho$ satisfy the preconditions of \cref{lem:hypercontractivity}, and so applying \cref{lem:hypercontractivity} we have
	\[
	\sum_{\vecu\in\Z_q^n}\rho^{2\|\vecu\|_0}|\widehat{f}(\vecu)|^2 = \|T_\rho f\|_2^2  \leq \|f\|_p^2 = \left(\frac{1}{q^n}\sum_{\vecx\in\Z_q^n}|f(\vecx)|^{p}\right)^{2/p} \leq \left(\frac{|B|}{q^n}\right)^{2/p},
	\]
    where the last inequality uses the fact that $f(x) \in [-1,1]$ for all $x$.
    
	Now, suppose $h\in \{1,\dots, 4b\}$. Noting that $\rho^{2h}\sum_{\substack{\vecu\in\Z_q^n\\\|\vecu\|_0=h}}|\widehat{f}(\vecu)|^2 \leq \sum_{\vecu\in\Z_q^n} \rho^{2\|\vecu\|_0} |\widehat{f}(\vecu)|^2$, we have 
	\begin{align}
		\frac{q^{2n}}{|B|^2}\sum_{\substack{\vecu\in\Z_q^n\\\|\vecu\|_0=h}}|\widehat{f}(\vecu)|^2 &\leq \frac{1}{\rho^{2h}}\left(\frac{q^n}{|B|}\right)^{2-2/p} \nonumber \\
		&\leq \frac{1}{\rho^{2h}} q^{(2-2/p)b} \nonumber\\
		&= \frac{q^{\left(-1+\frac{2b}{h}+\frac{2}{p}-\frac{2b}{hp}\right)h}}{(p-1)^h} \nonumber\\
		&= \left(\frac{6 b}{h} \cdot  q^{-1+\frac{2b}{h}+\frac{2}{p}-\frac{2b}{hp}}\right)^h, \label{eq:qbratio}
	\end{align}
	where the first equality above is by our choice of $\rho$ and the second by our choice of $p$.

	Observe that the exponent of $q$ in the final expression above can be bounded as follows:
	\begin{align}
		-1 + \frac{2b}{h}+\frac{2}{p}-\frac{2b}{ph} &= -1 + \frac{2b}{h} + \frac{2\left(1-\frac{b}{h}\right)}{1+\frac{h}{6 b}} \nonumber\\
		&= \left(1+\frac{h}{6 b}\right)^{-1} \left(\frac{2}{6} - \frac{h}{6 b} + 1\right) \nonumber\\
		& \le 4/3\, .\label{eq:qexp}
	\end{align}


    The expression from \eqref{eq:qbratio} can now be bounded from above by $\left(\frac{\zeta_q b}{h}\right)^h$, where $\zeta_q = 6 q^{4/3}$, implying that
	\[
	\frac{q^{2n}}{|B|^2}\sum_{\substack{\vecu\in\Z_q^n\\\|\vecu\|_0=h}}|\widehat{f}(\vecu)|^2 \leq \left(\frac{\zeta_q b}{h}\right)^h\, .
	\]
    
	In order to extend the above to sums over translational shifts, i.e., $\vecu$ such that $\|\vecu+\vecv\|_0 = h$ for an arbitrary $\vecv\in\Z_q^n$, consider the function $g(\vecx)=f(\vecx)\cdot{\omega^{\vecx^\top\vecv}}$. We have for every $\vecx\in\Z_q^n$,
	\[
	\widehat{g}(\vecu)=q^{-n}\sum_{\veca\in\Z_q^n}g(\veca)\overline{\omega^{\veca^\top\vecu}}=q^{-n}\sum_{\veca\in\Z_q^n}f(\veca)\overline{\omega^{\veca^\top(\vecu-\vecv)}}=\widehat{f}(\vecu-\vecv) \, .
	\]
	By applying the above analysis on $g$, we have
	\[
	\frac{q^{2n}}{|B|^2}\sum_{\substack{\vecu\in\Z_q^n\\\|\vecu+\vecv\|_0=h}}|\widehat{f}(\vecu)|^2 = \frac{q^{2n}}{|B|^2}\sum_{\substack{\vecu\in\Z_q^n\\\|\vecu\|_0=h}}|\widehat{g}(\vecu)|^2 \leq\left(\frac{\zeta_q b}{h}\right)^h,
	\]
	as desired.
\end{proof}

\section{Communication problems} \label{sec:comm}

Throughout this paper, we will be dealing with $k$-hypermatchings on vertices from the set $[n]$, i.e., a set of edges $e_1,\ldots,e_m$ where $e_i \subseteq [n]$, $|e_i| = k$ and $e_i \cap e_j = \emptyset$ for every $i \ne j \in [m]$. We let $e_i = \{(e_i)_1,\ldots,(e_i)_k\}$. The direct encoding of a matching $M = \{e_1,\ldots,e_m\}$ will be given by a \emph{hypermatching matrix} $A \in \{0,1\}^{km \times n}$ where $A_{k(i-1)+\ell,j} = 1$ if and only if $j = (e_i)_\ell$, for $i\in [m], \ell\in [k]$. (Thus, $A$ is a matrix with row sums being $1$ and column sums being at most $1$. Note that $A$ also depends on the ordering of $e_1, e_2, \dots, e_m$ as well as the ordering of the nodes within each $e_i$.) 

We will also find it convenient to refer to edges by their indicator vectors in $\Z_q^n$. For an edge $e_i$, we will use the boldface notation $\vece_i \in \Z_q^n$ to refer to this vector, i.e., $(\vece_i)_j = 1$ if $j = (e_i)_\ell$ for some $\ell \in [k]$, while $(\vece_i)_j = 0$ otherwise.

We are now ready to define the communication game, which we term the Implicit Randomized Mask Detection (IRMD) problem:

\begin{definition}[Implicit Randomized Mask Detection (IRMD) Problem]\label{def:IRMD}
	
	Let $q,k,n,T\in\mathbb{N}$ and $\alpha\in(0,1/k)$ be parameters. Let $\cD_Y$ and $\cD_N$ be distributions over $\Z_q^k$.
	In the $(\cD_Y,\cD_N)$-$\IRMD_{\alpha,T}$ game, there are $T$ players and a hidden $q$-coloring encoded by a random $\vecx^*\in\Z_q^n$. The $t$-th player has two inputs: (a.) $A_t\in\{0,1\}^{\alpha kn\times n}$, the hypermatching matrix (see above) corresponding to a random hypermatching $M_t$ of size $\alpha n$ and (b.) a vector $\vecz_t\in\Z_q^{\alpha kn}$ that can be generated from one of two different distributions:
	\begin{itemize}
		\item (Yes) $\vecz_t=A_t \vecx^*+\vecb_t$ where $\vecb_t\in\Z_q^{\alpha kn}$ is of the form $\vecb_t=(\vecb_{t,1},\dots,\vecb_{t,\alpha n})$ and each $\vecb_{t,i}\in\Z_q^k$ is sampled from $\cD_Y$.
		\item (No) $\vecz_t=A_t \vecx^*+\vecb_t$ where $\vecb_t\in\Z_q^{\alpha kn}$ is of the form $\vecb_t=(\vecb_{t,1},\dots,\vecb_{t,\alpha n})$ and each $\vecb_{t,i}\in\Z_q^k$ is sampled from $\cD_N$.
	\end{itemize}

    This is a one-way game where the $t$-th player broadcasts a message to all other players after receiving messages from players $1,\ldots,t-1$. 
    The goal is for the $T$-th player to be able to decide whether the $\{\vecz_t\}$ have been chosen from the ``Yes'' distribution or ``No'' distribution. The advantage of a protocol (in which the $T$-th player outputs either ``Yes'' or ``No'') is defined as $|\Pr_{\cD_Y}[\text{the }T\text{-th player outputs Yes}]-\Pr_{\cD_N}[\text{the }T\text{-th player outputs Yes}]|$.
\end{definition}

\begin{remark}
	We remark that the inputs to the $T$ players in the \IRMD\ problem can be viewed as a stream $\vecsigma = \vecsigma^{(1)} \circ \cdots \circ \vecsigma^{(T)}$, where the $t$-th player's input $(A_t,\vecz_t)$ is converted to a stream $\vecsigma^{(t)}=(\sigma^{(t)}(i) | i \in [\alpha n])$ where the elements of the stream are of the form $\sigma^{(t)}(i) = (\vecj^{(t)}(i),\vecz^{(t)}(i))$ with $\vecj^{(t)}(i) \in [n]^k$ is a sequence of $k$ distinct elements of $[n]$ and $\vecz^{(t)}(i) \in \Z_q^k$. This ``streaming'' representation will be used when we relate the complexity of \IRMD\ to the approximability of various $\maxF$ problems in \cref{thm:main-technical}. 
\end{remark}

We suppress the subscripts $\alpha$ and $T$ when they are clear from context. Furthermore, 
we simply use $\IRMD$ to refer to $(\cD_Y,\cD_N)$-$\IRMD$ with $\cD_Y$ being the uniform distribution over $\{0^k,1^k,\dots,(q-1)^k\}$ and $\cD_N$ being the uniform distribution over $\Z_q^k$. 
The following theorem shows that in this special case, the $\IRMD$ problem requires linear communication. 
We remark that the theorem could hold for other pairs of distributions and leave the question of when such a lower bound holds as an interesting open problem.

\begin{theorem}[Linear lower bound for $\IRMD$]\label{thm:main IRMD}
	For every $q,k\in\mathbb{N}$ and $\delta\in(0,1/2)$, there exists $\alpha_0 \in (0,1/k)$ such that for every $\alpha\in(0,\alpha_0]$ and $T\in\mathbb{N}$, there exists $n_0 \in \mathbb{N}$ and  $\tau\in(0,1)$ such that the following holds. If $\cD_Y$ is the uniform distribution over $\{0^k,1^k,\dots,(q-1)^k\}$, $\cD_N$ is the uniform distribution over $\Z_q^k$, and $n \geq n_0$ then every protocol for $(\cD_Y,\cD_N)$-$\IRMD_{\alpha,T}$ with advantage $\delta$ requires $\tau n$ bits of communication.
\end{theorem}

\cref{thm:main IRMD} is proved at the end of this section. Its proof uses \cref{thm:main IFRMD} and \cref{lem:reduction IFRMD to IRMD} which we state below.

We prove the hardness of $\IRMD$ by showing the hardness of a \emph{folded} version of $\IRMD$ defined below.
In the folded version of the communication problem, we augment each hyperedge with an associated \emph{center} $c \in e$. Given a $k$-hypermatching $M = (e_1,\ldots,e_m)$ and a sequence of centers $\vecc = (c_1,\ldots,c_m)$ with $c_i \in e_i = \{(e_i)_1,\ldots,(e_i)_k\}\subseteq [n]$, the $\vecc$-centered folded encoding of $M$ is the matrix $A_{\vecc}\in \Z_q^{(k-1)m \times n}$ whose columns are indexed by the vertex set $[n]$ and rows are indexed by pairs $(i,\ell)$ with $i \in [m]$ and $\ell \in [k]\setminus \{j\}$ where $j \in [k]$ is the index of the center, i.e., $(e_i)_j = c_i$ with entries of $A_{\vecc}$ given by
\begin{equation}
(A_\vecc)((i,\ell),u)  = \left\{\begin{array}{ll}
	1     &  \text{ if }u = (e_i)_\ell \\
	-1     & \text{ if }u=c_i \\
	0       & \text{ otherwise.}
\label{eqn:Ac-defn}\end{array}
\right.
\end{equation}

See \cref{fig:folded matrices p} for an example. We define the folded version of the $\IRMD$ problem below (note that all the arithmetic is over $\Z_q$):

\begin{definition}[Implicit Folded Randomized Mask Detection (IFRMD) Problem]\label{def:IFRMD}
	Let $q,k,n,T\in\mathbb{N}$ and $\alpha\in(0,1/k)$ be parameters.
	In the \IFRMD\ game, there are $T$ players and a hidden $q$-coloring encoded by a random $\vecx^*\in\Z_q^n$. The $t$-th player has a pair of inputs $(A_{t,\vecc_t},\vecw_t)$ given as follows. $A_{t,\vecc_t}\in\Z_q^{\alpha(k-1)n\times n}$ gives a $\vecc_t$-centered folded encoding of a hypermatching $M_t$ of size $\alpha n$ where $M_t$ is chosen uniformly at random and $\vecc_t$ is chosen uniformly from all possible centers for $M_t$. And $\vecw_t\in\Z_q^{\alpha(k-1)n}$ is a vector that can be generated from two different distributions:
	\begin{itemize}
		\item (\yes) $\vecw_t=A_{t,\vecc_t} \vecx^*$.
		\item (\no) $\vecw_t$ is uniform over $\Z_q^{\alpha(k-1)n}$.
	\end{itemize}

    This is a one-way game where the $t$-th player broadcasts a message to all other players after receiving messages from players $1,\ldots,t-1$. 
    The goal is to decide (by the $T$-th player) whether the $\{\vecw_t\}$ are coming from the \yes\ distribution or the \no\ distribution. The advantage of a protocol is defined as $$\left|\Pr_{(A_{t,c_t},\vecw_t)_{t\in T} \sim \yes}[\text{the }T\text{-th player outputs Yes}]-\Pr_{(A_{t,c_t},\vecw_t)_{t\in T} \sim \no}[\text{the }T\text{-th player outputs Yes}]\right|.$$
\end{definition}

The main technical theorem of this paper is the following $\Omega(n)$ communication lower bound for $\IFRMD$.

\begin{theorem}[Linear lower bound for $\IFRMD$]\label{thm:main IFRMD}
	For every $q,k\in\mathbb{N}$ and $\delta\in(0,1/2)$, there exists $\alpha_0>0$ such that for every $\alpha\in(0,\alpha_0]$ and $T\in\mathbb{N}$, there exists $n_0\in\N$ and $\tau\in(0,1)$ such that the following holds. When $n\geq n_0$, any protocol for $\IFRMD$ with advantage $\delta$ requires $\tau n$ bits of communication.
\end{theorem}

An instance of the $\IFRMD$ problem can be viewed as giving $\alpha n T$ constraints  on $n$ variables $X_1,\ldots,X_n$ where each constraint is of the form $(c; i_i,\ldots,i_{k-1}; v_1,\ldots,v_{k-1})$ with $c,i_1,\ldots i_{k-1} \in [n]$ and $v_1,\ldots,v_{k-1} \in \Z_q$ with the constraint requiring 
$\wedge_{j=1}^{k-1} (X_{i_j} - X_c = v_j)$. Thus each instance of $\IFRMD$ specifies an instance of the aforementioned \textsf{Max $k$-LIN}-$\bmod\; q$ problem where $\cD_Y$ is supported on instances which are always satisfiable (by setting $X =\vecx^*$). It turns out $\cD_N$ is supported on roughly random instances and thus it is unlikely to have a solution satisfying more than $q^{-(k-1)}$ fraction of the constraints. (This is implicit in the proof of \cref{thm:main-technical}.) The indistinguishability result in \cref{thm:main IFRMD} thus effectively implies a $q^{-(k-1)}+\epsilon$-inapproximability for this problem. This is formally proved in Example 4 in \cref{ssec:examples}.

The proof of~\cref{thm:main IFRMD} is given in \cref{ssec:proof-IFRMD-thm}. We now establish a reduction from $\IFRMD$ to $\IRMD$ that preserves the communication complexity. By this reduction,~\autoref{thm:main IRMD} will be an immediate corollary of~\autoref{thm:main IFRMD}.

\begin{lemma}\label{lem:reduction IFRMD to IRMD}
	Let $n,k,\alpha$ be the parameters. Suppose there exists a protocol for $\IRMD$ using at most $s$ bits of communication with advantage $\delta$, then there exists a protocol for $\IFRMD$ using at most $s$ bits communication with advantage $\delta$.
\end{lemma}

\begin{proof}

	Suppose we have an instance of $\IFRMD$ with input $(A_{t,\vecc_t},\vecw_t)$ to the $t$-th player. We need to transform this to an input $(\Pi_t,\vecz_t)$ to the IRMD problem (while respecting the right distributions). (Furthermore the transformation $(A_{t,\vecc_t},\vecw_t) \mapsto (\Pi_t,\vecz_t)$ should be locally computable by the $t$th player.)
    
    Let $m = \alpha n$.  Let $e^{(t)}_1, e^{(t)}_2, \dots, e^{(t)}_m$ be the hyperedges corresponding to $A_{t,\vecc_t}$. For $i \in [m]$ let us write $e^{(t)}_i = \{(e^{(t)}_i)_1,\ldots, (e^{(t)}_i)_k\} \subseteq [n]$.\footnote{Note that the choice of ordering of vertices within an edge is arbitrary. Altering this will only (simultaneously) permute the rows of $\Pi_t$ and $\vecz_t$.} Further let $j(i)\in [k]$ be the unique index so that $(e_i^{(t)})_{j(i)} = c_{t,i}$.

    For each $t$, the $t$-th player performs the following computations on his/her input:
	\begin{enumerate}
		\item We index the columns of $\Pi_t$  by the vertex set $[n]$ and the rows  by $[m]\times  [k]$. We set $\Pi_t((i,\ell),u) = 1$ if $(e^{(t)}_i)_\ell = u$ and $0$ otherwise. 
		\item For each $i\in[m]$, sample $a_{t,i}\in\Z_q$ uniformly at random.  Again we assume the coordinates of $\vecz_t \in \Z_q^{km}$ are indexed by pairs $(i,\ell) \in [m] \times [k]$. We set $(\vecz_t)(i,\ell) = a_{t,i}$ if $\ell = j(i)$ and $(\vecz_t)(i,\ell) = (\vecw_t)(i,\ell) + a_{t,i}$ otherwise. 
	\end{enumerate}
	We claim that the inputs $(A_t, \vecz_t)$ correspond to an instance of $\IRMD$. It suffices to show that if $(\{(A_{t,\vecc_t},\vecw_t)\}_{t\in[T]},\vecx^*)$ follows the \yes{} (resp. \no{}) distribution of $\IFRMD$, then $(\{(A_t,\vecz_t)\}_{t\in[T]},\vecx^*)$ follows the \yes{} (resp. \no{}) distribution of $\IRMD$.	
	The \no{} case is easy to see: $\Pi_t$ encodes a random $k$-hypermatching of size $m$ and $\vecz_t$ is uniform over $\Z_q^{km}$ since $\vecw_t \in \Z_q^{(k-1)m}$ and $a_{i,t}\in \Z_q$ are uniform and independent of each other and of $A_{t,\vecc_t}$.

	We now turn to the \yes~case. Fix $i \in [m]$. For $\ell = j(i)$, we have
	\begin{align*}
		(\vecz_t)(i,\ell) = a_{t,i} = \vecx^*_{(\vecc_t)_i} + (-\vecx^*_{(\vecc_t)_i} + a_{t,i}) = \vecx^*_{(e^{(t)}_i)_\ell} + (-\vecx^*_{(\vecc_t)_i} + a_{t,i}).
	\end{align*}
    For $\ell \in [k]\setminus \{j(i)\}$, we have 
    \begin{align*}
		(\vecz_t)(i,\ell) &= (\vecw_t)(i,\ell) + a_{t,i} = \vecx^*_{(e_i^{(t)})_\ell} - \vecx^*_{(\vecc_t)_i)} + a_{t,i} = \vecx^*_{(e_i^{(t)})_\ell} + (-\vecx^*_{(\vecc_t)_i} + a_{t,i}),
	\end{align*}
    where the second equality uses $\vecw = A_{t,\vecc_t}\cdot \vecx^*$ in the \yes\ case.

	Thus, it follows that $\vecz_t = \Pi_t \vecx^* + \vecb_t$, where $\vecb_t = (\vecb_{t,1}, \dots, \vecb_{t,m})$ is given by $\vecb_{t,i} = (-\vecx^*_{c_{t,i}} + a_{t,i})\cdot\mathbf{1}_k$ where $\mathbf{1}_k$ is the all $1$ vector of length $k$. Thus for every $t,i$, $\vecb_{t,i} \in \Z_q^k$ is a uniformly chosen constant vector independent of $\vecx^*$ (and of other $\vecb_{t,i'}$ for $i'\ne i$) as required in the \yes\ case and thus showing that $(\Pi_t,\vecz_t)$ are distributed according to $\cD_Y$. 

\end{proof}

\begin{proof}[Proof of~\cref{thm:main IRMD} (assuming~\cref{thm:main IFRMD})]
	For the sake of contradiction, suppose there exists a protocol for IRMD with advantage $\delta$ using fewer than $\tau n$ bits of communication. Then by~\autoref{lem:reduction IFRMD to IRMD} there exists a protocol for IFRMD with advantage $\delta$ using fewer than $\tau n$ bits of communication, which contradicts \autoref{thm:main IFRMD}. This completes the proof of~\autoref{thm:main IRMD}.
\end{proof}

In the following section we show how \cref{thm:main IRMD} yields the claimed hardness of streaming problems. In the rest of this paper, we focus on the proof of~\autoref{thm:main IFRMD}, i.e., the linear communication lower bound for IFRMD.

\section{Streaming problems and hardness} \label{sec:streaming}

In this section we state our main technical theorem establishing linear space lower bounds for the approximability of many CSPs. We also prove these lower bounds assuming \cref{thm:main IFRMD} and in particular its corollary \cref{thm:main IRMD}.

Below we define the two crucial constants associated with a family $\cF$ which lay out the ``trivial'' approximability, and the inapproximability that we prove. In particular we define the notion of a width
$\omega(\cF) \in [1/q,1]$ for every family $\cF$. The notion of a wide family from \cref{thm:intro-approx-res} corresponds to a family with maximum width, i.e., $\omega(\cF)=1$.

\begin{definition}[Minimum value, Width of $\cF$]
	For a family $\cF$, we define its minimum value $\rho_{\min}(\cF)$ to be the infimum over all instances $\Psi$ of $\maxF$ of $\val_\Psi$. For $\vecb \in \Z_q^k$ and $f:\Z_q^k \to \{0,1\}$ we define $\vecb$-width of $f$, denoted $\omega_{\vecb}(f)$ to be the quantity $\frac{|\{a\in\Z_q\, |\, f(\vecb + a^k)=1\}|}q$. The {\em width} of $f$, denoted $\omega(f)$, is given by $\omega(f) = \max_{\vecb \in \Z_q^k} \{\omega_{\vecb}(f)\}$. Finally for a family $\cF$, we define its {\em width} to be $\omega(\cF) =\min_{f\in\cF} \{\omega(f)\}$. 
    We say that a family $\cF$ is {\em wide} if $\omega(\cF) = 1$. 
\end{definition}

As described above $\rho_{\min}(\cF)$ may not even be computable given $\cF$, but as pointed out in \cite{CGSV21-conference} it is a computable function. Key to this assertion is the following equivalent definition of $\rho_{\min}(\cF)$ which follows from Definition 2.4 and Proposition 2.5 of \cite{CGSV21-conference}. 

\begin{proposition}[\protect{\cite[Proposition 2.5]{CGSV21-conference}}]\label{prop:rho-equiv}
	For every $k,q, \cF \subseteq \{f:\Z_q^k \to \{0,1\}\}$ we have
	$$\rho_{\min}(\cF) = \rho(\cF) := \min_{\cD_{\cF} \in \Delta(\cF)} \left\{ \max_{\cD \in \Delta([q])} \left\{ \Exp_{f \sim \cD_{\cF},  \veca \sim \cD^k} [ f(\veca) ]         \right\} \right\}.$$
\end{proposition}

We are now ready to prove the main theorem of the paper on the approximability of CSPs by applying \cref{thm:main IRMD}.

\begin{theorem}[Linear Space Inapproximability of CSPs]\label{thm:main-technical} For every $k,q, \cF \subseteq \{f:\Z_q^k \to \{0,1\}\}$ and every $\epsilon > 0$ we have the following: Every randomized single-pass streaming $\left(1+\epsilon\right)\cdot \frac{\rho(\cF)}{\omega(\cF)}$-approximation algorithm for $\maxF$ requires $\Omega(n)$ space.
\end{theorem}

\begin{proof}
	We assume $0 < \epsilon \leq 1/10$, since the theorem only gets weaker for larger $\epsilon$.  Given $k$, $q$, $\cF$ we let $\alpha = \min\{\alpha_0,\epsilon/(100k^2q)\}$ where $\alpha_0$ is the constant from \cref{thm:main IRMD} with $\delta=1/6$. We now set $T$ to be some large enough constant that only depends on $q,k,\cF,\epsilon,\alpha$ (but not $n$). 
    
	Let $\ALG$ be a space $s$ algorithm distinguishing instances from the set $\{\Psi ~|~ \val_\Psi \geq (1-\epsilon/3)\omega(\cF)\}$ from instances from the set $\{\Psi ~|~ \val_\Psi \leq (1+\epsilon/3)\rho(\cF)\}$ with success probability at least $2/3$. We show how to use $\ALG$ to devise an $s$-bit
	communication protocol for $\IRMD = \IRMD_{\alpha,T}$ with advantage at least $1/6$.

	For $f\in \cF$, let $\vecb_f \in \Z_q^k$ be a sequence maximizing $\omega_{\vecb_f}(f)$ and let $S_f = \{\vecb_f + a^k\, |\, a \in \Z_q\}$. Further let $\cD_\cF \in \Delta(\cF)$ be a distribution achieving the minimum in the equivalent definition of $\rho(\cF)$ from \cref{prop:rho-equiv}. 
	Let $\vecsigma = (\sigma_1,\ldots,\sigma_m)$ be an instance of $\IRMD$ with $T$ players, so that $m = T \alpha n$ and $\sigma_i = (\vecj(i), \vecz(i))$ where $\vecj(i) \in [n]^k$ is a sequence of $k$ distinct elements of $[n]$ and $\vecz(i) \in \Z_q^k$. For each $\sigma_i$ we either generate $0$ or $1$ constraint of $\maxF$ as follows: We sample $f(i)\sim \cD_\cF$ and output the constraint $(f(i),\vecj(i))$ if $\vecz(i) \in S_{f(i)}$ and output no constraint otherwise.
	Applying this step independently to each $\sigma_i$ generates an instance $\Psi$ of $\maxF$ with $\tilde{m} \leq m$ constraints on $n$ variables. We make the following claims about $\Psi$.
	
	\begin{enumerate}[label=(\arabic*)]
		\item $\Pr_{\yes}[\tilde{m}>(1+\epsilon/10)\cdot q^{-(k-1)}\cdot m]=o(1)$ and $\Pr_{\no}[\tilde{m}<(1-\epsilon/10)\cdot q^{-(k-1)}\cdot m]=o(1)$, i.e., the number of constraints $\tilde{m}$ does not deviate (in the wrong direction) from its expectation $q^{-(k-1)}\cdot m$ with too high a  probability.\footnote{In these claims the $o(1)$ term goes to zero as $n \to \infty$. In fact, the proof will show that these terms go to zero exponentially fast in $n$ but we won't need this additional fact.}
		\item If $\vecsigma$ is generated from the $\yes$ distribution with hidden vector $\vecx^*$ then with high probability the number of constraints of $\Psi$ satisfied by $\vecx^*$ is  at least $\omega(\cF)(1-\varepsilon/10) \cdot q^{-(k-1)}\cdot m$. 
		In particular, 

		$\Pr_{\yes}[\val_\Psi\leq(1-\epsilon/3)\cdot\omega(\cF)]=o(1)$.
		\item If $\vecsigma$ is generated from the $\no$ distribution with hidden vector $\vecx^*$ then with high probability for every $\vecnu$ the number of constraints of $\Psi$ satisfied by $\vecnu$ is  at most $\rho(\cF)(1+\varepsilon/10) \cdot q^{-(k-1)}\cdot m$. In particular, 
		$\Pr_{\no}[\val_\Psi\geq(1+\epsilon/3)\cdot\rho(\cF)]=o(1)$.
	\end{enumerate}

	With the above claims in hand, it is straightforward to convert $\ALG$ into an $O(s)$-bit communication protocol for $\IRMD$ with advantage at least $1/6$ --- the $t$-th player gets the state of $\ALG$ after processing constraints corresponding to the first $t-1$ blocks from the $(t-1)$-th player; generates the constraints corresponding to the $t$-th block of the stream $\vecsigma$, and simulates $\ALG$ on this part of the stream corresponding to $\Psi$, and passes the resulting state on to the $(t+1)$-th player.
	The $T$-th player outputs $1$ if $\ALG$ outputs $1$ and $0$ otherwise. It is straightforward to see that if $\ALG$ is correct on every input with probability $2/3$ and Claims (1)-(3) above hold, then the resulting communication protocol achieves advantage at least $1/3 - o(1) \geq 1/6$ on $\IRMD$. Finally, we invoke~\autoref{thm:main IRMD} and conclude that $s=\Omega(n)$.
	
	We thus turn to proving claims (1)-(3). Given $\sigma_1,\ldots,\sigma_m$ and $\vecnu \in \Z_q^n$, we create a collection of related variables as follows: For $i \in [m]$, let $V_i = 1$ if $\sigma_i$ results in a constraint and $0$ otherwise.  Further, let  $Y_i(\vecnu)=1$ if $V_i=1$ and the resulting constraint is satisfied by the assignment $\vecnu$. (Note all these are random variables depending on $\vecsigma$). 
	Below, we bound the expectations of the sums of these random variables in the \yes\ and \no\ cases, and also argue that these variables are close to their expectations (or at least give bounds on deviating from the expectation in one direction). This will suffice to prove claims (1)-(3) and thus the theorem.

	\paragraph{Proof of Claim (1).}
	We start with $\tilde{m} = \sum_{i=1}^m V_i$
	in the \no\ case: In this case $\Exp[V_i] = |S_f|/q^k = q^{-(k-1)}$ (note that $|S_f| = q$ for every $f$). Furthermore the $V_i$'s are independent since $\vecz(i)$'s are uniform and independent of each other. Thus $\tilde{m}$ is sharply concentrated around $q^{-(k-1)}\cdot m$ and we get that $\Pr_{\no}[\tilde{m} \not\in (1 \pm \epsilon/10)\cdot q^{-(k-1)}\cdot m] = o(1)$. 
	
	Turning to the \yes\ case, since $\vecz(i)$'s are no longer independent, the $V_i$'s are correlated. To enable the analysis, we define a vector $\vecx^*$ to be {\em $\gamma$-good} for $\gamma > 0$ if for every $\tau \in \Z_q$ we have $\Pr_{i\in[n]}[\vecx^*_i = \tau] \in (1\pm \gamma)(1/q)$. Note that for every constant $\gamma>0$, the probability that $\vecx^*$ is not $\gamma$-good is $o(1)$. Fix $\vecx^*$ that is $\gamma$-good. We claim that in this case,  $\Exp[V_i\, |\, V_{1:i-1}] \leq q^{-(k-1)} \cdot (1+\gamma+\alpha qk)^k$. To see this note that the effect of conditioning on $V_{1:i-1}$ only affects $V_i$ due to the fact that now
	$\vecj(i)$ is chosen from a smaller set of variables and not all of $[n]$. Let $t\in[T]$ denote the block containing $i$ (i.e., $i \in ((t-1)\alpha n,t \alpha n]$).  Let $S$ denote the set of variables that do not participate in the edges $\vecj((t-1)\alpha n + 1),\ldots,\vecj(i-1)$. Note $|S| \geq (1 - k\alpha)n$ and so for every $\tau \in \Z_q$ we have $\Pr_{\ell \in S}[\vecx^*_\ell = \tau] \leq (1+\gamma+\alpha k q)/q$. We conclude that the probability $\Pr[\vecx^*|_{\vecj(i)} \in S_f\, |\, V_{1:i-1}] \leq |S_f|\cdot ((1+\gamma+\alpha k q)/q)^k = q^{-(k-1)} \cdot (1+\gamma+\alpha qk)^k$. Setting $\gamma = \epsilon/(100k)$ and 
	using $\alpha \leq \epsilon/(100k^2q)$, we conclude $\Exp[V_i\, |\, V_{1:i-1}] \leq q^{-(k-1)} \cdot (1+\epsilon/(50k))^k \leq q^{-(k-1)} \cdot (1+\epsilon/20)$ (where we use $\epsilon \leq 1/10$ to get $(1+\epsilon/(50k))^k (1+\epsilon/20)$).  Applying \cref{lem:our-azuma} we conclude that here again we get that $\Pr_{\yes}[\tilde{m} = \sum_i V_i > (1+\epsilon/10) q^{-(k-1)}m] = o(1)$. (Note that the $o(1)$ term goes to zero exponentially fast with $m$.)

	
	\paragraph{Proof of Claim (2).}
	Now we analyze the number of satisfiable constraints of the resulting instance $\Psi$ in the \yes\ case, where we argue that $\vecx^*$ satisfies a large fraction of constraints with high probability. Again with probability $1-o(1)$ we have that $\vecx^*$ is $\gamma$-good. Now an argument similar to the one in the analysis of $X$ in the \yes\ case shows that for every $\vecb\in\Z_q^k$, $\Pr[\vecx^*|_{\vecj(i)} = \vecb\, |\, Y_{1:i-1}] \geq (1-\epsilon/50)\cdot q^{-k}$.  Fix $f(i)$ and let $W = S_{f(i)} \cap f(i)^{-1}(1)$. Note by definition of $\omega(\cF)$ that $|W| \geq \omega(\cF)\cdot q$. 
	The event that the $i$-th constraint is satisfied by $\vecx^*$ is equivalent to the event that $\vecx^*_{\vecj(i)} \in T$ and the probability of this event, conditioned on $Y_{1:i-1}$ is at least $|W| \cdot (1-\epsilon/50)\cdot q^{-k} \geq (1-\epsilon/50) \cdot \omega(\cF) \cdot q^{-(k-1)}$. Using \cref{lem:our-azuma} we conclude again that 
	$\Pr[Y(\vecx^*) = \sum_{i=1}^m Y_i(\vecx^*) \leq (1-\epsilon/10)\cdot \omega(\cF)\cdot q^{-(k-1)}\cdot m] = o(1)$. Combining this with the lower bound on $\tilde{m}$ from Claim (1) we conclude that $\Pr[\val_\Psi \leq (1-\epsilon/3)\cdot \omega(\cF)] = o(1)$. 
	
	\paragraph{Proof of Claim (3).}
	Finally we analyze the number of satisfiable constraints in the \no\ case. Fix $\vecnu \in \Z_q^k$ and let $\cD\in\Delta(\Z_q)$ be the distribution obtained by sampling a uniformly random $\ell\in[n]$ and outputting $\vecnu_\ell$. By \cref{prop:rho-equiv} we have that $\Exp_{f \sim \cD_\cF, \vecb \sim \cD^k}[f(\vecb)] \leq \rho(\cF)$. We use this to prove that for every $i \in [m]$, $\Exp[Y_i(\vecnu)|Y_{1:i-1}(\vecnu)] \leq (1+\epsilon/50)\cdot \rho(\cF)\cdot q^{-(k-1)}$. 
	
	First, as in the proof for Claim (2) we have that the total variation distance between $\vecb \sim \cD^k$ and $\{\vecnu_{\vecj(i)}| Y_{1:i-1}(\vecnu)\}$ is at most
	$k^2 \alpha$. (In particular, this is upper bounded by the probability that $k$ uniformly and independently chosen elements of $[n]$ either collide or fall in a set of size at most $k (\alpha n - 1)$.) We conclude that the probability that the $i$-th ``potential constraint'' (given by $(f(i),\vecj(i))$) is satisfied is at most $\rho(\cF) + k^2\alpha$. Next, note that the event $X_i = 1$ (i.e., the $i$-th constraint is chosen in $\Psi$) is independent of $Y_i(\vecnu)$ since in the \no\ case $\vecz(i)\in \Z_q^k$ is uniform and independent of all other random variables. We conclude that $\Exp[Y_i(\vecnu)|Y_{1:i-1}(\vecnu)] \leq (1+\epsilon/50)\cdot \rho(\cF)\cdot q^{-(k-1)}$. Finally, we apply \cref{lem:our-azuma} again to conclude that 
	$\Pr[Y(\vecnu) = \sum_{i=1}^m Y_i(\vecnu) > (1 + \epsilon/10)\cdot \rho(\cF)\cdot q^{-(k-1)}\cdot m] \leq c^{-m}$ where $c>1$ depends on $q,k,\cF,\alpha,\epsilon$ but not on $T$ or $n$. Thus by setting $T$ large enough, we can bound $c^{-m} \leq q^{-2n}$. This allows us to use the union bound to conclude that the probability that there exists $\vecnu \in \Z_q^n$ such that $Y(\vecnu) > (1+\epsilon/10)\cdot\rho(\cF)\cdot q^{-(k-1)}\cdot m$ is at most $q^{-n} = o(1)$. Combining with the lower bound on $\tilde{m}$ from Claim (1) we get that with probability $1 - o(1)$ we have $\val_{\Psi} \leq (1+\epsilon/3)\cdot\rho(\cF)$ in this case.\\
	
	This concludes the proofs of the claims and thus the proof of~\autoref{thm:main-technical}.
	
\end{proof}

\cref{thm:intro-approx-res,thm:intro-q-factor} follow immediately from \cref{thm:main-technical} as we show below.

\begin{proof}[Proof of \cref{thm:intro-approx-res}] 
	The theorem follows from the fact that for a wide family $\omega(\cF) = 1$ and in this case \cref{thm:main-technical} asserts that a $\rho(\cF) + \epsilon$ approximation requires linear space.
\end{proof}

\begin{proof}[Proof of \cref{thm:intro-q-factor}] 
	The theorem follows from the fact that for every non-zero function $f$ we have $\omega(f) \geq 1/q$ and 
	so for every family $\cF$ also we have $\omega(\cF) \geq 1/q$. Thus \cref{thm:main-technical} asserts that a $\rho(\cF)\cdot q + \epsilon$ approximation requires linear space, where $\rho(\cF)$ approximation is trivial.
\end{proof}

\subsection{Some examples}\label{ssec:examples}

We now give some examples illustrating the power of \cref{thm:main-technical}. Our first example is the familiar $q$-coloring problem.

\begin{examplebox}{Example 1 ($\textsf{Max-}q\textsf{Col}$).}
	Let $k=2$ and $q\geq2$. Let $\cF=\{f:\Z_q^2\to\{0,1\}\}$ where $f(u,v)=1$ if and only if $u \ne v$. The ``Max $q$-Coloring'' problem is defined to be  $\textsf{Max-}q\textsf{Col}=\textsf{Max-CSP}(\cF)$. It is easy to verify $\rho(\cF) = 1 - 1/q$ and $\omega(\cF) = 1$. We thus conclude by \cref{thm:intro-approx-res} that $\textsf{Max-}q\textsf{Col}$ is approximation resistant.
\end{examplebox}

Next we turn to the Unique Games Problem.

\begin{examplebox}{Example 2 ($\textsf{Max-}q\textsf{UG}$).}
	Let $k=2$ and $q\geq2$. Let $\cF=\{f:\Z_q^2\to\{0,1\}\, |\, f^{-1}(1)\text{ is a bijection}\footnote{We consider a set $S\subseteq \Z_q^2$ to be a bijection if for every $a\in \Z_q$, there exists a unique $a'\in Z_q$ such that $(a,a')\in S$ and there exists a unique $a''\in Z_q$ such that $(a'',a)\in S$.}\}$. The ``$q$-ary Unique Games'' problem is defined to be  $\textsf{Max-}q\textsf{UG}=\textsf{Max-CSP}(\cF)$. We show below that $\rho(\cF)=1/q$. We also show that there exists $\cF' \subseteq \cF$ such that
	$\rho(\cF') = 1/q$ and $\omega(\cF') = 1$. Applying \cref{thm:intro-approx-res} to $\cF'$ we get that $1/q + \epsilon$ approximating $\textsf{Max-CSP}(\cF')$ requires linear space and the same holds for $\textsf{Max-}q\textsf{UG}=\textsf{Max-CSP}(\cF)$ by monotonicity.
	
	We define the family $\cF'$ to be $\cF' = \{f_a | a \in \Z_q\}$ where $f_a(u,v) = 1$ if and only if $u=v+a$. Let $\cD = \textsf{Unif}(\Z_q)$. For every $f \in \cF$ we have that
	$\Exp_{(u,v) \sim \cD^2}[f(u,v)] = 1/q$. So for every $\cD_\cF \in \Delta(\cF)$ we have $\Exp_{f \sim \cD_\cF} \Exp_{(u,v) \sim \cD^2}[f(u,v)] = 1/q$.
	This proves $\rho(\cF),\rho(\cF') \geq 1/q$. To get the upper bound we let $\cD_\cF$ be uniform over $\cF'$. For every $(u,v) \in \Z_q^2$ we have $\Exp_{f \sim \cD_\cF} [f(u,v)] = 1/q$ and so for every distribution $\cD \in \Delta(\Z_q^k)$ (which is more than we need) we have $\Exp_{f \sim \cD_\cF} \Exp_{(u,v) \sim \cD}[f(u,v)] \leq 1/q$. This proves $\rho(\cF'),\rho(\cF) = 1/q$ (since $\cD_\cF$ is supported on $\cF'$).
	
	Now turning to $\omega(\cF')$, note that for every $f_a \in \cF'$ we have $\{(b+a,b) | b \in \Z_q\} \subseteq f_a^{-1}(1)$. Thus $\omega(f_a) \geq \omega_{(a,0)}(f_a) = 1$. It follows that $\omega(\cF') = 1$. 
	
\end{examplebox}

Our third example talks about constraints that are simple equalities.

\begin{examplebox}{Example 3 ($\textsf{Max-$k$-All-Equal}_{q}$).}
	For $k\geq 2$ and prime $q$ , we define $\textsf{Max-$k$-All-Equal}_{q}$ to be the $\maxF$ for $\cF = \{f_{\mathrm{All-EQ}}\}$ where $f_{\mathrm{All-EQ}}(x_1,\ldots,x_k)= 1$ if and only if $x_1 = \cdots = x_k$. 
	It is easy to verify that for every $k,r,q$, $\rho(\cF) \geq q^{-1}$ (In particular $\cD_\cF$ is trivial since $|\cF| = 1$.)  Since every $\vecb \in \\F_q^k$ has width $1$ it follows that $\omega(\cF)=1$ and so $\maxF$ can not be approximated to within $\frac1q(1+\epsilon)$-factor in $o(n)$ space (and is thus approximation resistant.
\end{examplebox}

Our next example generalizes the above to all  linear systems.

\begin{examplebox}{Example 4 ($\textsf{Max-Lin}_{k,r,q}$).}
	For $k\geq 2$ and prime $q$ and $0 \leq r < k$, we define $\textsf{Max-Lin}_{k,r,q} = \maxF$ for $\cF = \cF_{k,r,q} = \{f_{A,\vecb}:\Z_q^k\to\{0,1\} | A \in \Z_q^{r\times k},\vecb \in \Z_q^{k}\}$ where $f_{A,\vecb}(x) = 1$ if and only if $Ax = Ab$. (Thus constraints are systems of satisfiable linear equations with solutions of dimension at least $k-r$.)
    Note that the \textsf{Max $k$-LIN}-$\bmod\; q$ problem mentioned in the abstract and \cref{ssec:intro-tech} is the special case where $r=k-1$. We show below that $\textsf{Max-Lin}_{k,r,q} = \maxF$ is approximation-resistant for every $1 \leq r \leq k-1$. 
	Let $\cF'_{k,r,q} = \{f_{r,k}\}$ where $f_{r,k}(x_1,\ldots,x_k) = 1$ if and only if $x_1 = \cdots = x_{r+1}$. It is easy to see that $q^{-r} \leq \rho(\cF) \leq \rho(\cF') = q^{-r}$. Furthermore $\omega(\cF') = 1$ (as argued in Example 3). 
    Thus, applying \cref{thm:intro-approx-res} to $\cF'$ we get that $\textsf{Max-CSP}(\cF')$ is approximation-resistant. The same holds for $\textsf{Max-Lin}_{k,r,q}=\textsf{Max-CSP}(\cF)$ by monotonicity.\footnote{We believe this system is not approximation resistant for $r=k$. This is proved for $q=2$ in \cite[Lemma 2.14]{CGSV21}. The case of general $q$ may not have been explicitly resolved in previous work.} 
\end{examplebox}

Finally we mention one more problem. This problem arises in the work of Singer, Sudan and Velusamy~\cite{SSV21} who use it to show the approximation resistance of the ``maximum acyclic subgraph'' problem to $o(\sqrt{n})$ space algorithms. We suspect the improved space lower bound should improve their work to rule out $o(n)$ space algorithms. 

\begin{examplebox}{Example 5 ($\textsf{Max-Less-Than}_{q}$).}
	For $k=2$ and $q\geq 2$ we define $\cF = \{<_q\}$ where $<_q : \Z_q^2 \to \{0,1\}$ is given by $<_q(u,v) = 1$ if and only if $u < v$. 
	It is possible to show $\rho(\cF) = \frac12(1-1/q)$. Also $\omega_{(0,1)}(<_q) = 1-1/q$ and this can be used to show that $\omega(\cF) = 1 - 1/q$. By \cref{thm:main-technical} it follows that $1/2+\epsilon$-approximating $\maxF$ requires linear space.
\end{examplebox}

\section{Lower bound on the communication complexity}\label{sec:proof overview}

In this section we prove a linear lower bound on the communication complexity of \IFRMD\ (\cref{thm:main IFRMD}). Our proof is via a hybrid argument which starts with all players receiving inputs from the \no\ distribution, and switching the players' input distributions one at a time, starting with Player~1, to the \yes\ distribution.
We state a key ``hybrid lemma'' (\cref{lem:hybrid}) which asserts that any one step of switching does not alter the distribution of the message output by the switched player.

To state our lemma we recall some notations and set up a few new ones. Let $\alpha,n,k,q,T,m=\alpha n\in\N$ denote the usual parameters of \IFRMD. Recall that the player $t$ gets as input a matrix $A_{t,\vecc_t}\in\Z_q^{(k-1)m\times n}$ corresponding to a $k$-uniform hypermatching $M_t$ consisting of $m$ hyperedges folded over the center vector $\vecc_t$ and a vector $\vecw_t \in \Z_q^{(k-1)m}$. For notational convenience, we will separate the input $A_{t,c_t}$ into a matrix $A_t \in \Z_q^{(k-1)m\times n}$ and the center $\vecc_t$. The message $S_t$ sent by the $t$-th player is a function of $A_{1:t},\vecc_{1:t},\vecw_t$ and $S_{1:t-1}$.\footnote{Note that even though the $t$-th player does not have access to $A_{1:t-1},\vecc_{1:t-1}$, and $S_{1:t-2}$, allowing them to see these only makes our lower bound stronger.} Next, note that by Yao's principle \cite{yao1977probabilistic}, we may assume that the messages sent by the players in $\IFRMD$ are all deterministic. 
Namely, a protocol for $\IFRMD$ can be specified by deterministic message functions $r_1,r_2,\dots,r_T$ so that $S_t=r_t(A_{1:t},\vecc_{1:t},S_{1:t-1},\vecw_t)$ denotes the message sent by the $t$-th player. The communication complexity of a protocol is defined as the largest output length of $r_t$.
When $(A_{1:T},\vecc_{1:T},\vecw_{1:T})$ is drawn from the \yes\ distribution (resp. the \no\ distribution), we denote by $S_{1:T}^Y$ (resp. $S_{1:T}^N$) the resulting messages. Without loss of generality $S_T$ is just a bit ```Yes/No'' indicating the output of the protocol. Thus, to prove~\autoref{thm:main IFRMD} we need to show that $S^Y_T$ and $S^N_T$ are close in total variation distance. For the induction we prove the much stronger statement that  $(A_{1:T},\vecc_{1:T},S_{1:T}^Y)$ and $(A_{1:T},\vecc_{1:T},S_{1:T}^N)$ are close in total variation distance, i.e.,
\[
\|(A_{1:T},\vecc_{1:T},S_{1:T}^Y)-(A_{1:T},\vecc_{1:T},S_{1:T}^N)\|_{tvd}\leq\delta \, .
\]
The following lemma provides the key step in this analysis. Roughly it says that if the first $t-1$ players' inputs are according to the \yes\ distribution then the $t$-th player's output on the \yes\ input is typically distributed very similarly to the output on the \no\ distribution (even conditioned on all previously announced hypermatchings, centers and messages). Formally, the lemma identifies a sequence of events $\cE_1\supset\cE_2\supset\cdots\supset\cE_T$ such that (i) $\cE_t$ enforces a ``typicality'' restriction on the messages and inputs that the $t$-th player receives and (ii) if the messages and input received by the $t$-th player are typical then the player cannot distinguish whether its input is sampled from the \yes\ distribution or the \no\ distribution (assuming all previous players' inputs were from the \yes\ distribution).


\paragraph{The Probability Space:} 
In what follows in the rest of this section (and indeed in the rest of this paper), the underlying probability space will be that of describing all the inputs in the communication problem. Specifically, we let $\Omega = \Omega_{k,q,\alpha,n,T}$ be the distribution over tuples $(\vecx^*, A_{1:T}, \vecc_{1:T}, \vecw_{1:T})$ where $\vecx^*\sim \Unif(\Z_q^n)$, $A_t \in \{0,1\}^{\alpha k n \times n}$ is the incidence matrix of a uniform random $k$-hypermatching on $[n]$ with $\alpha n$ edges, $\vecc_t \in [n]^{\alpha n}$ is a uniform choice of centers consistent with $A_t$, and $\vecw_t \in \Z_q^{\alpha(k-1)n}$ is a uniform vector, for every $t \in [T]$. These variables along with a deterministic protocol given by $r_1,\ldots,r_T$ specify additional random variables that are determined by $(\vecx^*, A_{1:T}, \vecc_{1:T}, \vecw_{1:T})$ including $B_{1:t}$, $B_{r,1:t}$, $S^Y_t$. 
Thus when we write a probability expression of the form $\Pr[X]$ without specifying the random variables we intend the space to be $\Omega$. Furthermore an expression of the form $\forall Y$, $\Pr[X|Y]$ is shorthand for $\forall y$, $\Pr[X|Y=y]$.

\begin{restatable}[Hybrid lemma]{lemma}{hybrid}\label{lem:hybrid}
	For every $q,k\in\N$, there exists $\alpha_0>0$ such that for every $T\in\N$, and $\delta\in(0,1)$, there exists $\tau\in(0,1)$ and $n_0 <\infty$  such that  the following holds for every $n \geq n_0$:
	
    Let $\Pi=(r_1,\dots,r_T)$ be a deterministic protocol for $\IFRMD$ where each message function $r_t$ outputs a message of at most $\tau n$ bits. Let $(\vecx^*,A_{1:T},\vecc_{1:T},\vecw_{1:T}) \sim \Omega$. 
	Then there exists a sequence of events $\{\cE_t\}_{t\in[T]}$ and non-negative $\delta_1,\ldots,\delta_T$ with $\sum_{t=1}^T \delta_t \leq \delta/2$  
 such that:
    \begin{enumerate}[label=(\roman*)]
        \item $\cE_1$ holds with probability $1$. For $t \geq 2$,  $\cE_t$ only depends on $(A_{1:t},\vecc_{1:t})$ and $S^Y_{1:t-1}$ (with $S_{1:0}$ denoting an empty set of variables). 
        \item For every $t \geq 2$,  $\cE_{t} \Rightarrow \cE_{t-1}$ and $\Pr[\overline{ \cE_{t}}\, |\, \cE_{t-1}]\leq\delta_t$.
        \item For every fixed $(A_{1:t},\vecc_{1:t})$ and $S^Y_{1:t-1}$ satisfying $\cE_t$, one has
	\begin{equation}\label{eq:hybrid}
		\|S^Y_t-r_t(A_{1:t},\vecc_{1:t},S^Y_{1:t-1},U)\|_{tvd}\leq\delta_t,
	\end{equation}
	where $U\sim\Unif(\Z_q^{(k-1)\alpha n})$.
    \end{enumerate}
\end{restatable}

\cref{thm:main IFRMD} follows almost immediately from \cref{lem:hybrid} as shown in \cref{ssec:proof-IFRMD-thm}. 
In the rest of this paper we prove \cref{lem:hybrid}. In this section we introduce some new notions and state three key lemmas that together suffice to prove \cref{lem:hybrid}. This (conditional) proof is given in~\cref{sec:hybrid}. In the following sections we prove the key lemmas.
First we give an overview of the proof of \cref{lem:hybrid} that explains the nature of these key lemmas. 

The general idea behind the proof of \cref{lem:hybrid} is to argue that information about $\vecx^*$ ``leaked'' by the messages of the first $t-1$ players (i.e., $S_{1:t-1}$) is not sufficient for the $t$-th player to distinguish between the case where $\vecw_t = A_{t,c_t} \vecx^*$ (the \yes\ case) and the case where $\vecw_t$ is uniform. The earlier proofs of this type (in particular as in \cite{KKS}) simply counted the total information gleaned about $\vecx^*$ which is bounded by the total communication. Such proofs are inherently limited to achieving only a $\sqrt{n}$ lower bound. To go further \cite{KK19} introduced the approach of reasoning about the structure of the information learned about $\vecx^*$.  Note in particular that no player sees $\vecx^*$ directly, and the $t'$-th player only sees $A_{t',c_{t'}}\cdot \vecx^*$. (In particular no coordinate of $\vecx^*$ is revealed directly, though the sum of many pairs of coordinates are directly revealed.) Thus the information about $\vecx^*$ comes from a ``reduced space'' and we would like to capture and exploit the structural restriction imposed by this restriction. Information-theoretic tools seem to fail to capture this restriction and the key to the work of \cite{KK19} is to give a Fourier analytic condition, that they call ``boundedness'', that captures this restriction.

The boundedness condition applies to what we call the ``posterior distribution'' of $\vecx^*$, i.e., the distribution of $\vecx^*$ conditioned on the first $t$ messages. This distribution turns out to be the uniform distribution over a set $B_t \subseteq \Z_q^n$ (see \cref{lem:conditional prob}). The boundedness condition places restrictions on the Fourier spectrum of the indicator function of this set. (See \cref{def:boundedness}.) To use this condition we need three ingredients elaborated below, which we abstract as lemma statements in this section and prove in later sections. Given these three lemmas the proof of \cref{lem:hybrid} follows and is given in \cref{sec:hybrid}. 

The first ingredient we need is that boundedness of $B_{t-1}$ does imply that the $t$-th player is unable to distinguish between its input being from the \yes\ distribution or the \no\ distribution. This is stated as \cref{lem:boundedness implies uniform}. Next we need to show that given information about $A_{t,c_t}\vecx^*$, the posterior distribution of $\vecx^*$ is indeed bounded, and we assert this in \cref{lem:boundedness base case}. 
Finally we argue that if $B_{t-1}$ is bounded, then for most pairs of matchings $A_{t}$ and  centers $\vecc_t$ the resulting set $B_t$ is bounded. This is asserted in \cref{lem:induction step}. See also~\cref{fig:hybrid} for a pictorial overview of the proof structure of~\cref{lem:hybrid}.

In the rest of this section, after showing that \cref{lem:hybrid} implies \cref{thm:main IFRMD} in \cref{ssec:proof-IFRMD-thm}, we introduce the posterior sets and discuss their basic properties in \cref{ssec:posterior}, we introduce boundedness and state the three lemmas above in \cref{ssec:bounded-def}, and finally conclude with the proof of \cref{lem:hybrid} in \cref{sec:hybrid}.

\subsection{Proof of \texorpdfstring{\cref{thm:main IFRMD}}{Theorem~\ref{thm:main IFRMD}}}\label{ssec:proof-IFRMD-thm}

We now show how the lemma suffices to prove~\autoref{thm:main IFRMD}. The proof is analogous to the proof of Lemma 6.3 in \cite{KK19}. We remark that the lemma is not immediate and effectively depends on the fact that players can jointly sample from the \no\ distribution on their own. (Note the players can't jointly sample from the $\yes$ distribution since these samples are correlated by the hidden vector $\vecx^*$. So the proof is inherently asymmetric via the treatment of the \yes\ and \no\ distributions.)

\begin{proof}[Proof of~\cref{thm:main IFRMD}]

    {
     For the sake of contradiction, assume that there exists a protocol $\Pi=(r_1,\dots,r_T)$ that solves IFRMD with advantage more than $\delta$ and less than $\tau n$ bits of communication for some $n\geq n_0$. In what follows, we will show that $\|(A_{1:T},\vecc_{1:T},S^Y_{1:T})-(A_{1:T},\vecc_{1:T},S^N_{1:T})\|_{tvd}\leq\delta$, which implies that the advantage of the protocol cannot be greater than $\delta$, hence producing a contradiction.

     Let $\cE_1 \supset \cE_2 \supset \cdots \supset \cE_T$ be the sequence of events guaranteed by \autoref{lem:hybrid} such that $\Pr\left[\overline{\cE_t}\, |\, \cE_{t-1}\right]\leq\delta_t$ for $t\geq2$. Note that by \autoref{lem:hybrid}, we also have 
    \[\|S^Y_t-r_t(A_{1:t},\vecc_{1:t},S^Y_{1:t-1},U)\|_{tvd,\cE_t}\leq\delta_t\] 
    for all $t\in[T]$, where $\|\cdot\|_{tvd,\cE_t}$ denotes the total variation distance, conditioned on $\cE_t$.
    We inductively show that for every $t\in [T]$,
    \[
	\|(A_{1:t},\vecc_{1:t},S^Y_{1:t})-(A_{1:t},\vecc_{1:t},S^N_{1:t})\|_{tvd}\leq  \sum_{1\leq j\leq t}\left( \delta_{j} +  \Pr[\overline{\cE_j}|\cE_{j-1}]\right)\,  \tag{Induction hypothesis}
	\] 
    where $\cE_0$ is the trivial event that is always true.

	First, we prove the base case $t=1$. Recalling that $S_0^Y=S_0^N$, we have
	\begin{align*}
		\|(A_1,\vecc_1,S_1^Y)-(A_1,\vecc_1,S_1^N)\|_{tvd,\cE_1}&=\|(A_1,\vecc_1,S_1^Y)-(A_1,\vecc_1,r_1(M_1,\vecc_1,S_0^N,U_1))\|_{tvd,\cE_1}\\
		&=\|(A_1,\vecc_1,S_1^Y)-(A_1,\vecc_1,r_1(M_1,\vecc_1,S_0^Y,U_1))\|_{tvd,\cE_1} \, .
	\end{align*}
	Observe that for every fixed $A_1,\vecc_1$ and $S^Y_0$ satisfying $\cE_1$, we have 
    $\|S_1^Y-r_1(M_1,\vecc_1,S_0^Y,U_1)\|_{tvd}\le \delta_1$%
    , where the randomness is over $S_1^Y$ and $U_1$. It follows from \cref{lem:statistical_test} that 
    \[\|(A_1,\vecc_1,S_1^Y)-(A_1,\vecc_1,r_1(M_1,\vecc_1,S_0^Y,U_1))\|_{tvd,\cE_1} \le \delta_1.\]
    Therefore,
    \begin{align*}
       \|(A_1,\vecc_1,S_1^Y)-(A_1,\vecc_1,r_1(M_1,\vecc_1,S_0^Y,U_1))\|_{tvd} &\le \|(A_1,\vecc_1,S_1^Y)-(A_1,\vecc_1,r_1(M_1,\vecc_1,S_0^Y,U_1))\|_{tvd,\cE_1} + \Pr[\overline{\cE_1}]\\
       & \le \delta_1 + \Pr[\overline{\cE_1}] \, , 
    \end{align*} 
    which completes the base case.

    Next, we prove the inductive step. For every $t=2,\dots,T$, we have 
	\begin{multline*}
		\| (A_{1:t},\vecc_{1:t},S^Y_{1:t})-(A_{1:t},\vecc_{1:t},S^N_{1:t})\|_{tvd} \\
		= \|(A_{1:t},\vecc_{1:t},S^Y_{1:t-1},r_t(A_{1:t},\vecc_{1:t},S^Y_{t-1},A_{t } \vecx^*))-(A_{1:t},\vecc_{1:t},S^N_{1:t-1},r_t(A_{1:t},\vecc_{1:t},S^N_{t-1},U))\|_{tvd}\, .
	\end{multline*}
	Let us define $Q_{t-1}^Y = (A_{1:t-1},\vecc_{1:t-1},S^Y_{1:t-1})$  and $Q_{t-1}^N = (A_{1:t-1},\vecc_{1:t-1},S^N_{1:t-1})$. 
	Then, we can rewrite the above expression for total variation distance in terms of the new notation as follows:
	\begin{multline}\label{eqn:tvd_substitution}
		\|(A_{1:t},\vecc_{1:t},S^Y_{1:t-1},r_t(A_{1:t},\vecc_{1:t},S^Y_{t-1},A_{t ,\vecc_t} \vecx^*))-(A_{1:t},\vecc_{1:t},S^N_{1:t-1},r_t(A_{1:t},\vecc_{1:t},S^N_{t-1},U))\|_{tvd} \\
		= \|(Q_{t-1}^Y,A_t,\vecc_t,r_t(Q_{t-1}^Y,A_t,\vecc_t,A_{t ,\vecc_t} \vecx^*))-(Q_{t-1}^N,A_t,\vecc_t,r_t(Q_{t-1}^N,A_t,\vecc_t,U))\|_{tvd} \, .
	\end{multline}
	We now apply \autoref{lem:KKsubstitutionlemma} to \autoref{eqn:tvd_substitution}. Applying this lemma with $X^1 = Q_{t-1}^Y$, $X^2 = Q_{t-1}^N$, $Z^1=(A_t,\vecc_t,A_{t ,\vecc_t} \vecx^*)$, $Z^2=(A_t,\vecc_t,U)$, and $f$ as the function that maps the tuple $(X,(B,C))$ to $(B,r_t(X,B,C))$, we get
	\begin{align}\label{eqn:conditionalTVD}
		&\|(Q_{t-1}^Y,A_t,\vecc_t,r_t(Q_{t-1}^Y,A_t,\vecc_t,A_{t ,\vecc_t} \vecx^*))-(Q_{t-1}^N,A_t,\vecc_t,r_t(Q_{t-1}^N,A_t,\vecc_t,U))\|_{tvd} \nonumber \\
		\le\ &\|Q_{t-1}^Y- Q_{t-1}^N\|_{tvd} + \| (Q_{t-1}^Y,A_t,\vecc_t,r_t(Q_{t-1}^Y,A_t,\vecc_t,A_{t ,\vecc_t} \vecx^*))-(Q_{t-1}^Y,A_t,\vecc_t,r_t(Q_{t-1}^Y,A_t,\vecc_t,U))\|_{tvd} \, .
	\end{align}

	Now, by applying the induction hypothesis, we have that 
    \begin{equation}\label{eqn:induction}\|Q_{t-1}^Y- Q_{t-1}^N\|_{tvd} \le \sum_{j=1}^{t-1} \left(\delta_j + \Pr[\overline{\cE_j}|\cE_{j-1}]\right).
	\end{equation}
	Next, we bound the second term on the right hand side of \eqref{eqn:conditionalTVD}, i.e.,
	\[\| (Q_{t-1}^Y,A_t,\vecc_t,r_t(Q_{t-1}^Y,A_t,\vecc_t,A_{t ,\vecc_t} \vecx^*))-(Q_{t-1}^Y,A_t,\vecc_t,r_t(Q_{t-1}^Y,A_t,\vecc_t,U))\|_{tvd},\] 
	by applying condition (iii) from \cref{lem:hybrid}. According to this condition, for every \emph{fixed} $(A_{1:t},\vecc_{1:t})$ and $S^Y_{1:t-1}$ satisfying $\cE_t$, we have
    \begin{equation*}
		\|r_t(A_{1:t},\vecc_{1:t},S^Y_{1:t-1},A_{t ,\vecc_t} \vecx^*)-r_t(A_{1:t},\vecc_{1:t},S^Y_{1:t-1},U)\|_{tvd}\leq\delta_t,
	\end{equation*}
	where $U\sim\Unif(\Z_q^{(k-1)\alpha n})$. Thus, by \autoref{lem:statistical_test}, it follows that 
    \begin{equation}\label{eqn:tvd}\| (Q_{t-1}^Y,A_t,\vecc_t,r_t(Q_{t-1}^Y,A_t,\vecc_t,A_{t ,\vecc_t} \vecx^*))-(Q_{t-1}^Y,A_t,\vecc_t,r_t(Q_{t-1}^Y,A_t,\vecc_t,U))\|_{tvd,\cE_t} \le \delta_t\, .\end{equation}
	Combining \cref{eqn:tvd_substitution,eqn:conditionalTVD,eqn:induction,eqn:tvd}, we have 
    \[
	\|(A_{1:t},\vecc_{1:t},S^Y_{1:t})-(A_{1:t},\vecc_{1:t},S^N_{1:t})\|_{tvd}\le \sum_{j=1}^{t} \left(\delta_j + \Pr[\overline{\cE_j}|\cE_{j-1}]\right) ,
	\]
	which completes the induction.
    
    Thus,
    \[|(A_{1:T},\vecc_{1:T},S^Y_{1:T})-(A_{1:T},\vecc_{1:T},S^N_{1:T})\|_{tvd}\le \sum_{j=1}^T\left( \delta_j + \Pr[\overline{\cE_j}|\cE_{j-1}]\right) \le 2\cdot \sum_{j=1}^T \delta_j  \le \delta\, .\]
    This implies that $\Pi$ cannot have advantage more than $\delta$,
    which contradicts the assumptions of the theorem statement. Therefore, we conclude that any protocol for $\IFRMD$ with advantage $\delta$ requires $\tau n$ bits of communication, as desired.
    }
\end{proof}

\subsection{Posterior sets and functions}\label{ssec:posterior}

The main challenge in proving~\autoref{lem:hybrid} lies in the condition (iii), i.e., requiring the closeness of the Yes message (i.e., $S^Y_t=r_t(A_{1:t},\vecc_{1:t},S^Y_{1:t-1},A_{t,\vecc_t}\vecx^*)$) and the hybrid No message (i.e., $r_t(A_{1:t},\vecc_{1:t},S^Y_{1:t-1},U)$). 
Intuitively, if $\vecx^*\sim\Unif(\Z_q^n)$ and is independent of the other arguments, then $A_{t,\vecc_t}\vecx^*$ is uniformly distributed over $\Z_q^{(k-1)\alpha n}$ and hence $S^Y_t$ follows the same distribution as $r_t(A_{1:t},\vecc_{1:t},S^Y_{1:t-1},U)$. However, $\vecx^*$ is correlated\footnote{In particular, $\vecx^*$ has to be consistent with the previous messages $S^Y_{1:t-1}$.} with the previous messages $S^Y_{1:t-1}$ so the above ideal situation would not happen in general. Nevertheless, we are able to analyze the conditional distribution of $A_{t,\vecc_t}\vecx^*$ on the previous messages by explicitly characterizing the \textit{posterior distribution} of $\vecx^*$ after receiving the messages from the first $t-1$ players. That is, the conditional distribution of $A_{t,\vecc_t}\vecx^*$ can be described by first sampling $\vecx^*$ from the posterior distribution and then applying $A_{t,\vecc_t}$.

For every fixed $A_{1:t},\vecc_{1:t}$ and $S_{1:t}$, we would like to identify a distribution $\cD_t$ over $\Z_q^n$ such that $\cD_t$ is the conditional distribution of $\vecx^*$ given messages $S_{1:t}$. Note that by the choice of the No case, the conditional distribution of $\vecx^*$ given messages $S_{1:t}$ is simply the uniform distribution over $\Z_q^n$. Thus, we only need to worry about the Yes case.

\begin{definition}[Posterior sets and functions]\label{def:set and pdf}
	Under the setting described above, for each $t$ and fixed $A_{1:t}$, $\vecc_{1:t}$, and $S_{1:t}$, define
	\begin{itemize}
		\item (Reduced posterior set) $B_{r,t}\subseteq\Z_q^{(k-1)m}$ be the set of possible values of $z_t=A_{t,\vecc_t}\vecx$ that leads to message $S_t$; Note that $B_{r,t}$ should be thought of as a function on $A_t$, $\vecc_t$, and $S_t$ in the sense that $B_{r,t}=g_t^{-1}(S_t)$ where $g_t(\cdot)=r_t(A_{1:t},\vecc_{1:t},S_{1:t-1},\cdot)$. Let $\mathbf{1}_{B_{r,t}}$ be the indicator function of $B_{r,t}$. 
		\item (Posterior set and function) Let
		\[
		B_t := \{\vecx\in\Z_q^n\, |\, A_{t,\vecc_t}\vecx\in B_{r,t}\} \, .
		\]
		Also, let $\mathbf{1}_{B_t}:\Z_q^n\to\{0,1\}$ be the indicator function of $B_t$.
		\item (Aggregated posterior set and function) Let
		\[
		B_{1:t}:=\{\vecx\in\Z_q^n\, |\, A_{t',\vecc_{t'}}\vecx\in B_{r,t'},\ \forall t'=1,\dots,t\}=\bigcap_{t'=1}^tB_{t'} \, .
		\]
		Also, let $\mathbf{1}_{B_{1:t}}:\Z_q^n\to\{0,1\}$ be the indicator function of $B_{1:t}$. Namely, $\mathbf{1}_{B_{1:t}}=\prod_{t'=1}^t\mathbf{1}_{B_{t'}}$.
	\end{itemize}
\end{definition}

Now, we show that $\one_{B_{1:t}}$ captures the posterior distribution (i.e., the conditional distribution) of $\vecx$ given messages $S_1, S_2, \dots, S_t$:

\begin{lemma}[Posterior function $\one_{B_{1:t}}$ captures the posterior distribution.]\label{lem:conditional prob}
	For every $t\in[T]$, $\veca \in \F_q^n$, $A_{1:t}$, $\vecc_{1:t}$ and  $S_{1:t}$,   
    \[\Pr[\vecx^*=\veca| A_{1:t}, \vecc_{1:t}, S^Y_{1:t}] = \one_{B_{1:t}}(\veca)/|B_{1:t}| \, .\]
	In particular, for fixed $A_{1:t}$, $\vecc_{1:t}$, and $S_{1:t-1}^Y$, we have $S^Y_t= r_t(A_{1:t},\vecc_{1:t},S_{1:t-1}^Y,A_{t,\vecc_t}\vecx^*)$, where $\vecx^*\sim\Unif(B_{1:t-1})$.
\end{lemma}

\begin{proof} 
    Recall that $S^Y_t = r_t(A_{1:t},\vecc_{1:t},S_{1:t-1}^Y,A_{t,\vecc_t}\vecx^*)$ by definition, \\ $B_t = \{\vecb \in \F_q^n \mid r_t(A_{1:t},r_t(A_{1:t},\vecc_{1:t},S_{1:t-1}^Y,A_{t,\vecc_t}\vecb) = S^Y_t\},$
    and $B_{1:t} = B_1 \cap B_2 \cap \cdots \cap B_t$. 
    It follows that if $\veca \not\in B_{1:t}$ then there must exists a smallest index such that $S^Y_i \ne r_i(A_{1:i},r_i(A_{1:i},\vecc_{1:i},S_{1:i-1}^Y,A_{i,\vecc_i}\veca)$ and so the probability that $\vecx^* = \veca$ conditioned on $r_i(A_{1:i},\vecc_{1:i},S_{1:i-1}^Y,A_{i,\vecc_i}\vecx^*) = S^Y_i$ is zero. For $\veca \in B_{1:t}$, we simply note that $\vecx^*$ is a priori uniformly distributed over $\Z_q^n$ and conditioning on any event (in our case that $\vecx^*\in B_{1:t}$) its distribution is uniform on the subset of $\Z_q^n$ for which the event holds.
\end{proof}

Now that we have a characterization of the posterior distribution of $\vecx^*$, the following corollary shows that~\autoref{eq:hybrid} (i.e., the condition (iii) of~\cref{lem:hybrid}) can be simplified to bounding the total variation distance between the posterior distribution and the uniform distribution.
\begin{corollary}[Reducing~\cref{eq:hybrid}]\label{cor:reduce hybrid condition}
	Let $r_t,S^Y_{1:t-1},A_{1:t},\vecc_{1:t},B_{1:t},U$ be defined as before, we have
	\[
	\|r_t(A_{1:t},\vecc_{1:t},S^Y_{1:t-1},A_{t,\vecc_t}\vecx^*)-r_t(A_{1:t},\vecc_{1:t},S^Y_{1:t-1},U)\|_{tvd}\leq\|(A_{t,\vecc_t}\vecx^*)-U\|_{tvd}
	\]
	where $\vecx^*\sim\Unif(B_{1:t})$.
\end{corollary}
\begin{proof}
	By~\cref{lem:conditional prob}, we have
	\[
	S^Y_t=r_t(A_{1:t},\vecc_{1:t},S^Y_{1:t-1},A_{t,\vecc_t}\vecx^*)
	\]
	where $\vecx^*\sim\Unif(B_{1:t})$. Note that when we fix $A_{1:t}$, $\vecc_{1:t}$, and $S^Y_{1:t-1}$ (hence $B_{1:t}$ is also fixed), by data processing inequality (see item 2 of~\autoref{prop:tvd properties}) we have
	\[
	\|r_t(A_{1:t},\vecc_{1:t},S^Y_{1:t-1},A_{t,\vecc_t}\vecx^*)-r_t(A_{1:t},\vecc_{1:t},S^Y_{1:t-1},U)\|_{tvd}\leq\|(A_{t,\vecc_t}\vecx^*)-U\|_{tvd} \, .\qedhere
	\]
\end{proof}
Namely,~\autoref{eq:hybrid} (i.e., the condition (iii) of~\autoref{lem:hybrid}) can be replaced with $\|(A_{t,\vecc_t}\vecx^*)-U\|_{tvd}\leq\delta/T$, i.e., after applying a random folded hypermatching matrix $A_{t,\vecc_t}$ to the posterior distribution $\Unif(B_{1:t})$, the distribution of the resulting string is close to the uniform distribution $\Unif(\Z_q^{(k-1)\alpha n})$.

Finally, the following lemma shows that when the amount of communication is small, the posterior set is large with high probability.
\begin{lemma}[Posterior set is large]\label{lem:posterior set is large}
	Let $\Pi=(r_1,\dots,r_T)$ be a deterministic protocol for $\IFRMD$ where each message function $r_t$ outputs a message of length at most $s$ bits for some $1\leq s\leq n$. Let $B_t$ be the posterior set defined in~\autoref{def:set and pdf} for every $t\in[T]$. For every $\delta\in(0,1)$ and $t\in[T]$, we have $|B_t|\geq \delta\cdot q^{n-s}$ with probability at least $1-\delta$ over the randomness of $\vecx\in\Z_q^n$.
\end{lemma}
\begin{proof}
	Fix a hypermatching $M$ and centers $\vecc$, the $t$-th message function induces a partition $P_1\cup P_2\cup\cdots\cup P_{2^{s}}$ of $\Z_q^n$. For each $\vecx\in\Z_q^n$, we define $P(\vecx)$ to be the part that contains $\vecx$, i.e, if $\vecx\in P_i$, then $ P(\vecx)=P_i$. Note that
	\[
	\Exp_{\vecx\in\Z_q^n}\left[\frac{1}{|P(\vecx)|}\right] = \sum_{i=1}^{2^s}\frac{\Pr_{\vecx\in\Z_q^n}[\vecx\in P_i]}{|P_i|} = \sum_{i=1}^{2^{s}}\frac{|P_i|\cdot q^{-n}}{|P_i|}=\frac{2^{s}}{q^n}\leq q^{s-n} \, .
	\]
	By Markov's inequality, we have $|P(\vecx)|<\delta\cdot q^{n-s}$ with probability at most $\delta$ as desired.
\end{proof}

\subsection{Fourier analytic conditions}\label{ssec:bounded-def}
In this subsection, we define and analyze Fourier-analytic properties of the posterior set $B$ and show that these properties are sufficient for the condition (iii) (i.e.,~\autoref{cor:reduce hybrid condition}) of~\autoref{lem:hybrid}.

\subsubsection{Three key definitions}
Recall that given a matching $M = (e_1,\ldots,e_m)$ and centers $\vecc = (c_1,\ldots,c_m)$, $A_\vecc$ is the $\vecc$-centered folded encoding of $M$. We are going to define three properties for sets $B$ in $\Z_q^n$.
First, we say a set $B\subseteq\Z_q^n$ is $(M,\vecc)$-restricted if $B$ is a union of cosets  (affine shifts) of the null space of $A_{\vecc}$.

\begin{restatable}[Restricted set]{definition}{restricted}\label{def:resitrctness}
	Let $M$ be a $k$-hypermatching of size $m$ and $\vecc$ be centers. We say a set $B\subseteq\Z_q^n$ is $(M,\vecc)$-restricted if there exists a (``reduced'') set $B_r \subseteq \Z_q^{(k-1)m}$ such that $B = \{\vecx \in \Z_q^n \, |\, A_{\vecc}\vecx \in B_r\}$.
\end{restatable}

Note that the posterior set (\cref{def:set and pdf}) of round $t$ is $(M_t,\vecc_t)$-restricted.

Next, we introduce the notion of a subset of $\Z_q^n$ being \textit{(strongly/weakly) bounded}. These notions are similar to those in \cite[Definition 4.3]{KK19}. They say that a set $B$ is bounded if the Fourier spectrum of the indicator function $\mathbf{1}_B$ can be appropriately bounded in terms of the $\ell_1$ norm on various Hamming levels. 

First, we introduce some notation. Note that for every set $B \subseteq \Z_q^n$, the $\mathbf{0}$-th Fourier coefficient of the indicator function $\mathbf{1}_B$ is $|B|/q^{n}$. In what follows we study the Fourier coefficients of $\mathbf{1}_B$ after scaling by $q^n/|B|$ so that the $\mathbf{0}$-th Fourier coefficient after scaling has value $1$. In what follows we define weak and strong bounding functions for the $\ell_1$ norm of the Fourier coefficients based on Hamming weight. Not all functions will satisfy the desired bounds, but indicators of posterior sets turn out to satisfy these bounds and this is the crux of our (and \cite{KK19}'s) analysis. 

\begin{equation}
	W_{C,s}(h) := \begin{cases}
		1, \quad &h = 0, \\
		\left(\frac{C\sqrt{s n}}{h}\right)^{h/2},     &1\leq h\leq s, \\
		\infty,     &h>s.
	\end{cases} \label{eq:weak-ucs-def}
\end{equation}
\begin{equation}
	U_{C,s}(h) := \begin{cases}
		W_{C,s}(h),     &0\leq h\leq s, \\
		\min\left\{W_{C,h}(h), \left(\frac{2q^2e^2 n}{h}\right)^{h/2}\right\},     &h>s.
	\end{cases} \label{eq:strong-ucs-def}
\end{equation}

(Above, $U$ stands for Upper bound, while $W$ stands for a Weak upper bound.)

\begin{restatable}[(Strongly/weakly) Bounded set]{definition}{bounded}\label{def:boundedness}
	Let $n,q\in\mathbb{N}$, $0\leq s\leq n$, $C>0$, and $B\subset\Z_q^n$. We say $B$ (as well as its indicator function $\mathbf{1}_B$) is $(C,s)$-(strongly)-bounded if, for every $h \in [n]$,
	\begin{equation}\label{eq:boundedness}
		\sum_{\substack{\vecu\in\Z_q^n\\\|\vecu\|_0=h}}\frac{q^{n}}{|B|}\left|\widehat{\mathbf{1}_B}(\vecu)\right|\leq U_{C,s}(h).
	\end{equation}
	We say that $B$ is {\em $(C,s)$-weakly-bounded} if the bound on the RHS above is replaced by $W_{C,s}(h)$. (Unless otherwise specified we use ``bounded'' to mean ``strongly bounded''.)
\end{restatable}

\begin{remark}
	As we keep track of posterior sets that are inductively refined, we will need the \emph{entire} Fourier spectrum of the corresponding indicator functions to be bounded from above by the function $U_{C,s}$ (for appropriate $C,s > 0$). The notion of boundedness is such that it allows us to show that $A_\vecc \vecx$ is close to the uniform distribution on $\Z_q^{(k-1)\alpha n}$ when $\vecx$ is drawn from a bounded posterior set $B \subset \Z_q^n$ (see \cref{lem:boundedness implies uniform}). Our upper bounds typically establish only the weak bound $W_{C,s}(h)$ (particularly,~\cref{lem:induction step exp} and~\cref{lem:q-simplify}), we usually prove this holds for every $s$ in some large interval and this allows us along with standard Fourier analysis (see \cref{lem:weak-vs-strong}) to establish the stronger bound for a slightly worse choice of constant $C$. 
\end{remark}

We describe some non-trivial properties of boundedness in \cref{sec:reduced bounded} but we start with some elementary assertions. 

\begin{proposition}\label{prop:basic-boundedness}
\begin{enumerate}

    \item If $B\subseteq \Z_q^n$ is $(C,s)$(-strongly/weakly)-bounded then it is also $(C',s)$(-strongly/weakly)-bounded for every $C' \geq C$.
    \item The set $B_0 = \Z_q^n$ is $(C,s)$-strongly-bounded for every $C \geq 0$ and every $0 \leq s \leq n$. 
\end{enumerate}
\end{proposition}

\begin{proof}
    Part (1) follows from the fact that $W_{C,s}(h) \leq W_{C',s}(h)$ and $U_{C,s}(h) \leq U_{C',s}(h)$ for every $s,h$ and $C' \geq C$. 
    Part (2) follows from the fact that $\widehat{\mathbf{1}_B}(\vec0) = 1$ and $\widehat{\mathbf{1}_B}(\vecu) = 0$ for all non-zero $\vecu \in \Z_q^n$ and so $B_0$ is $(0,s)$-strongly-bounded for every $0 \leq s \leq n$. Combining with Part (1) now yields the claim for every $C \geq 0$. 
\end{proof}

Finally, in what follows we will show that the intersection of a bounded set with a ``restricted set'' is also bounded and this will be the core of our induction. To do this we need to understand the Fourier behavior of restricted sets. It turns out that restricted sets satisfy a property stronger than being bounded, which we term ``reduced''-ness below. 


\begin{restatable}[(Weakly/Strongly) Reduced set]{definition}{reduced}\label{def:reducedness}
	Let $n,q\in\mathbb{N}$, $0\leq s\leq n$, $C>0$, and $B\subset\Z_q^n$. Let $M$ be a $k$-hypermatching. We say $B$ (as well as its indicator function $\mathbf{1}_B$) is $(M,C,s)$-(strongly)-reduced if the following hold. 
	\begin{itemize}
		\item For every $\vecu\in\Z_q^n$, if there exists $i\in[n]$ such that $u_i=1$ but $i$ is not contained~\footnote{We use ``contained in'' and ``touched by'' interchangeably as in some later contexts it makes more sense to use ``touched by'' when working with a set of vertices or hyperedges.} in $M$ (i.e., none of the hyperedges of $M$ contains $i$), then $\widehat{\mathbf{1}_B}(\vecu)=0$.
		\item For every $\vecu\in\Z_q^n$, if there exists a hyperedge $e_i$ of $M$ such that $\langle \vecu, \vece_i \rangle \not\equiv 0 \bmod{q}$, then $\widehat{\mathbf{1}_B}(\vecu)=0$.
		\item 
        For every $h\in\{1,\dots,n\}$ and $\vecv\in\Z_q^n$,
		\[
		\sum_{\substack{\vecu\in\Z_q^n\\\|\vecu+ \vecv\|_0=h}}\frac{q^{n}}{|B|}\left|\widehat{\mathbf{1}_B}(\vecu)\right|\leq U_{C,s}(h) \, .
		\]
	\end{itemize}
	If the bound in the RHS is replaced by the weaker $W_{C,s}(h)$ bound, then we say that $B$ is a {\em weakly-reduced} set. (Again we usually suppress the word ``strongly'' and simply refer to strongly-reduced sets as reduced set.)
\end{restatable}

As a remark, the first two conditions in the definition of reducedness are motivated by the Fourier analytic properties of posterior sets (e.g.,~\cref{lem:fourier coeff of set}, and~\cref{claim:g hat combo}). The third condition is a strengthening of boundedness. In particular the third condition applied with $\vecv=\mathbf{0}$ implies that every $(M,C,s)$-(strongly/weakly)-reduced set $B$ is also $(C,s)$-(strongly/weakly)-bounded.
That is why we say reducedness is the intersection of restrictedness and boundedness.

In summary, restrictedness (\cref{def:resitrctness}) is a certain posterior property and boundedness is a certain Fourier analytic condition while reducedness is the intersection of the two. By~\cref{def:set and pdf}, we immediately have that each posterior set $B_t$ is $(M_t,\vecc_t)$-restricted and in the lemmas stated below we will establish that $B_t$ is $(M_t,C_0,s)$-reduced and the aggregated posterior set $B_{1:t}$ is $(C_t,s)$-bounded with high probability (for some choices of parameters $s,C_0,C_1,\dots,C_T$). See also~\cref{fig:hybrid} for a pictorial view of these definitions.

\subsubsection{Three key lemmas}

There are three key lemmas about these Fourier analytic conditions. 
The first lemma establishes the ``large'' enough restricted sets are reduced. We typically apply this lemma to the sets $B_t$.

\begin{restatable}[Posterior set]{lemma}{boundednessbasecase}
	\label{lem:boundedness base case}
	For every $q,k\geq2$, there exist constants $\epsilon_0 > 0$ and $C_0 < \infty$ such that 
	for every  sufficiently large $n$, every $k$-hypermatching $M$ on vertex set $[n]$, every pair of integers $b,s$ satisfying $0<b\leq s\leq \epsilon_0 \cdot n$ the following holds. If $B\subseteq\Z_q^n$ satisfies (i) there exists a sequence of centers $\vecc$ such that $B$ is $(M,\vecc)$-restricted, and (ii) $|B| \geq q^{n-b}$, then $B$ is $(M,C_0,s)$-reduced.
\end{restatable}

The proof of~\cref{lem:boundedness base case} is given in \cref{sec:reduced bounded base case}.

Recall from~\autoref{cor:reduce hybrid condition} that the condition (iii) in~\autoref{lem:hybrid} is implied by showing $A_{\vecc}\vecx$ is close to the uniform distribution over $\Z_q^{(k-1)m}$ with high probability over the choice of $A_{\vecc}$ where $\vecx$ is sampled uniformly from the posterior set $B_{1:t}$. The second key lemma shows that $A_{\vecc}\vecx^*$ is indeed close to uniform when the posterior set is bounded.

\begin{restatable}[Boundedness implies (closeness to) uniformity]{lemma}{boundednessuniform}\label{lem:boundedness implies uniform}
	For every $q,k\geq2$ there exists $\alpha_0 = \alpha_0(k,q)$ such that for every $\delta \in (0,1/2)$ and $C< \infty$, there exists $\tau=\tau(q,k,\delta,C)>0$ and $s_0 = s_0(\delta) < \infty$ such that the following holds for every sufficiently large $n$:
	
	Let $B\subset \Z_q^n$ be a $(C,s)$-bounded set with $|B| \geq q^{n-b}$, for $s_0\leq b \leq s\leq \tau n$. Let $M$ be a random $k$-hypermatching of size $m \leq \alpha_0 n$ and $\vecc$ be a uniformly random sequence of centers for $M$ and let $A_{\vecc}$ denote the $\vecc$ centered folded encoding of $M$. Then, with probability at least $1-\delta$ over the choice of $M$ and $\vecc$, for every $\vecz_0\in\Z_q^{(k-1)m}$, we have that
	\[
	1-\delta<q^{(k-1)m}\Pr_{\vecx\sim\Unif(B)}[A_{\vecc}\vecx=\vecz_0] < 1+\delta \, .
	\]
	As a consequence, we also have (with probability at least $1-\delta$ over the choice of $(M,\vecc)$):
	
	\begin{enumerate}
		\item  $\|(A_{\vecc}\vecx)-U\|_{tvd}\leq\delta$ where $\vecx\sim\Unif(B)$ and $U\sim\Unif(\Z_q^{(k-1)m})$.
		\item For every non-negative function $f:\Z_q^{(k-1)m}\to \mathbb{R}^{\geq 0}$, 
		\[
		(1-\delta)\le \frac{\mathbb{E}_{\vecx\sim \Unif(B)}\left[f(A_c\vecx)\right]}{\mathbb{E}_{\vecz\sim \Unif(\Z_q^{(k-1)m})}\left[f(\vecz)\right]} \le (1+\delta) \, .
		\]
	\end{enumerate}

\end{restatable}

The proof of~\cref{lem:boundedness implies uniform} is postponed to~\cref{sec:bounded implies uniform}.

Our final lemma of this section asserts that if $\one_{B_{1:t}}$ is $(C,s)$-bounded, then $f_{1:t+1}$ is $(O(C),s)$-bounded with high probability. 

\begin{restatable}[Induction step]{lemma}{leminductionstep}
\label{lem:induction step}

    For every $q,k\in\N$ there exist $\alpha_0>0$ and $C_0< \infty$ such that for every $C\geq C_0$, and $\delta\in(0,1/2)$, there exist  $\tau_0 = \tau_0(q,k,\delta,C) \in(0,1)$ and $C' = C'(q,k,\delta,C) >0$ 
	such that the following holds.
	For every $n,b,b',s,m\in \N$, satisfying $m \leq \alpha_0 n$, $0<b,b',s<\tau_0 n$ and every $(C,s)$-bounded set $B\subset\Z_q^n$ satisfying $|B| \geq q^{n-b}$, we have that with probability at least $1-4\delta$ over a uniformly random $k$-hypermatching $M$ of size $m$ and every $(M,C_0,s)$-reduced set $B'\subset\Z_q^n$ satisfying $|B'|\geq q^{n-b'}$ and $|B\cap B'|\geq(1-\delta)\cdot|B|\cdot|B'|/q^n \geq q^{n-s}$, we have $B\cap B'$ is $(C',s)$-bounded.

\end{restatable}

\cref{lem:induction step} is proved in \cref{sec:induction}. In our inductive application of the lemma above,
we  set $B\gets B_{1:t-1}$ and $B'\gets B_t$ for every $t\in\{2,3,\dots,T\}$ to get that all the $B_t$'s are bounded and this is the core of the proof of \cref{lem:hybrid}.

\newpage 

\subsection{Proof of Lemma~\ref{lem:hybrid}}\label{sec:hybrid}

	\begin{figure}[ht]
		\centering
		\includegraphics[width=15cm]{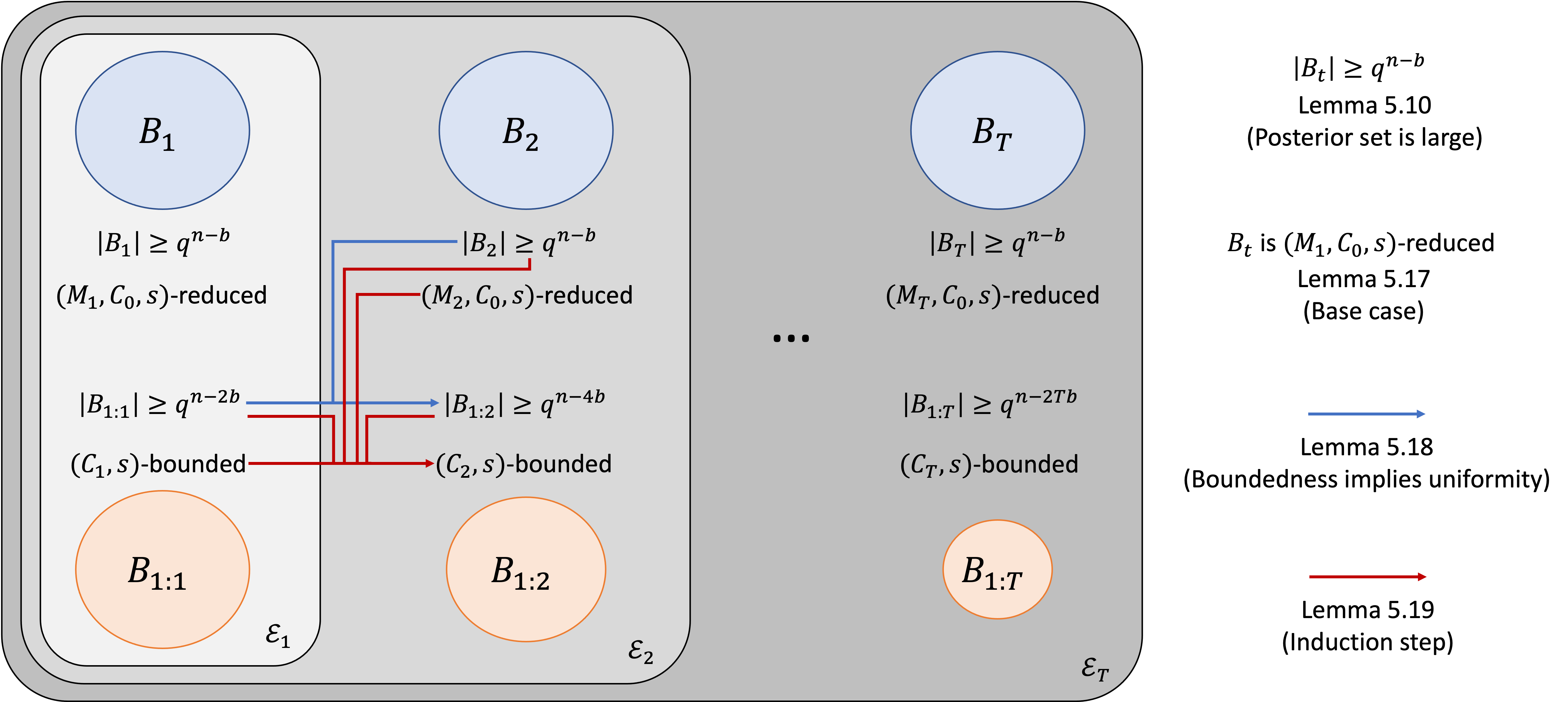}
		\caption{A pictorial overview of the proof of~\autoref{lem:hybrid}. Each posterior set $B_t$ (the blue sets) is both $(M_t,\vecc_t)$-restricted (followed from~\cref{def:set and pdf}) and $(M_t,C_0,s)$-reduced (followed from~\cref{lem:boundedness base case}). Each aggregated posterior set $B_{1:t}$ (the orange sets) is $(C_t,s)$-bounded (followed from~\cref{lem:induction step}).}
		\label{fig:hybrid}
	\end{figure}
\hybrid*

\begin{proof}[Proof of~\cref{lem:hybrid}]~\\


     \paragraph{Overview of proof:}  (See~\cref{fig:hybrid} for a pictorial overview of the proof.) The rough overview of the proof is to show that for an appropriate choice of the constants $C_1,\ldots,C_T$, for every $t \in [T]$, the posterior set $B_{1:t}$ is $(C_t,s)$-bounded. Once we have this, we can apply ``boundedness implies uniformity'' lemma (\cref{lem:boundedness implies uniform}) to conclude that the messages sent by the $t$th player on the YES and NO distributions are indistinguishable. To show the boundedness condition for $B_{1:t}$ we use induction to deduce that $B_{1:t-1}$ is bounded, and then reason about $B_t$ to conclude that it is large, $(M,C_0,s)$-reduced (for appropriate $C_0$), and crucially that it is roughly independent of $B_{1:t-1}$. 
     Proving the above involves, among other reasoning, another application of the boundedness implies uniformity lemma (on $B_{1:t-1}$). With these ingredients in place the induction step lemma (\cref{lem:induction step}) yields the boundedness of $B_{1:t}$. We give the details below.

       
    \paragraph{Setting of parameters:} We note that in addition to the parameters $\alpha_0$, $\tau$,  $n_0$, $\{\delta_t\}_{t\in [T]}$ required by the lemma statement, we also need to specify the constants $\{C_t\}_{t \in [T]}$ alluded to in the overview. Additionally we also specify three integer parameters: $s^*$ which specifies the length of the message, $b$ which quantifies largeness of various posterior sets, and $s$ which quantifies the boundedness of posterior functions. 
 
    Given $q$ and $k$, let $\alpha_{0,1}$ be the $\alpha_0(k,q)$ from \cref{lem:boundedness implies uniform} and $\alpha_{0,2}$ be the $\alpha_0(k,q)$ from \cref{lem:induction step}. We set $\alpha_0 = \min\{\alpha_{0,1},\alpha_{0,2}\}$. 
    Let $\epsilon_0 = \epsilon_0(k,q)$ and $C_0 = C_0(k,q)$ be the constants from \cref{lem:boundedness base case}. 
    Now given $T$ and $\delta$, we need to specify $\tau > 0$ and $n_0 < \infty$. We first set a large number of intermediate parameters that will be used in the rest of the proof. Recall that \cref{lem:induction step} takes as input parameters $q,k,\delta$ and $C \geq C_0$ and gives constants $C' = C'(q,k,\delta,C)$ and $\tau_0(q,k,\delta,C)$ for which the lemma holds. 
    We let $\delta' = \delta/(12\cdot 2^T)$ and $\delta_1 =6 \delta'$ and $\delta_{t+1} = 2\delta_t$ for $1 \leq t \leq T-1$. (Note these settings satisfy $\sum_{t=1}^T \delta_t \leq \delta/2$, as required in the conclusion of the lemma, and $\delta_{t+1} \geq \sum_{i=1}^t \delta_i + 6\delta'$ as required in the proof below.) 
    For $t \in [T]$ we set $C_t = C'(q,k,\delta',C_{t-1})$ where $C'(\cdots)$ is the aforementioned function from \cref{lem:induction step}. 
    Next for every $t \in [T]$ we set $\gamma_t = \tau_0(q,k,\delta',C_t)$. 
    Further, let $\tau(q,k,\delta,C)$ and $s_0(\delta)$ be the functions from \cref{lem:boundedness implies uniform}. For $t \in [T]$, let $\rho_t = \tau(q,k,\delta',C_t)$ and $s_{0} = s_0(\delta')$.  Let $\zeta = \min\{\epsilon_0, \min_{t \in [T]} \{\gamma_t\}, \min_{t \in [T]}\{\rho_t\}\}$. Let $\nu = \zeta/(2T)$ and let $\tau = \nu/2$. 

    Finally we let $n_0 = \max\{\frac2\nu\log_q(1/\delta'), s_0/\nu\}$.
    \footnote{The reader may notice that several of the terms in the parameter settings obviously dominate the others and we could skip the mins and maxes thereby simplifying the expressions. But we keep them separate for easier verifiability in the proof. We follow this practice through most of this paper.}  Finally, given $n\geq n_0$ we set $s = \zeta n$, $b = \nu n$ and $s^* = \tau n$.

    Note that these settings ensure $b\geq 1$, $s^* \leq b/2 \leq b-\log_q(1/\delta')$, $2tb \leq s \leq \epsilon_0 n$ for every $t \in [T]$, $s \leq \gamma_t n = \tau_0(q,k,\delta',C_{t}) n$ for every $t \in [T]$ and $s_0 \leq b \leq 2tb \leq s \leq \rho_t n = \tau(q,k,\delta',C_t)n$ for every $t\in [T]$.  These inequalities will be used in the proof below.

    \paragraph{The events $\{\cE_t\}_{t\in [T]}$:}
    Recall the notion of posterior sets $B_t$ and aggregate posterior sets $B_{1:t}$ for $t \in [T]$ from \cref{def:set and pdf}. 
    Let $B_0 = \Z_q^n$. 
    We define $\cE_1$ to be the event that $B_{0}$ is $(C_0,s)$-bounded and large i.e., $|B_0|\geq q^{n}$. 
    For $2\le t \le T$, let $\cE^1_t$ denote the event that $B_{t-1}$ is $(M_{t-1},C_0,s)$-reduced and the aggregated posterior set $B_{1:t-1}$ is $(C_{t-1},s)$-bounded and large i.e., $|B_{1:t-1}|\geq q^{n-2(t-1)b}$. 
    We also define $\cE^2_t$ to be the event that $\|A_{t,\vecc_t}\vecx^*-U\|_{tvd}\leq\delta_{t}$, where $\vecx^*\sim\Unif(B_{1:{t-1}})$ and $U\sim\Unif(\Z_q^{(k-1)\alpha n})$. Finally, let $\cE_t = \cE_{t-1} \cap \cE^1_t \cap \cE^2_t$. We now turn to proving conditions (i)-(iii) hold for this choice of events.

    \paragraph{Causality and Indistinguishability:} It is immediate from the definition that $\cE_{t}\implies \cE_{t-1}$ and $\cE_t$ only depends on $A_{1:t},\vecc_{1:t}$, and $S^Y_{1:t-1}$. 
    This establishes condition (i). 
    Next we note that conditioned on $\cE_t$ we have condition (iii). 
    In particular, by the definition of $\cE_t$, we have $\|(A_{t,\vecc_t}\vecx^*)-U\|_{tvd}\leq\delta_{t}$ where $\vecx^*\sim\Unif(B_{1:t-1})$ and $U\sim\Unif(\Z_q^{(k-1)\alpha n})$. As $S_t^Y=r_t(A_{1:t},\vecc_{1:t},S^Y_{1:t-1},A_{t,\vecc_t}\vecx^*)$ where $\vecx^*\sim\Unif(B_{1:t-1})$ (by \cref{lem:conditional prob}), by the data processing inequality we have 
    $$\|S^Y_t-r_t(A_{1:t},\vecc_{1:t},S^Y_{1:t-1},U)\|_{tvd}\leq \| A_{t,c_t}\vecx^* - U\| \leq \delta_{t}$$ 
    as desired for condition (iii). It remains to prove (ii) which we do from now on.

    \paragraph{Probability of bad events:} 
    By definition of $B_0 = \Z_q^n$, we have it is large (i.e., $|B_0| \geq q^n$). Further by Part (2) of \cref{prop:basic-boundedness} we have that $B_0$ is $(C_0,s)$-bounded. Thus it follows that $\cE_1$ holds with probability $1$. 
    We now analyze $\cE^1_{t+1}$ for $t \ge 1$ (and then turn to $\cE^2_{t+1}$).
    We show that, conditioned on $\cE_{t}$, $\cE^1_{t+1}$ holds with probability at least $1 - 5\delta'-\sum_{i=1}^t \delta_i$.  The main part of it is proving that $B_{1:t}$ is large, which we do in the claim below. (Proving boundedness is then a straightforward application of \cref{lem:induction step}, as we show later.)

	\begin{claim}\label{clm:hybrid}
        Let $2\leq t \leq T$. Let $M_{t}$ and $\vecc_{t}$ be chosen uniformly conditioned on $\cE_{t}$. Then with probability at least $(1-\delta'-\sum_{i=1}^t \delta_i)$ the posterior set $B_{1:t}$ satisfies
		\[|B_{1:t}|\geq(1-\delta')\cdot|B_{1:t-1}|\cdot|B_{t}|/q^n\, .\]
	\end{claim}
	\begin{proof}
    Fix some $B_{1:t-1}$ that is $(C_{t-1},s)$-bounded and satisfies $|B_{1:t-1}| \geq q^{n-2(t-1)b}$ and consider a uniform choice of $A_t$ and $\vecc_t$. 
    We now apply \cref{lem:boundedness implies uniform}. Note this lemma takes four ``parameters'' $B,b,\delta$ and $C$. We apply the lemma with $(B,b,\delta,C)_{\mathrm{\cref{lem:boundedness implies uniform}}}= (B_{1:t-1},2(t-1)b,\delta',C_{t-1})$.
    Note that the parameter settings ensure $|B_{1:t-1}| \geq q^{n-2(t-1)b}$, $s_0 \leq 2(t-1)b \leq s \leq \tau(q,k,\delta',C_{t-1}) n$ and so the preconditions of \cref{lem:boundedness implies uniform} are satisfied. By Part(2) of the lemma we get: 
    	\begin{equation}\label{eqn:close_to_uniform}
			(1-\delta')\le \frac{\mathbb{E}_{x\sim \Unif(B_{1:t-1})}\left[f(A_{t,\vecc_t}x)\right]}{\mathbb{E}_{z\sim \Unif(\Z_q^{(k-1)\alpha n})}\left[f(z)\right]} \le (1+\delta')\,
	   \end{equation}
	for every non-negative function $f$ over $\Z_q^{(k-1)\alpha n}$, with probability at least $1 - \delta'$  (over a \emph{uniform} choice of $M_t$ and $\vecc_t$).
	Setting $f$ to be the indicator function of $B_{r,t}$ (recall that $B_{r,t}$ is the ``reduced posterior set'' from \cref{def:set and pdf}) and applying~\cref{eqn:close_to_uniform}, we have 
	\[ \mathbb{E}_{x\sim \Unif(B_{1:t-1})}\left[\one_{B_{r,t}}(A_{t,\vecc_t}x)\right] = \frac{|B_{1:t}|}{|B_{1:t-1}|}\, ,\]and
	\[ \mathbb{E}_{z\sim \Unif(\Z_q^{(k-1)\alpha n})}\left[\one_{B_{r,t}}(z)\right] = \frac{|B_{r,t}|}{q^{(k-1)\alpha n}} = \frac{|B_t|}{q^n}\, .\]
		
	We have
		\begin{align*}
			\frac{|B_{1:t}|}{q^n} &= \frac{|B_{1:{t-1}}|}{q^n} \cdot  \mathbb{E}_{x\sim \Unif(B_{1:t-1})}\left[\one_{B_{r,t}}(A_{t,\vecc_t}x)\right] \\
			&\ge (1-\delta') \frac{|B_{1:{t-1}}|}{q^n} \cdot \mathbb{E}_{z\sim \Unif(\Z_q^{(k-1)\alpha n})}\left[\one_{B_{r,t}}(z)\right]\\
			&= (1-\delta') \frac{|B_{1:{t-1}}|}{q^n} \cdot \frac{|B_t|}{q^n} \, .
		\end{align*}
    We conclude that for every $B_{1:t-1}$, with probability at least $1-\delta'$ over a uniform choice of $A_t$ and $\vecc_t$, we have that if $B_{1:t-1}$ is $(C_{t-1},s)$-bounded and satisfies $|B_{1:t-1}| \geq q^{n-2(t-1)b}$, then
   \[ \frac{|B_{1:t}|}{q^n} \geq (1-\delta') \frac{|B_{1:{t-1}}|}{q^n} \cdot \frac{|B_t|}{q^n} \, .\]

    Now we condition on the event $\cE_t$. Doing so, alters the distribution of $(A_t,\vecc_t)$ (since $\cE_t$ depends on $A_t,\cE_t$) but the total variation distance is bounded by $\Pr[\overline{ \cE_t}] \leq \sum_{i=1}^t \delta_i$. We thus have that for every $B_{1:t-1}$, with probability at least $1-\delta'-\sum_{i=1}^t \delta_i$ over choice of $A_t$ and $\vecc_t$ conditioned on $\cE_t$, we have that if $B_{1:t-1}$ is $(C_{t-1},s)$-bounded and satisfies $|B_{1:t-1}| \geq q^{n-2(t-1)b}$
   \[ \frac{|B_{1:t}|}{q^n} \geq (1-\delta') \frac{|B_{1:{t-1}}|}{q^n} \cdot \frac{|B_t|}{q^n} \, .\]
   But finally note that $\cE_t$ implies $B_{1:t-1}$ is $(C_{t-1},s)$-bounded and satisfies $|B_{1:t-1}| \geq q^{n-2(t-1)b}$, and so we simply get that with probability at least $1-\delta'-\sum_{i=1}^t \delta_i$ we have, over the choice of $B_{1:t-1}, A_t$ and $\vecc_t$ conditioned on $\cE_t$,  
   \[ \frac{|B_{1:t}|}{q^n} \geq (1-\delta') \frac{|B_{1:{t-1}}|}{q^n} \cdot \frac{|B_t|}{q^n} \, .\qedhere\]
	\end{proof}
     
    To use the claim above we now analyze $|B_t|$. By the ``posterior set is large'' lemma (i.e.,~\cref{lem:posterior set is large}) we have $|B_t|\geq q^{n-b}$ (again using $s \leq b - \log_q(1/\delta')$). When $B_t$ is large, then by \cref{lem:boundedness base case} we have $B_t$ is $(M_t,C_0,s)$-reduced using $b \leq s\leq \epsilon_0 n$. Furthermore if $|B_t| \geq q^{n-b}$, then combined with the inductive bound that $|B_{1:t-1}| \geq q^{n-2(t-1)b}$ (implied by $\cE_{t}$), \cref{clm:hybrid} implies 
    \[|B_{1:t}|=|B_{1:t-1}\cap B_t|\geq(1-\delta')\cdot|B_{1:t-1}|\cdot|B_t|/q^n \geq (1-\delta')q^{n-(2t-2)b}q^{-b} \geq q^{n-2tb}\, ,\]
    where the final inequality uses the very crude (but true) inequality $1 - \delta' \geq \frac12 \geq q^{-b}$.

    Conditioned on $B_t$ being large and reduced, we can finally invoke the ``induction step'' lemma (i.e.,~\cref{lem:induction step}) with $B = B_{1:t-1}$ and $B' = B_t$ with $C=C_{t-1}$ to get that $B_{1:t}=B_{1:t-1}\cap B_t$ is $(C_t,s)$-bounded with probability at least $1-4\delta'$ where we use $C_t = C'(q,k,\delta',C_{t-1})$. We note this application requires $\max\{b,2(t-1)b,s\} \leq \gamma_{t-1} n := \tau_0(q,k,\delta',C_{t-1})n$ which is ensured by our setting of parameters.
    
    Taking the union bound over the three error events, namely (a) $B_t$ not being large, (b) $B_{1:t}$ not being large conditioned on $B_t$ being large and (c) $B_{1:t}$ not being bounded condition on being large, we get that $\cE^1_{t+1}$ holds with probability at least $1 - 5\delta'-\sum_{i=1}^t \delta_i$ conditioned on $\cE_{t}$. 

    Finally we turn to bounding $\cE^2_{t+1}$ conditioned on $\cE^1_{t+1}$. Since $|B_{1:t}| \geq q^{n-2tb}$ and $B_{1:t}$ is $(C_t,s)$-bounded, we   apply \cref{lem:boundedness implies uniform} to analyze $\|(A_{t+1,\vecc_{t+1}}\vecx^*)-U\|_{tvd}$, for $\vecx^*\sim \Unif(B_{1:t})$. The application  requires $s_0  \leq 2(t+1)b \leq s \leq \rho_{t+1} n := \tau(q,k,\delta',C_{t+1})n$ which we do have with our setting of parameters. 
    We conclude that $\|(A_{t+1,\vecc_{t+1}}\vecx^*)-U\|_{tvd}\leq\delta'$ with probability at least $1 - \delta'$ for every fixing of $A_{1:t,\vecc_{1,t}}$ and $S_{1:t-1}^Y$ (over the choice of $M_{t+1}$ and $\vecc_{t+1}$). We thus have that $\cE^2_{t+1}$ holds with probability at least $1 - \delta'$ conditioned on $\cE^1_{t+1}$. 

    Putting the two together we get for every $t \in [T]$ (including $t=1$), we have $\cE_t$ holds with probability at least $1 - 6\delta' - \sum_{i=1}^t \delta_i \geq 1 - \delta_{t+1}$ as required for condition (ii). 

    This completes the proof of~\autoref{lem:hybrid}. 
\end{proof}

\section{Analysis of bounded functions}\label{sec:reduced bounded}

In this section, we prove three important lemmas from~\cref{sec:proof overview}: the ``posterior set'' lemma (i.e.,~\cref{lem:boundedness base case}), the ``boundedness implies uniformity'' lemma (i.e.,~\cref{lem:boundedness implies uniform}), and the ``induction step'' lemma (i.e.,~\cref{lem:induction step}).
We first establish useful structure on the Fourier coefficients of restricted sets (posterior set is a special case of restricted set) in~\autoref{sec:fourier coeff posterior function}. Next, we prove useful properties for the Fourier analytic conditions in~\autoref{sec:reduced bounded definition}. Finally, we prove the three lemmas in~\cref{sec:reduced bounded base case},~\cref{sec:bounded implies uniform}, and~\cref{sec:induction} respectively.

\subsection{Fourier coefficients of the posterior function}\label{sec:fourier coeff posterior function}
Given a $k$-hypermatching $M = (e_1,\ldots,e_m)$ and centers $\vecc = (c_1,\ldots,c_m)$ we say that a set $B \subseteq \Z_q^n$ is $(M,\vecc)$-restricted if there exists a (``reduced'') set $B_r \subseteq \Z_q^{(k-1)m}$ such that $B = \{\vecx \in \Z_q^n | A_{\vecc}\vecx \in B_r\}$, where $A_\vecc$ is the $\vecc$-centered folded encoding of $M$. In this section we aim to prove that large restricted sets are bounded.
Recall that given a $k$-hypermatching $M = (e_1,\ldots,e_{m})$ on vertex set $[n]$ with $m=\alpha n$ edges and sequence of centers $\vecc = (c_1,\ldots,c_{m})$ with $c_i \in e_i \subseteq [n]$, the $\vecc$-centered folded representation of $M$ was denoted $A_{\vecc} \in \Z_q^{(k-1)m \times n}$. Recall that a set $B \subseteq \Z_q^n$ is said to be $(M,\vecc)$-restricted if there exists a (``reduced'') set $B_r \subseteq \Z_q^{(k-1)m}$ such that $B = \{\vecx \in \Z_q^n | A_{\vecc}\vecx \in B_r\}$. For our next lemma we will also need a variant of this matrix named the $\vecc$-centered projection induced by  $M$, which we denote $\widetilde{A}_\vecc \in \Z_q^{(k-1)m \times n}$, which is simply the matrix $A_\vecc$ with the columns corresponding to $c_1,\ldots,c_{m}$ zeroed out. (In $A_\vecc$ each of these columns has $(k-1)$ $-1$'s. See \autoref{fig:folded matrices p}.) With this definition in place we can now relate the Fourier coefficients of the indicator of a restricted set to its image.

\begin{figure}[ht]
	\centering
	\includegraphics[width=12cm]{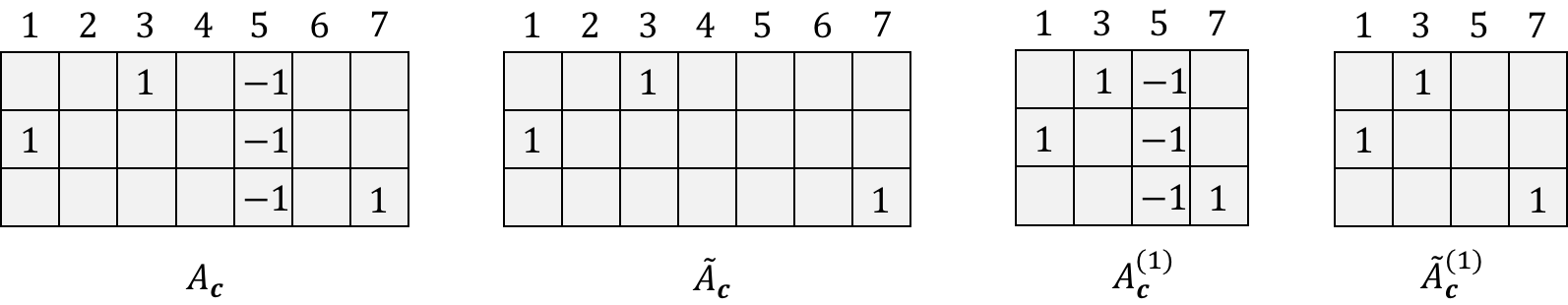}
	\caption{An example of $A_\vecc,\tilde{A}_\vecc,A^{(1)}_\vecc,\tilde{A}^{(1)}_\vecc$ with $m=1$, $n=7$, $k=4$, $e_1=\{1,3,5,7\}$ and $c_1=\{5\}$.}
	\label{fig:folded matrices p}
\end{figure}
Recall that we use $\vece_i \in \Z_q^n$ to denote the indicator vector of hyperedge $e_i$ (see \autoref{sec:comm}).

\begin{lemma}[Fourier coefficients of the posterior function]\label{lem:fourier coeff of set}
	Let $M$ be a $k$-hypermatching of size $m$ and $\vecc$ be a sequence of centers. Let $A_\vecc$ be the folded representation of $M$ and $\widetilde{A}_\vecc$ be the  projection induced by $M$. Furthermore, let $B \subseteq \Z_q^n$ be an $(M,\vecc)$-restricted set with $B_r \subseteq \Z_q^{(k-1)m}$ satisfying $B = \{\vecx \in \Z_q^n | A_\vecc \vecx \in B_r\}$. Let $\one_B$ denote the indicator function of $B$. Then for every $\vecu \in \Z_q^n$ we have: 
	\[
	\widehat{\one_B}(\vecu)=\begin{cases}
		0     , \quad &\text{if $\vecu$ contains a node not in $M$.}\\
		0     ,  \quad &\text{if $\exists i \in [m]$ such that $\langle \vecu, \vece_i \rangle \not\equiv 0 \pmod{q}$}\\
		\widehat{\mathbf{1}}_{B_r}(\widetilde{A}_\vecc  \vecu), \quad &\text{otherwise.}
	\end{cases}.
	\]
\end{lemma}
\begin{proof}
	From the definition of the Fourier coefficient we have $\widehat{\one_B}(\vecu) = \tfrac{1}{q^n}\sum_{\vecx \in \Z_q^n} \one_B(\vecx) \omega^{\vecu^\top \vecx}$ where $\omega=e^{2\pi i/q}$ being the primitive $q$-th root of unity. Using the fact that $B$ is restricted, we get 
	\[
	\widehat{\one_B}(\vecu)= \frac{1}{q^n} \sum_{\vecx \in \Z_q^n} \one_B(\vecx) \cdot\omega^{\vecu^\top \vecx}
	= \frac{1}{q^{n}} \sum_{\substack{\vecz\in \Z_q^{(k-1)m}}}\sum_{\substack{\vecx\in\Z_q^n\\A_\vecc \vecx = \vecz}} \one_{B_r}(\vecz) \cdot \omega^{\vecu^\top \vecx}
	= \frac{1}{q^{n}} \sum_{\substack{\vecz\in B_r}}\sum_{\substack{\vecx\in\Z_q^n\\A_\vecc \vecx = \vecz}} \omega^{\vecu^\top \vecx}.
	\]
	
	We now fix $\vecz \in B_r$ and explore the final term $\sum_{{\vecx\in\Z_q^n, A_\vecc \vecx = \vecz}} \omega^{\vecu^\top \vecx}$. Let $\vecz = (\vecz^{(1)},\ldots,\vecz^{(m)})$ where $\vecz^{(i)} \in \Z_q^{k-1}$. Also let $A_\vecc^{(1)},\ldots,A_\vecc^{(m)} \in \Z_q^{(k-1)\times n}$ denote the blocks of $A_\vecc$ corresponding to the $m$ edges. Now, think of $\Z_q^n$ as a free module over $\Z_q$ and consider the direct sum decomposition $\Z_q^n = W^{(0)} \oplus \cdots \oplus W^{(m)}$ where for $i \in [m]$,  $W^{(i)}$ is the sub-module of $\Z_q^n$ generated by $e_i$, $W^{(0)}$ is the sub-module generated by $[n]-(\cup_i e_i)$, and ``$\oplus$'' denotes the direct sum of modules. Let us write $\vecx = \vecx^{(0)} + \cdots + \vecx^{(m)}$ where for $i \in \{0,\ldots,m\}$, $\vecx^{(i)} \in W^{(i)}$.
	Similarly write $\vecu = \vecu^{(0)} + \cdots + \vecu^{(m)}$. 
	Since $(\vecu^{(i)})^\top \vecx^{(j)} = 0$ if $i \ne j$ we have $\vecu^\top \vecx = \sum_{i=0}^m (\vecu^{(i)})^\top\vecx^{(i)}$.
	Note also that $\vecz = A_\vecc \vecx$ if and only if $\vecz^{(i)} = A_\vecc^{(i)} \vecx^{(i)}$ for every $i \in [m]$. Using this notation, we have 
	\[
	\sum_{\substack{\vecx\in\Z_q^n\\A_\vecc \vecx = \vecz}} \omega^{\vecu^\top \vecx} = \left(\sum_{\vecx^{(0)}\in W^{(0)}} \omega^{(\vecu^{(0)})^\top \vecx^{(0)}}\right)\cdot \prod_{i=1}^m 
	\left(\sum_{\substack{\vecx^{(i)}\in W^{(i)} \\ A_\vecc^{(i)}\vecx^{(i)} = \vecz^{(i)}}} \omega^{(\vecu^{(i)})^\top \vecx^{(i)}}\right) \, . \]
	Now note that if $\vecu^{(0)} = 0$ then the first term is $|W^{(0)}| = q^{n-km}$, else it is zero. Similarly for $i\in [m]$, there are exactly $q$ vectors $\vecx^{(i)}\in W^{(i)}$ such that $A_\vecc^{(i)} \vecx^{(i)} = \vecz^{(i)}$ (which are additive shifts of each other on coordinates in $e_i$). Concretely, these two solutions are of the form $(\tilde{A}_\vecc^{(i)})^\top\vecz^{(i)}+a^k$ for some $a\in\Z_q$. So we have
	\[
	\sum_{\vecx^{(i)}\in W^{(i)}:A_\vecc^{(i)}\vecx^{(i)} = \vecz^{(i)}} \omega^{(\vecu^{(i)})^\top \vecx^{(i)}}=\sum_{ a\in\Z_q}\omega^{(\vecu^{(i)})^\top (\tilde{A}_\vecc^{(i)})^\top\vecz^{(i)}+(\vecu^{(i)})^\top a^k}=\omega^{(\vecu^{(i)})^\top (\tilde{A}_\vecc^{(i)})^\top\vecz^{(i)}}\sum_{ a\in\Z_q}\omega^{a\cdot\|\vecu^{(i)}\|_1} \, .
	\]
	Moreover,
	\[
	\sum_{ a\in\Z_q}\omega^{a\cdot\|\vecu^{(i)}\|_1}  = \begin{cases}
		0     &  \text{if }\|\vecu^{(i)}\|_1\not\equiv0\bmod{q}\\
		q     & \text{otherwise}.
	\end{cases}.
	\]

	Putting all the above together we get 
	\[
	\sum_{{\vecx\in\Z_q^n\\A_\vecc \vecx = \vecz}} \omega^{\vecu^\top \vecx}=\begin{cases}
		0     \quad &\text{if $\vecu$ contains a node not in $M$.}\\
		0     \quad &\text{if $\exists i \in [m]$ such that $\langle \vecu, \vece_i \rangle \not\equiv 0 \bmod{q}$}\\
		q^{n-(k-1)m}\cdot\omega^{(\widetilde{A}_\vecc  \vecu)^\top \vecz} \quad &\text{otherwise.}
	\end{cases}.
	\]
	Summing up over all $\vecz \in \Z_q^{(k-1)m}$ and normalizing yields the lemma.

\end{proof}

\subsection{Basic properties of large weakly-bounded sets}\label{sec:reduced bounded definition}

In this section, we relate weakly bounded sets to strongly bounded ones, and also show that the notion of restricted-ness of sets is independent of the choice of centers. These help us prove boundedness in the base case.

We start by stating an immediate consequence of Parseval's lemma applied to our indicator functions.

\begin{lemma}\label{lem:parseval}
	For every $B \subseteq \Z_q^n$ we have $\sum_{\vecv \in \Z_q^n} \widehat{\mathbf{1}_B}(\vecv)^2 \leq |B|/q^n$.
\end{lemma}

(\cref{lem:parseval} follows from \cref{prop:parseval} by noticing that $\sum_{\veca \in \Z_q^n} \one_{B}(\veca) = |B|$.)
Recall that the $(C,s)$-bounded criterion bounds the sum of Fourier coefficients with a fixed weight at most $s$. As we also need to bound the sum of Fourier coefficients of high weight, this can be guaranteed from Parseval's inequality as shown in the following lemma.
\begin{lemma} \label{lem:highregime}
	Suppose $B \subseteq \Z_q^n$ satisfies $|B| \geq q^{n-b}$ for some $b\in\N$. Then, for every $\vecv\in\Z_q^n$ and $b < h \leq n$, we have
	\[
	\sum_{\substack{\vecu\in\Z_q^n\\\|\vecu+\vecv\|_0=h}} \frac{q^n}{|B|}\left|\widehat{\one_B}(\vecu)\right| \leq \left(\frac{2q^2e^2 n}{h}\right)^{h/2}.
	\]
\end{lemma}
\begin{proof}
	From \cref{lem:parseval}, we have that
	$\sum_{\|\vecu+\vecv\|_0=h}\left|\widehat{\mathbf{1}_B}(\vecu)\right|^2 \leq |B|/q^n.$
	Using 
	\[
	\left|\left\{\vecu\in\Z_q^n \, |\, \|\vecu+\vecv\|_0=h\right\}\right|\leq (q-1)^h\cdot\binom{n}{h}\leq\left(\frac{qen}{h}\right)^h
	\]
	and the Cauchy-Schwarz inequality we get that
	\begin{align*} 
		\sum_{\substack{\vecu\in\Z_q^n\\\|\vecu+\vecv\|_0=h}}\frac{q^n}{|B|}\left|\widehat{\one_B}(\vecu)\right| 
		& \leq  \frac{q^n}{|B|} \sqrt{(|B|/q^n) \cdot (qen/h)^h} \\
		& =  \sqrt{(q^n/|B|) \cdot (qen/h)^h} \\
		& \leq \sqrt{q^b  (qen/h)^h} & (\because |B|\geq q^{n-b}) \\
		& \leq \sqrt{q^h  (qen/h)^h} & (\because h>b)\\
		& \leq (2q^2e^2 n/h)^{h/2}.
	\end{align*}

\end{proof}

Next, we note a basic monotonicity property of the notion of boundedness which will be useful in the future. Recall the function $W_{C,s}(\cdot)$ used in the notion of  {\em weakly-bounded} sets from \cref{def:boundedness}.

\begin{lemma}[Monotonicity of boundedness]\label{lem:ucs-monotonic} The following monotonicities hold for $W$ and $U$.
\begin{enumerate}
    \item If $h \leq s \leq s'$ then $W_{C,s}(h) \leq W_{C,s'}(h)$. 
    \item If $s \leq s'$ then for every $h$, $U_{C,s}(h) \leq U_{C,s'}(h)$. Consequently, if a set $B$ is $(C,s)$(-strongly)-bounded then it is also $(C,s')$(-strongly)-bounded.
    \item 	If $C > e$, then $W_{C,s}(h)$ and $U_{C,s}(h)$ are monotonically increasing in $h \in [1,s]$.
\end{enumerate}
\end{lemma}
\begin{proof} The first two monotonicities are definitional, whereas the third one requires some calculations. Details are provided below.
\begin{enumerate}
    \item The inequality holds trivially for $h=0$. For $1\leq h \leq s \leq s'$, we have that $W_{C,s}(h) = (C\sqrt{sn}/h)^{h/2} \leq  (C\sqrt{s'n}/h)^{h/2} = W_{C,s'}(h)$. 
    \item Here we consider three possible ranges for $h$. For $h \leq s \leq s'$, $U_{C,s}(h) = W_{C,s}(h)$ and $U_{C,s'}(h) = W_{C,s'}(h)$ and the inequality follows from Item (1). For $h > s'$, we have $U_{C,s}(h)  = \min\{W_{C,h}(h),(2q^2e^2n/h)^{h/2}\}= U_{C,s'}(h)$ yielding the desired inequality as an equality.
    For $s < h \leq s'$, we have $U_{C,s}(h)  = \min\{W_{C,h}(h),(2q^2e^2n/h)^{h/2}\} \leq W_{C,h}(h) \leq W_{C,s'}(h) = U_{C,s'}(h)$, where the second inequality again follows from the Item (1). Thus in all cases we have $U_{C,s}(h) \leq U_{C,s'}(h)$ and so if a set $B$ is $(C,s)$-(strongly-)bounded then it is also $(C,s')$-(strongly-)bounded.
    \item Recall that the function $f(x) = x^{1/x}$ is decreasing in the interval $(e,\infty)$  (since $f'(x) = x^{1/x} \cdot \frac{1-\ln x}{x^2}$ is negative for $x > e$). Note that for $h \in [1,s]$, $W_{C,s}(h) = U_{C,s}(h) = x^{\frac{C\sqrt{sn}}{2x}}$ for $x = \frac{C\sqrt{sn}}{h}$. Note that $x$ is a strictly decreasing function of $h$. Moreover, for $h\leq s$, we have $x \geq \frac{C\sqrt{sn}}{s} \geq C > e$. Hence, it follows from monotonically decreasing property of $f$ that $W_{C,s}(h)$ and $U_{C,s}(h)$ are monotonically increasing in the described interval, as desired.
\end{enumerate}
\end{proof}

We now show how weak boundedness of a large set in an entire regime of $s$  implies it is strongly bounded.

\begin{lemma}[From weak-boundedness to strong-boundedness]\label{lem:weak-vs-strong}
For every $q, C$, and $\epsilon_0$ there exists $C'$ s.t. for all $n$ and $s \leq \epsilon_0 n$ the following holds:
If $B\subseteq\Z_q^n$ with $|B| \geq q^{n-s}$ is $(C,s')$-weakly-bounded for every $s \leq s' \leq \epsilon_0 n$ then $B$ is $(C',s)$-strongly bounded. Similarly, if $B\subseteq\Z_q^n$ with $|B| \geq q^{n-s}$ is $(C,s')$-weakly-reduced for every $s \leq s' \leq \epsilon_0 n$ then $B$ is $(C',s)$-strongly reduced.
\end{lemma}

\begin{proof}

    We prove the lemma for $C' := \max\{C, 2q^2e^2/\epsilon_0^{1/2}\}$. We prove the reducedness condition (and the boundedness follows similarly). Fix $B \subseteq \Z_q^n$ with $|B| \geq q^{n-s}$ and $\vecv \in \Z_q^n$. Let 
    \[
	\wt(h) := \sum_{\substack{\vecu\in\Z_q^n\\\|\vecu+\vecv\|_0=h}} \frac{q^n}{|B|}\left|\widehat{\one_B}(\vecu)\right|,
	\]
	and let $\tilde{W}(h) = (2q^2e^2n/h)^{h/2}$. 
    Our goal is to prove that for $h \leq s$, $\wt(h) \leq U_{C',s}(h) = W_{C',s}(h)$, and for $h > s$ that $\wt(h) \leq \min\{W_{C',h}(h),\tilde{W}(h)\}$. For $h \leq s$, by the fact that $B$ is  $(C,s)$-weakly-reduced, we have $\wt(h) \leq W_{C,s}(h) \leq W_{C',s}(h)$ where the second inequality follows from the definition of $W_{C,s}$ which is monotone in $C$. For every $h > s$ we have that $\wt(h) \leq \tilde{W}(h)$ by \cref{lem:highregime} and so it suffices to prove that $\wt(h) \leq W_{C',h}(h)$ for every $h > s$. For $h \leq \epsilon_0 n$ we use that $B$ is $(C,s')$-weakly-reduced for $s' = h$ (this is ok since $s \leq s'=h \leq \epsilon_0 n$) to conclude that $\wt(h) \leq W_{C,h}(h) \leq W_{C',h}(s)$. For $h > \epsilon_0 n$ we note that  

    \[
    W_{C',h}(h) 
    = (C'\sqrt{hn}/h)^{h/2} 
    =(C'\sqrt{n/h})^{h/2} 
    \geq (2q^2e^2\sqrt{n/\eps_0 h})^{h/2}
    \geq (2q^2e^2n/h)^{h/2} 
    = \tilde{W}(h) \, ,
    \]
    where the first inequality uses $C' \geq (2q^2e^2/\epsilon_0^{1/2})$ and the next inequality uses  $h>\epsilon_0n$.

    Thus in this case we have $\wt(h) \leq \tilde{W}(h) \leq W_{C',h}(h)$ as desired, concluding the proof that $B$ is $(C',s)$-strongly-reduced.
\end{proof}

Finally we show that the notion of a set being restricted is independent of the choice of centers. Recall the definition of set being restricted from \cref{def:resitrctness}. 

\begin{lemma}[Recentering]\label{lem:recenter}
	Let $\vecc = (c_1,\ldots,c_m)$ and $\vecc' = (c'_1,\ldots,c'_m)$ be two sequences of centers for the same matching $M$. Then a set $B \subseteq \Z_q^n$ is 
	$(M,\vecc)$-restricted if and only if it is $(M,\vecc')$-restricted.
\end{lemma}
\begin{proof}
	Let $e_i^{(t)} = ((e_i^{(t)})_1, \dots, (e_i^{(t)})_k = c_t)$ (for $t=1,2,\dots,m$) be the ordering of hyperedges corresponding to centering $\vecc$, and let ${e'}_i^{(t)} = (({e'}_i^{(t)})_1, \dots, ({e'}_i^{(t)})_k = c'_t)$.
	
	Given a permutation $\pi: [k] \to [k]$, let $P_\pi$ be a $(k-1)\times (k-1)$ matrix defined as follows: For $1\leq i,j \leq k-1$, let
	\[
	(P_\pi)_{i,j} = \begin{cases} 1 &\text{if $j=\pi(i)$}\,,\\ -1 &\text{if $j=\pi(k)$}\,,\\ 0 &\text{otherwise}\,.\end{cases}
	\]
	
	For $t=1,2,\dots, m$, let $\pi_t:[k]\to [k]$ be the permutation defined by $({e'}_i^{(t)})_j = (e_i^{(t)})_{\pi(j)}$, and let $\pi'_t: [k]\to[k]$ be the permutation defined by $(e_i^{(t)})_j = ({e'}_i^{(t)})_{\pi(j)}$. Then, it is not hard to see that $A_{\vecc'} = Q \cdot A_\vecc$ and $A_\vecc = Q' \cdot A_{\vecc'}$
	where
	\[
	Q = \begin{pmatrix} P_{\pi_1} & & & \\ & P_{\pi_2} & & \\ & & \ddots & \\ & & & P_{\pi_m} \end{pmatrix}, \qquad\qquad Q' = \begin{pmatrix} P_{\pi'_1} & & & \\ & P_{\pi'_2} & & \\ & & \ddots & \\ & & & P_{\pi'_m} \end{pmatrix},
	\]
	and $QQ' = Q'Q = I$, the $(k-1)m\times (k-1)m$ identity matrix.
	
	Now, suppose $B\subseteq \F_2^n$ is $(M,\vecc)$-restricted. Let $B_r \subseteq \Z_q^{(k-1)m}$ be the corresponding reduced set satisfying $B = \{\vecx \in \Z_q^n \mid A_\vecc \vecx \in B_r\}$. Then, let $B_r' = \{Q\vecy \mid \vecy \in B_r\}$.
	
	Note that if $\vecx \in B$, then $A_{\vecc'} \vecx = Q(A_\vecc \vecx) \in B_r'$. Similarly, if $A_{\vecc'} \vecx \in B_r'$, then there is some $\vecy\in B_r$ such that $A_{\vecc'} \vecx = Q\vecy$, and so, $A_\vecc \vecx = Q' A_{\vecc'} \vecx = Q' Q \vecy = \vecy \in B_r$, implying that $\vecx \in B$. It follows that $B$ is $(M,\vecc')$-restricted with reduced set $B'_r$.
	
	In an analogous fashion, it follows that if $B$ is $(M,\vecc')$-restricted, then $B$ is also $(M,\vecc)$-restricted. This completes the proof.
\end{proof}

\subsection{Proof of the ``posterior set'' lemma}\label{sec:reduced bounded base case}

In this subsection, we prove the ``posterior set'' lemma (\cref{lem:boundedness base case}), which shows that every posterior set $B_t$ is $(M_t,C,s)$-reduced for some constant $C$. We include the statement again below for convenience.

\boundednessbasecase*

To see how the above lemma connects to posterior sets, think of $B$ as $B_t$, $M$ as $M_t$, and $\vecc$ as $\vecc_t$. Note that condition (i) of~\cref{lem:boundedness base case} holds by the definition of $B_t$. As for condition (ii), it holds when the message $S_t$ is \textit{typical} and we know by averaging argument that this is the case with high probability (see~\autoref{lem:posterior set is large} for more details).

We now turn to the proof of Lemma~\ref{lem:boundedness base case}. The overall proof follows the outline of \cite{KK19}, but we require extra care in our case, and the proof crucially depends on the ability to recenter (\autoref{lem:recenter}) and a slightly more careful probabilistic analysis.

\begin{proof}[Proof of~\cref{lem:boundedness base case}]

    Given $q$, let $\zeta_q$ be the constant from \cref{lem:hyper}. Let $\zeta_1 = \max\{1,\zeta_q\}$. Given $k$, let $\epsilon_0 = \min\{1, k/(8\zeta_1)\}$. Further let $C_1 = \sqrt{2^{7/2}\zeta_1^{1/2} ek^{3/2} q^{2k}}$ and $C_2 = (2^8\zeta_1e^2kq^{2k})^{1/2}$ and $C = \max\{2,C_1+C_2\}$. For this choice of $C$ and $\epsilon_0$ let $C'$ be the constant given by \cref{lem:weak-vs-strong}. We prove the lemma for $C_0 = C'$. 
  
	Let $M$ be a hypermatching with $m$ edges. (Note we must have $m \leq n/k$.) 
    Recall the definition of $(M,C,s)$-reducedness (\cref{def:reducedness}). The first two conditions of $(M,C,s)$-reducedness are immediate corollaries of~\cref{lem:fourier coeff of set}. In the rest of the proof we focus on showing for every $b \leq s \leq \epsilon_0 n$ and every $h\in\{1,\dots,s\}$ and $\vecv\in\Z_q^n$,
	\[
	\sum_{\substack{\vecu\in\Z_q^n\\\|\vecu+\vecv\|_0=h}}\frac{q^{n}}{|B|}\left|\widehat{\mathbf{1}_B}(\vecu)\right|\leq W_{C,s}(h) \, . \tag{Goal of~\cref{lem:boundedness base case}}
	\]
	Since this holds for every $s \in [b,\epsilon_0 n]$, by \cref{lem:weak-vs-strong}, we get that there is a $C'$ such that $B$ is $(C',s)$-reduced for every $s \in [b,\epsilon_0 n]$ and this yields the lemma given the bound above. 

	Fix an arbitrary $\vecv\in\Z_q^n$. For each $h\in\{1,\dots,s\}$, let $S_h = S_{\vecv,h} = \{\vecu\, \colon\, \|\vecu + \vecv\|_0 = h\}$, i.e., the set of Fourier coefficients in the LHS of the above inequality. We prune $S_h$ to eliminate some terms that are zero. Recall by \cref{lem:fourier coeff of set} that
	$\widehat{\one_B}(\vecw) = 0$ if $\supp(\vecw) \not\subseteq \supp(M)$, or if there exists $i \in [m]$ such that $\langle \vecw, \vece_i \rangle \not\equiv 0 \bmod{q}$.
	Let 
	\[
	T_{\vecv,h,M} = \{\vecu \in S_{\vecv,h} \, |\, \supp(\vecu) \subseteq \supp(M), \langle \vecu, \vece_i \rangle \equiv 0 \bmod{q} \ \forall i \in [m]\},
	\]
	denote the resulting set of vectors which includes all  non-zero Fourier coefficients. Roughly, our approach below is to (1) give an upper bound on the size of the set $T_{\vecv,h,M}$ and (2) bound the sum of the squares of the coefficients in this set. Once we have both these bounds, we can use the Cauchy-Schwartz inequality to conclude the desired bound. Before we undertake these steps, we make some simplifications and some refinements. 
	
	\paragraph{Step 0: Regular condition of $\vecv$.}
	First note that we can assume $\supp(\vecv) \subseteq \supp(M)$. If this is not the case, consider the vector $\widetilde{\vecv}$ given by $\widetilde{v_i} = v_i$ if $i \in \supp(M)$ and $\widetilde{v_i} = 0$ otherwise. Also, let $a = |\{i \, |\, \text{$v_i \neq 0$ and $i \not\in\supp(M)$} \}|$ be the number of nodes in the support of $\vecv$ that are not contained in the hypermatching $M$. Then note that $T_{\vecv,h,M} = T_{\tilde{\vecv},h-a,M}$. If we show that $(q^n/|B|)\cdot \sum_{\vecu \in T_{\tilde{\vecv},h-a,M}} |\widehat{\one_B}(\vecu)| \leq U_{C,s}(h-a)$ then, by the monotonicity of $U_{C,s}(\cdot)$ in the interval $[1,s]$ (see~\autoref{lem:ucs-monotonic}), it follows that $(q^n/|B|)\cdot \sum_{\vecu \in T_{\vecv,h,M}} |\widehat{\one_B}(\vecu)| \leq U_{C,s}(h)$. Thus, from now on, we assume $\supp(\vecv) \subseteq \supp(M)$.
	
	\paragraph{Step 1: A partition of $T_{\vecv,h,M}$.}
	We now further refine $T_{\vecv,h,M}$, i.e., the set of non-zero Fourier coefficients. For an integer $\ell$, let $T_{\vecv,h,\ell,M} = \{\vecu \in T_{\vecv,h,M} \, |\, \#\{i\in[m] \, |\, e_i \cap \supp(\vecu+\vecv) \not=\emptyset\} = \ell\}$ be the set $\vecu\in T_{\vecv,h,\ell,M}$'s such that the support of $\vecu+\vecv$ touches exactly $\ell$ edges. Since $\vecv,h$ and $M$ will be fixed in the rest of this proof, we simplify the notation and refer to this set as $T_\ell$. Note that $h/k \leq \ell \leq \min\{m,h\}$. Thus, the quantity we are interested in this lemma can be upper bounded as follows.
	\begin{align}
	\sum_{\substack{\vecu\in\Z_q^n\\\|\vecu+\vecv\|_0=h}}\frac{q^{n}}{|B|}\left|\widehat{\mathbf{1}_B}(\vecu)\right|
    & = \frac{q^n}{|B|} \sum_{\vecu\in S_h}\left|\widehat{\one_B}(\vecu)\right| \nonumber \\
    & = \frac{q^n}{|B|} \sum_{\vecu\in T_{\vecv,h,M}} \left|\widehat{\one_B}(\vecu)\right| \nonumber \\
    & = \sum_{\ell = h/k}^{\min\{m,h\}} \frac{q^n}{|B|} \sum_{\vecu\in T_{\ell}} \left|\widehat{\one_B}(\vecu)\right| \nonumber \\
    & \leq \sum_{\ell = h/k}^{\min\{m,h\}} \frac{q^n}{|B|} \sqrt{ |T_{\ell}| \sum_{\vecu\in T_{\ell}} \widehat{\one_B}(\vecu)^2} \label{eqn:base-case-1}
	\end{align}
	where the second equality is due to~\autoref{lem:fourier coeff of set}, the third equality is due to the partition, and the last inequality is by Cauchy-Schwarz inequality. (The reason why we partition $T_{\vecv,h,M}$ into $T_\ell$'s is that the Fourier square-mass within $T_\ell$ and the cardinality of $T_\ell$ can be properly upper bounded respectively.)
	
	\paragraph{Step 2: Upper bounding the squared Fourier mass within $T_\ell$.}
	To upper bound the squared Fourier mass within $T_\ell$, we utilize the fact that the posterior set $B$ is independent to the choice of center $\vecc$ (i.e.,~\autoref{lem:recenter}) and hypercontractivity (i.e., \cref{lem:hyper}). We stress that in the bound we establish below 
	it is crucial that the exponent of $b$ is $h-\ell$ (as opposed to the more trivial $h$, or $h(k-1)/k$). In turn this bound is obtained by using a random center $\vecc$ and this randomization is permitted at the analysis stage by \cref{lem:recenter}.
	
	\begin{claim}\label{clm:random-center}
		Let $\zeta_1 = \max\{1,\zeta_q\}$ where $\zeta_q$ is the constant from \cref{lem:hyper}. If $|B|\geq q^{n-b}$ for some $b\in\N$ then for every $1\leq\ell< h\leq b$, we have
		\[
		\sum_{\vecu \in T_\ell} \widehat{\one_{B}}(\vecu)^2 \leq k^{\ell} \left(\frac{|B|}{q^{n}}\right)^2 \left(\frac{\zeta_1 \cdot b}{h-\ell}\right)^{h-\ell} \, .
		\]
	\end{claim}
		The proof of this claim uses \cref{lem:fourier coeff of set} to relate the Fourier coefficients of the function $\one_B$ to those of $\one_{B_r}$. But note that the ``reduced set'' $B_r$ depends on the choice of the center. Furthermore the weight of the Fourier coefficient in the reduced space depends on how the centers overlap with $\supp(\vecu+\vecv)$. Specifically we have that for centers $\vecc$, $\|\widetilde{A}_{\vecc}(\vecu+\vecv)\|_0 = \|\vecu+\vecv\|_0 - t$, where $t = |\{i \in [m] \, |\, c_i \in \vecu + \vecv\}|$ is the number of centers contained in $\vecu+\vecv$. Note that $t \leq \ell$ since the number centers in $\supp(\vecu + \vecv)$ can not exceed the number of edges touching this set. The crux of this proof is that we if choose the centers randomly then there is a positive probability that all centers (of the edges that touch $\supp(\vecu+\vecv)$) are in $\supp(\vecu+\vecv)$. We argue the formal details below.
        
	\begin{proof}[Proof of \cref{clm:random-center}] 
		For a random center $\vecc$, let $A_{\vecc}$ denote the $\vecc$-centered folded encoding of $M$, and let $B_{r,\vecc}=\{A_\vecc\vecx \, |\, \vecx\in B \} \subseteq \Z_q^{(k-1)m}$. 
        For $\vecu \in T_\ell$, let $I_\vecu(\vecc) =1$ if $c_i \in \supp(\vecu+\vecv)$ for every $i\in[m]$ with $e_i\cap\supp(\vecu+\vecv)\neq\emptyset$ and $0$ otherwise. Note that $\Pr_{\vecc}[I_\vecu(\vecc)=1] \geq k^{-\ell}$. Now consider the following expression:
		\begin{align*}   
			\Exp_{\vecc} \left[\sum_{\substack{\vecw \in \Z_q^{(k-1)m}\\ \|\vecw+\widetilde{A}_\vecc\cdot \vecv\|_0 = h-\ell}} \widehat{\one_{B_{r,\vecc}}}(\vecw)^2 \right] 
			& \geq \Exp_{\vecc} \left[\sum_{\vecu \in T_{\ell}} I_\vecu(\vecc)\cdot \widehat{\one_{B}}(\vecu)^2 \right] & (\because~\autoref{lem:fourier coeff of set})\\
			& = \sum_{\vecu \in T_{\ell}} \left(\widehat{\one_{B}}(\vecu)^2 \Exp_{\vecc} \left[ I_\vecu(\vecc)  \right]\right) & (\because B\text{ and }T_{\ell}\text{ are independent to }\vecc)\\
			& \geq k^{-\ell} \sum_{\vecu \in T_{\ell}} \widehat{\one_{B}}(\vecu)^2 \, . & (\because \Pr_{\vecc}[I_\vecu(\vecc)=1] \geq k^{-\ell})
		\end{align*}
        Rearranging the above we get 
        \begin{equation}
            \label{eqn:exp-c-fourier-ub}
            \sum_{\vecu \in T_{\ell}} \widehat{\one_{B}}(\vecu)^2 \leq k^\ell \cdot \Exp_{\vecc} \left[\sum_{\substack{\vecw \in \Z_q^{(k-1)m}\\ \|\vecw+\widetilde{A}_\vecc\cdot \vecv\|_0 = h-\ell}} \widehat{\one_{B_{r,\vecc}}}(\vecw)^2 \right]. 
        \end{equation}

        On the other hand, since $B$ is $(M,\vecc)$-restricted, we have $|B|=|\{\vecx\in\Z_q^n\ |\ A_\vecc\vecx\in B_{r,\vecc}\}|=|B_{r,\vecc}|\cdot q^{n-\textsf{rank}(A_\vecc)}$. As $\textsf{rank}(A_\vecc)=(k-1)m$ and $|B|\geq q^{n-b}$, we have that $|B_{r,\vecc}|=|B|/q^{n-(k-1)m}\geq q^{(k-1)m-b}$. Hence, by \autoref{lem:hyper} (invoked with $n\gets(k-1)m$ and $B\gets B_{r,\vecc}$), for every $\vecc$ we have 
		\begin{align*}
			\sum_{\substack{\vecw \in \Z_q^{(k-1)m}\\ \|\vecw+\tilde{A}_\vecc\cdot \vecv\|_0 = h-\ell}} \widehat{\one_{B_{r,\vecc}}}(\vecw)^2 &\leq \left(\frac{|B_{r,\vecc}|}{q^{(k-1)m}}\right)^2 \left(\frac{\zeta_1 b}{h-\ell}\right)^{h-\ell} \, ,
		\end{align*}
		where $\zeta_1 = \max\{1,\zeta_q\}$ and $\zeta_q$ is the constant from \cref{lem:hyper}.
		Also, by $|B|=|B_{r,\vecc}|\cdot q^{n-(k-1)m}$, the above inequality becomes
		\begin{align*}
			\sum_{\substack{\vecw \in \Z_q^{(k-1)m}\\ \|\vecw+\tilde{A}_\vecc\cdot \vecv\|_0 = h-\ell}} \widehat{\one_{B_{r,\vecc}}}(\vecw)^2 &\leq
			\left(\frac{|B|}{q^{n}}\right)^2 \left(\frac{\zeta_1 \cdot b}{h-\ell}\right)^{h-\ell} \, .
		\end{align*}
		Taking expectations over $\vecc$ we thus get:
        
        \begin{equation}\label{eqn:exp-c-fourier-lb}
		   \Exp_{\vecc}\left[\sum_{\substack{\vecw \in \Z_q^{(k-1)m}\\ \|\vecw+\tilde{A}_\vecc\cdot \vecv\|_0 = h-\ell}} \widehat{\one_{B_{r,\vecc}}}(\vecw)^2 \right]  \leq
			\left(\frac{|B|}{q^{n}}\right)^2 \left(\frac{\zeta_1 \cdot b}{h-\ell}\right)^{h-\ell} \, . 
		\end{equation}
        Putting the inequalities \cref{eqn:exp-c-fourier-ub} and \cref{eqn:exp-c-fourier-lb} together we get
		$$\sum_{\vecu \in T_\ell} \widehat{\one_{B}}(\vecu)^2 \leq k^{\ell} \left(\frac{|B|}{q^{n}}\right)^2 \left(\frac{\zeta_1\cdot b}{h-\ell}\right)^{h-\ell},$$
		thus proving the claim.
	\end{proof}
	
	\paragraph{Step 3: Upper bounding the cardinality of $T_\ell$.}
	Next, we turn to bounding the size of the set $T_\ell$. To do so we explore the structure of the vectors in $T_{\ell}$. We start with some notation. Let 
	$E = \{i\in[m] \, |\,  \langle \vecv, \vece_i \rangle \equiv 0 \pmod{q}\}$ 
	and $O = \{i\in[m] \,  \langle \vecv, \vece_i \rangle \not\equiv 0 \pmod{q} \}$. Given a vector $\vecu$, we define $W_e = W_e(\vecu) = (\vecu+\vecv) \odot \left(\sum_{i\in E} \vece_i\right)$ and $W_o = W_o(\vecu) = (\vecu+\vecv) \odot \left(\sum_{i\in O} \vece_i\right)$, where $\odot$ is used to denote the Hadamard product (entrywise product) of two vectors.
	Let $\eta$ denote the number of edges touched by $W_e$ and let $o$ denote the number of edges touched by $W_o$.
	Note the following conditions hold when $\vecu \in T_\ell$.
	\begin{claim}\label{claim:reducedness base case combo}
		If $\vecu \in T_\ell$, then all the following conditions hold: (1) $|O| \leq h$, (2) $\eta + o = \ell$, and (3) $\eta\leq h/2$.
	\end{claim}
	\begin{proof}
		We prove each of the individual claims below:
		\begin{enumerate}
			\item Note that $\langle \vecu + \vecv, \vece_i \rangle \not\equiv 0 \pmod{q}$ for every $i\in O$, since $\langle \vecu, \vece_i \rangle \equiv 0 \pmod{q}$ and $\langle \vecv, \vece_i \rangle \not\equiv 0 \pmod{q}$. Therefore, $|\supp(\vecu+\vecv) \cap e_i| \geq 1$ for every $i\in O$, implying that $|O| \leq \sum_{i\in O} |\supp(\vecu+\vecv) \cap e_i| \leq \|\vecu+\vecv\|_0 = h$.
			\item Since $\vecu + \vecv$ touches $\ell$ edges, $\eta + o = \ell$.
			\item Note that $\langle \vecu + \vecv, \vece_i \rangle \equiv 0 \pmod{q}$ for every $i\in E$, since $\langle \vecu, \vece_i \rangle \equiv 0 \pmod{q}$ and $\langle \vecv, \vece_i \rangle \equiv 0 \pmod{q}$. Therefore, if $i\in E$ is touched by $W_e$ (i.e., $W_e \odot \vece_i \neq \mathbf{0}$), then it follows that $W_e$ touches it in at least two points, i.e., $|\supp(W_e\odot \vece_i)| \geq 2$ (see~\autoref{fig:cardinality}). Combined with the fact that $|\supp(W_e)| \leq \|\vecu+\vecv\|_0 = h$, we obtain $\eta \leq h/2$, as desired.
		\end{enumerate}
	\end{proof}

	\begin{figure}[ht]
		\centering
		\includegraphics[width=14cm]{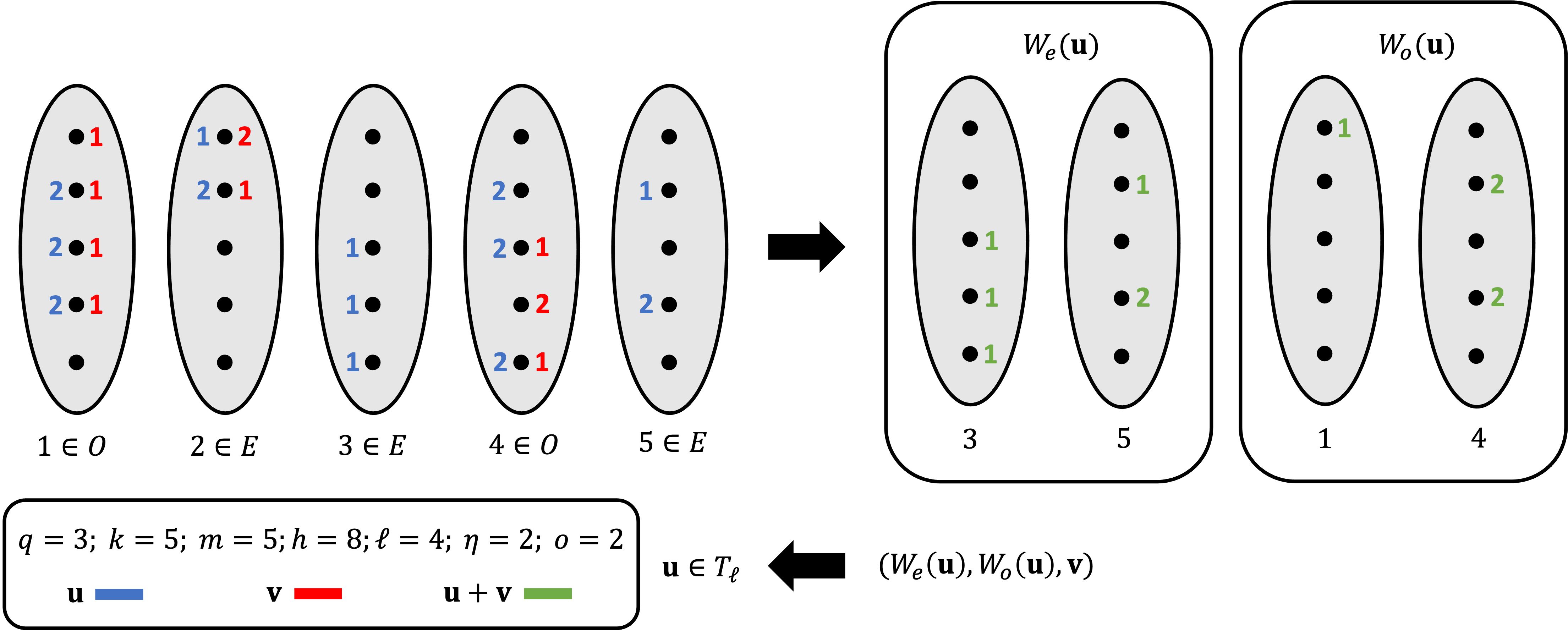}
		\caption{Upper bound the cardinality of $T_\ell$. In this example, $q=3, k=5, m=5$, and $\vecv$ is specified by integers in {\color{red!100!black!90!}red}. Note that $E=\{2,3,5\}$ and $O=\{1,4\}$. Next, we consider $h=8,\ell=4$ and a $\vecu\in T_\ell$ specified by integers in {\color{blue!60!black!90!} blue}. Note that by definition we have $\eta=o=2$. In particular, the tuple $(W_e(\vecu),W_o(\vecu))$ is described on the right and $\vecu+\vecv$ is specified by integers in {\color{green!60!black!90!} green}. It is immediate to see that $(W_e(\vecu),W_o(\vecu),\vecv)$ uniquely specifies $\vecu$ because one can subtract $W_e(\vecu)$ and $W_o(\vecu)$ by $\vecv$ to get the value of $\vecu$ in those coordinates. In the rest of the coordinates, $\vecu$ has the same values as $\vecv$. Moreover, observe that every hyperedge in $W_e(\vecu)$ should contain at least $2$ non-zero points because both $\vecu$ and $\vecv$ sum up to $0$ mod $q$ within those hyperedges.}
		\label{fig:cardinality}
	\end{figure}
Based on these restrictions on $\vecu\in T_\ell$, we can now get the following bound on $|T_\ell|$.

\begin{claim}\label{claim:T_l-ub} For every $\ell\in \{h/k,\ldots,\min\{h,m\}\}$ we have:
\[
|T_\ell| \leq (4q^{2k})^h (em/\eta^*)^{\eta^*}, \mbox{~~~~ where } \eta^*:=\min\{\ell,h/2\}\, .   
\]
\end{claim}

\begin{proof}
    Recall that each $\vecu\in T_\ell$ is uniquely specified by the pair $(W_e,W_o)$ (see~\autoref{fig:cardinality}) and it therefore suffices to count the number of distinct choices of $(W_e,W_o)$. First, we see that the number of possibilities for $W_o$ is at most $(q^k)^{|O|} \leq q^{kh}$ (since $|O|\leq h$ by the first item of~\cref{claim:reducedness base case combo}). Now, having fixed $W_o$ and $o$, consider the number of possibilities of $W_e$. We may choose $W_e$ by picking a set $F\subseteq E$ with $\eta$ edges, and then picking $|\supp(W_e)| = h - |\supp(W_o)|$ elements from the union of the edges in $F$, each of which is given a value in $\Z_q \setminus \{0\}$. Note that $F$ can be chosen in at most $\binom{|E|}{\eta} \leq \binom{m}{\eta}$ ways, after which $W_e$ can be chosen in $\leq q^{k\eta}$ ways. Finally, note that by the second and third items of~\autoref{claim:reducedness base case combo} we have $\eta\leq\min\{\ell,h/2\}=:\eta^*$. Putting these together we get:
    \[|T_\ell| \leq \sum_{\eta=0}^{\eta^*} \left\{q^{kh} \binom{m}{\eta} q^{k\eta}\right\} \leq q^{2kh} \sum_{\eta=0}^{\eta^*}  \binom{m}{\eta} ,\]
    where the second inequality uses $\eta \leq h$. Now we consider two cases based on whether $\eta^* \leq m/3$ or not. If $\eta^* > m/3$, since $\eta^* \leq \ell \leq m$ we have 
    \[
    \sum_{\eta=0}^{\eta^*}  \binom{m}{\eta} \leq \sum_{\eta=0}^{m}  \binom{m}{\eta} = 2^m \leq 2^{2h}\, ,
    \]
    where the final inequality uses $\min\{\ell,h/2\} > m/3$ to infer $m < 3h/2 < 2h$. 
    Thus in this case we get 
    \[
    |T_\ell| \leq q^{2kh} 4^h = (4q^{2k})^h \leq  (4q^{2k})^h (em/\eta^*)^{\eta^*}\, ,
    \]
    where the final inequality above uses $em \geq m \geq \eta^*$.
    In the case, $\eta^* \leq m/3$ we note that $\binom{m}{\eta} \geq 2 \binom{m}{\eta-1}$ for every $\eta \in \{0,\ldots,\eta^*\}$ and thus $\sum_{\eta=0}^{\eta^*}\binom{m}{\eta}$ is a telescoping sum bounded by $2\binom{m}{\eta^*}$ which in turn is bounded by $2(em/\eta^*)^{\eta^*} \leq 2^{2h} (em/\eta^*)^{\eta^*}$. Again the desired bound on $|T_\ell|$ follows. 
\end{proof} 
   
	\paragraph{Step 4: Completing the proof of~\cref{lem:boundedness base case}.}
	The boundedness of $B$ now follows from some straightforward (though tedious) calculations. 
    Continuing with the RHS of \cref{eqn:base-case-1}, we have: 
    
	\begin{align*}
	\sum_{\ell = h/k}^{\min\{h,m\}} \frac{q^n}{|B|} \sqrt{ |T_{\ell}| \sum_{\vecu\in T_{\ell}} \widehat{\one_B}(\vecu)^2}
		& \leq \sum_{\ell = h/k}^{\min\{h,m\}} \frac{q^n}{|B|} \sqrt{ (4q^{2k})^h\cdot  (e m/\eta^*)^{\eta^*} \cdot k^{\ell}\cdot  \left(\frac{|B|}{q^{n}}\right)^2 \cdot \left(\frac{\zeta_1 s}{h-\ell}\right)^{h-\ell} } \\
		& \mbox{~~~~~~~ (By \cref{clm:random-center} and \cref{claim:T_l-ub}, and using $b\leq s$)}\\
        & = \sum_{\ell = h/k}^{\min\{h,m\}} \sqrt{ (4q^{2k})^h\cdot  (e m/\eta^*)^{\eta^*} \cdot k^{\ell}\cdot \left(\frac{\zeta_1 s}{h-\ell}\right)^{h-\ell} } \\
        & = S_1 + S_2 
    \end{align*}
    
    where 
    $$S_1 := \sum_{\ell = h/k}^{h/2} \sqrt{ (4q^{2k})^h\cdot  (e m/\eta^*)^{\eta^*} \cdot k^{\ell}\cdot \left(\frac{\zeta_1 s}{h-\ell}\right)^{h-\ell} },$$ and 
    $$S_2 := \sum_{\ell = h/2+1}^{\min\{h,m\}} \sqrt{ (4q^{2k})^h\cdot  (e m/\eta^*)^{\eta^*} \cdot k^{\ell}\cdot \left(\frac{\zeta_1 s}{h-\ell}\right)^{h-\ell}}.$$
    Using $\eta^* = \ell$ for the regime in $S_1$ and $\eta^* = h/2$ in the $S_2$ regime we can simplify the above two sums as follows:
    \begin{align*}
        S_1 & = \sum_{\ell = h/k}^{h/2} \sqrt{ (4q^{2k})^h\cdot  (e m/\ell)^{\ell} \cdot k^{\ell}\cdot \left(\frac{\zeta_1 s}{h-\ell}\right)^{h-\ell} } \\
        & \leq \sum_{\ell = h/k}^{h/2} \sqrt{ (4ekq^{2k})^h\cdot  (m/\ell)^{\ell} \cdot \left(\frac{\zeta_1 s}{h-\ell}\right)^{h-\ell} } \\
        & \leq \sum_{\ell = h/k}^{h/2} \sqrt{ (4ekq^{2k})^h\cdot  (m/\ell)^{\ell} \cdot \left(\frac{2\zeta_1 s}{h}\right)^{h-\ell} } ~~~\mbox{(Using $\ell \leq h/2$)}\\
        & = (8\zeta_1ekq^{2k})^{h/2} \left(\frac{s}{h}\right)^{h/2}\sum_{\ell = h/k}^{h/2} \sqrt{ \left(\frac{m h }{2\zeta_1 s \ell}\right)^{\ell}  }\\
        & \leq (8\zeta_1ekq^{2k})^{h/2} \left(\frac{s}{h}\right)^{h/2}\sum_{\ell = h/k}^{h/2} \sqrt{ \left(\frac{k m}{2\zeta_1 s}\right)^{\ell}} \mbox{~~~(Using $\ell \geq h/k$)} \\
        & \leq (8\zeta_1ekq^{2k})^{h/2} \left(\frac{s}{h}\right)^{h/2}\sum_{\ell = h/k}^{h/2} \sqrt{ \left(\frac{k n}{2\zeta_1 s}\right)^{\ell}} \mbox{~~~(Using $m \leq n$)} \\
        & \leq (8\zeta_1ekq^{2k})^{h/2} \left(\frac{s}{h}\right)^{h/2}  \cdot 2\cdot \left(\frac{k n}{2\zeta_1 s}\right)^{h/4} \mbox{~~~(Using $s\leq \epsilon_0 n$. See footnote\footnotemark.)}\\     
        & = 2(2^{5/2}\zeta_1^{1/2} ek^{3/2} q^{2k})^{h/2} \left(\frac{\sqrt{sn}}{h}\right)^{h/2}  \\
        & \leq C_1^h \left(\frac{\sqrt{sn}}{h}\right)^{h/2} \, ,
        \end{align*}
for $C_1 \geq \sqrt{2^{7/2}\zeta_1^{1/2} ek^{3/2} q^{2k}}$. 
\footnotetext{Since $s\leq \epsilon_0 n \leq kn/8\zeta_1$, we have $\sqrt{\frac{kn}{2\zeta_1 s}}\geq 2$ and so the sum $\sum_{\ell=h/k}^{h/2} \sqrt{(\frac{kn}{2\zeta_1 s})^\ell}$ telescopes to at most $2\left(\frac{k n}{2\zeta_1 s}\right)^{h/4}$.}
We now turn to simplifying $S_2$. We have 
\begin{align*}
    S_2  &= \sum_{\ell = h/2+1}^{\min\{h,m\}} \sqrt{ (4q^{2k})^h\cdot  (e m/\eta^*)^{\eta^*} \cdot k^{\ell}\cdot \left(\frac{\zeta_1 s}{h-\ell}\right)^{h-\ell}} \\
    & = \sum_{\ell = h/2+1}^{\min\{h,m\}} \sqrt{ (4q^{2k})^h\cdot  (2 e m/h)^{h/2} \cdot k^{\ell}\cdot \left(\frac{\zeta_1 s}{h-\ell}\right)^{h-\ell}} \mbox{~~~(Using $\eta^*=h/2$ in this regime)}\\     
    & \leq \sum_{\ell = h/2+1}^{h} \sqrt{ (4q^{2k})^h\cdot  (2 e m/h)^{h/2} \cdot k^{\ell}\cdot \left(\frac{\zeta_1 s}{h-\ell}\right)^{h-\ell}} \\  
    & \leq (8ekq^{2k})^{h/2} (m/h)^{h/4} \sum_{\ell = h/2+1}^{h} \sqrt{\left(\frac{\zeta_1 s}{h-\ell}\right)^{h-\ell}} \\     
    & = (8ekq^{2k})^{h/2} (m/h)^{h/4} \sum_{\ell' = 0}^{h/2-1} \sqrt{\left(\frac{\zeta_1 s}{\ell'}\right)^{\ell'}}\\ 
    & \leq (8ekq^{2k})^{h/2} (m/h)^{h/4} \sum_{\ell' = 0}^{h/2-1} \sqrt{\left(\frac{\zeta_1 s}{h/2}\right)^{\ell'}e^{h/2}} \mbox{~~~(Using $(x/y)^y \leq (x/z)^y\cdot e^z, ~~\forall x>0, y\geq 1 ,z\geq1$)}\footnotemark\\
    & \leq (8e^2kq^{2k})^{h/2} (m/h)^{h/4} \sum_{\ell' = 0}^{h/2-1} \sqrt{\left(\frac{\zeta_1 s}{h/2}\right)^{\ell'}} \\
    & \leq (8e^2kq^{2k})^{h/2} (n/h)^{h/4} \zeta_1^{h/2} 4 ((2s)/h)^{h/4} \mbox{~~~(Using $m \leq n$, $\zeta_1 \geq 1$ and $h \leq s$)} \\
    & \leq C_2^h \left(\frac{\sqrt{sn}}{h}\right)^{h/2} \, , 
\end{align*}
\footnotetext{{This inequality is derived by seeing $(x/y)^y =(x/z)^y(z/y)^y$ and $(z/y)^y \leq e^z$ for every $y,z\geq 1$.}}
for $C_2 \geq (2^8\zeta_1e^2kq^{2k})^{1/2}$.
Combining the bounds on $S_1$ and $S_2$ we get 
\[
	\sum_{\substack{\vecu\in\Z_q^n\\\|\vecu+\vecv\|_0=h}}\frac{q^{n}}{|B|}\left|\widehat{\mathbf{1}_B}(\vecu)\right|
    = 	\frac{q^n}{|B|} \sum_{\vecu\in S_h}\left|\widehat{\one_B}(\vecu)\right|  \leq S_1 + S_2 \leq (C_1^h + C_2^h) (\sqrt{sn}/h)^{h/2} \leq C^h (\sqrt{sn}/h)^{h/2}\,  
    \]
for $C \geq \max\{2,C_1 + C_2\}$. 

Thus, we conclude that $B$ is $(M,C,s)$-weakly-reduced for every $s \in [b,\epsilon_0 n]$, and so by \cref{lem:weak-vs-strong}, $B$ is $(M,C_0,s)$-strongly-reduced for every $s \in [b,\epsilon_0 n]$. 
This concludes the proof of \cref{lem:boundedness base case}.
\end{proof}

\subsection{Proof: boundedness implies near uniformity}\label{sec:bounded implies uniform}

In this section we prove \cref{lem:boundedness implies uniform} which is used to prove 
condition (iii) of~\cref{lem:hybrid}.
We restate the lemma below for convenience.

\boundednessuniform*

In the following, we denote $m=\alpha n$ for simplicity. Let us start with defining a combinatorial quantity $p(h,k,m,n)$ and showing an upper bound on it.

\begin{definition}\label{def:p-one}
	Suppose $k, m, n > 0$ are integers. We define  $p(h,k,m,n)$ to be the probability that a uniformly random $k$-hypermatching $M$ on vertex set $[n]$ with $m$ hyperedges each of size $k$, the support of $M$ contains $[h]$ and further satisfies the condition that every hyperedge of $M$ contains either $0$ or at least two vertices from $[h]$. 
\end{definition}

\begin{restatable}[{\cite[Lemma~6.8]{CGSV21-conference}}]{lemma}{matchingprobfolded}\label{lem:matching prob folded}
	For every $\beta_0 > 0$ and $k$ there exists $\alpha_0 > 0$ such that for all integers $n,k$, $\alpha\in(0,\alpha_0]$, $m = \alpha n$, and $0\leq h\leq {km}$, we have
	\[
	p(h,k,m,n)\leq \left(\frac{\beta_0 h}{n}\right)^{h/2}  \, .
	\]
	Furthermore, $p(h,k,m,n)=0$ if $h>km$. 
\end{restatable}

We include a proof for convenience.
\begin{proof}
We prove the lemma for $\alpha_0 = \beta_0/(8e^3k^5)$. 
The definition of $p(\cdots)$ explores the probability that a fixed set $H = [h]$ satisfies some conditions with respect to a random matching $M$. By symmetry we can instead view it as the probability that a uniformly random set $H$ satisfies the same conditions with respect to a fixed matching $M$ with edges $e_1,\ldots,e_m$. (We abuse notation to also use $M$ to denote $\cup_{i\in[m]} e_i$, i.e.,  the subset of vertices incident to the matching.)

Given a matching $M$, let $\cF = \{H \subseteq [n] \mid |H| = h, H \subseteq M, |H \cap e_i|\ne 1, \forall i\in [m]\}$. We have $p(h,k,m,n) = |\cF|/\binom{n}h$, and so it suffices to bound $|\cF|$ from above.
Given $H \in \cF$, let $E(H) = \{i\in[m] | H \cap e_i \ne \emptyset\}$
denote the set of edges touching $H$ and let $\eta = |E(H)|$. We have that $h/k \leq \eta \leq h/2$ since every edge includes at least two vertices of $H$. To choose an $H\in \cF$ we may choose $\eta \in [h/k,h/2]$, $E \subseteq [m]$ of size $\eta$ and then choose $H$ of size $h$ from the set of vertices incident to $E$. (There are further conditions that we will ignore to get the upper bound.) Given $\eta$ there are $\binom{m}{\eta}$ ways of choosing $E$, and given $E$, there are at most $\binom{k\eta}{h}$ ways of choosing $H$ from the vertices touched by $E$. We thus get that $|\cF| \leq \sum_{\eta = h/k}^{h/2} \binom{m}{\eta} \binom{k\eta}{h}$. Applying this we now get the following inequalities: 
\begin{align*}
 p(h,k,m,n) & \leq \binom{n}h^{-1} \cdot \sum_{\eta = h/k}^{h/2} \binom{m}{\eta} \binom{k\eta}{h}  \\
            & \leq \sum_{\eta = h/k}^{h/2} \binom{m}{\eta} (e k \eta/h)^h (h/n)^h \mbox{~~~~(Using $(a/b)^b \leq \binom a b \leq (ea/b)^b$)} \\
            & \leq (e k h/n)^h \sum_{\eta = 1}^{h/2} \binom{m}{\eta}  \mbox{~~~~(Using $\eta \leq h$)}\\
            & \leq 2^h (ek h/n)^h (2ekm/h)^{h/2}\\
            & = (8 e^3 k^5 \alpha h/n)^{h/2}\, ,
\end{align*}
where the last inequality uses $\sum_{\eta = 1}^{h/2} \binom{m}{\eta} \leq 2^h (2ekm/h)^{h/2}$ for every $m$ and $h \in [km]$. (If $h \leq m$ the final term is the largest and bounded by $(2ekm/h)^{h/2}$ and so the entire sum is at most $h (2ekm/h)^{h/2} \leq 2^h (2ekm/h)^{h/2}$. If $h \in (m,km]$, then the sum is at most $2^m$ while the RHS is at least $2^h$ (in particular $2ekm/h\geq 1$).) So we have that $p(h,k,m,n) \leq (8 e^3 k^5 \alpha h/n)^{h/2} \leq (8 e^3 k^5 \alpha_0 h/n)^{h/2} = (\beta_0 n/h)^{h/2}$ since $\beta_0 =  8e^3 k^5 \alpha_0$. 
\end{proof}

The following lemma is an immediate corollary of~\cref{lem:matching prob folded} and will be useful later in the proof of~\autoref{lem:boundedness implies uniform}. 

\begin{lemma}\label{lem:boundedness implies uniform calculation}
    For every $k,q\geq 2$ there exists $\alpha_0 > 0$ such that for every $\delta \in (0,1/2)$ and $C<\infty$ there exists $\tau > 0$ and $s_0 < \infty$ such that for all integers $n$,  $s$ and $m$ satisfying $s_0\leq s\leq \tau n$, $m \leq \alpha_0 n$ we have:
	\[
	p(h,k,m,n)U_{C,s}(h)\leq\left\{\begin{array}{ll}
		\delta^{2h}     &  ,\ 1\leq h\leq s \\
		2^{-h/2}    & ,\ s<h\leq km \\
		0 & ,\ h > km
	\end{array}\right. \, .
	\]
	Specifically,
	\[
	\sum_{h=2}^{n}p(h,k,m,n)U_{C,s}(h)\leq\delta^2 \, .
	\]
\end{lemma}

\begin{proof}

    Let $\beta_0 = \frac{1}{4q^2e^2}$, and let $\alpha_0$ be as in \cref{lem:matching prob folded} for this choice of $\beta_0$. Let $\tau = \delta^8/(C^2\beta_0^2)$ and $s_0 = 4\log_2(3/\delta)$. 

    By \cref{lem:matching prob folded} and the definition of (strongly-)boundedness we have the following:
    \begin{itemize}
		\item If $1\leq h\leq s$, then
		\begin{align*}
			p(h,k,m,n)U_{C,s}(h)&\leq\left(\frac{\beta_0 h}{n}\right)^{h/2} \cdot \left(\frac{C \sqrt{sn}}{h}\right)^{h/2} 
            =\left(\frac{\beta_0 C \sqrt{s}}{\sqrt{n}}\right)^{h/2} \leq (\beta_0 C \sqrt{\tau})^{h/2} \leq \delta^{2h}, 
		\end{align*}
        where the second inequality uses $s\leq \tau n$ and the third uses $\beta_0 C \sqrt{\tau} \leq \delta^4$.
		\item If $s < h \leq km$, then 
		\begin{align*}
			p(h,k,m,n)U_{C,s}(h)&\leq\left(\frac{\beta_0 h}{n}\right)^{h/2} \cdot \left(\frac{2q^2 e^2 n}{h}\right)^{h/2} = (\beta_0 2q^2 e^2)^{h/2}
            \leq 2^{-h/2},
		\end{align*}
		where the final inequality uses $\beta_0 2q^2 e^2 \leq 1/2$. 
		\item If $h > km$, we have $p(h,k,m,n)=0$ and hence $p(h,k,m,n)U_{C,s}(h)=0$.
	\end{itemize}

	Finally, we have 
	\begin{align*}
		\sum_{h=2}^{n}p(h,k,m,n)U_{C,s}(h) &\leq \sum_{h=2}^{s} \delta^{2h} + \sum_{h=s+1}^{km} 2^{-h/2}\\
		&\leq \frac{\delta^4}{1-\delta^2} + \frac{2^{-s-1/2}}{1-(1/\sqrt{2})}\\
    	&\leq \frac{\delta^2}{2} + \frac{2^{-s-1/2}}{1-(1/\sqrt{2})}\\
		&\leq \frac{\delta^2}{2} + \frac{\delta^2}{2}\\
		&= \delta^2,
	\end{align*}
    where the second inequality uses $\delta < 1/2$ and the third uses $s \geq 4 \log_2(3/\delta)$. 
\end{proof}

Now, we are ready to prove the main lemma of this subsection.
\begin{proof}[Proof of~\cref{lem:boundedness implies uniform}] 
	Given $k,q \geq 2$ and $\delta>0$ and $C < \infty$ let $\alpha_0 = \alpha_0(k,q) > 0$ and $\tau = \tau(k,q,\delta,C) > 0$ and $s_0 = s_0(\delta)$ be as given by \cref{lem:boundedness implies uniform calculation}.

	Let $m \leq \alpha_0 n$, $s_0 \leq b \leq s\leq \tau n $, and let $B\subset \Z_q^n$ be a $(C,s)$-bounded set with $|B|\geq q^{n-b}$. The goal is to prove that with probability   at least $1-\delta$ over a uniform random choice of $k$-hypermatching $M$ on $m$ edges and a random choice of center sequence $\vecc$ the following holds for every $\vecz_0\in\Z_q^{(k-1)m}$,
	\[
	1-\delta\leq q^{(k-1)m}\Pr_{\vecx\sim \Unif(B)}[A_{\vecc}\vecx=-\vecz_0]\leq 1+\delta 
	\]
	(Recall that $A_{\vecc}$ was defined in \cref{eqn:Ac-defn}. Note that the switch from $\vecz_0$ to $-\vecz_0$ in the event described above does not alter the statement being proved since we are proving this for every vector $\vecz_0$.)
	
	Now, for a fixed $k$-hypermatching $M$ and fixed choice of centers $\vecc$, let us expand the marginal probability as follows. Let $f:\Z_q^n \to \{0,1\}$ be the indicator function of the set $B$. For a fixed $\vecz_0\in\Z_q^{(k-1)m}$, let $g = g_{A_{\vecc},\vecz_0}:\Z_q^n \to \{0,1\}$ be the function given by $g(\vecx)=\mathbf{1}_{A_{\vecc}\vecx=\vecz_0}$. Letting $g = g_{A_{\vecc},\vecz_0}$, we have
	\begin{align}
		q^{(k-1)m}\Pr_{\vecx\sim \Unif(B)}[A_{\vecc}\vecx=-\vecz_0]&=\frac{q^{(k-1)m}}{|B|}\sum_{\vecx\in\Z_q^n}f(\vecx)g(-\vecx) \nonumber\\
		&=\frac{q^{(k-1)m}}{|B|} (f\star g)(0) \nonumber\\
		&=\frac{q^{(k-1)m}}{|B|}\sum_{\substack{\vecu\in\Z_q^n}}\widehat{f\star g}(\vecu) \nonumber\\
		&=\frac{q^{(k-1)m+n}}{|B|}\sum_{\substack{\vecu\in\Z_q^n}}\widehat{f}(\vecu)\widehat{g}(\vecu) ~~~ \mbox{(By \cref{lem:convthm})} \nonumber\\
		&=1 + \frac{q^{n+(k-1)m}}{|B|}\sum_{\substack{\vecu\in\Z_q^n\\\vecu\neq0^n}}\widehat{f}(\vecu)\widehat{g}(\vecu) \label{eq:diff}\\
		&~~~ \mbox{(Since $q^{n}\widehat{f}(0)=|B|$ and $q^{n}\widehat{g}(0)=q^{n-(k-1)m}$).}\nonumber
	\end{align}

	We now analyze the Fourier coefficients of $g$ and use this to bound the right hand side above. Roughly the claim below establishes basic properties of the function $g$ that show that $g$ is also a somewhat reduced function (as in \cref{def:reducedness}). This, combined with the boundedness of $B$ allows us to establish the near-uniformity of the posterior distribution. 
	\begin{claim}\label{claim:g hat combo}
		Let $M$ be a $k$-hypermatching of size $m$, $\vecc$ be centers, and $\vecz_0\in\Z_q^{(k-1)m}$. Let $g(\vecx)=\mathbf{1}_{A_{\vecc}\vecx=\vecz_0}$. For every $\vecu\in\Z_q^n$, the following conditions hold:
		
		\begin{enumerate}
			\item If $\supp(\vecu) \not\subseteq \supp(M)$ then $\widehat{g}(\vecu)=0$.~\footnote{Recall that $\supp(\vecu) = \{i| u_i \ne 0\}$ and $\supp(M)$ is the subset of [n] consisting of vertices that are incident to some hyperedge in the matching $M$.}
			\item If there exists $i\in[m]$ such that $\langle \vecu, \vece_i \rangle \not\equiv0$ mod $q$ where $\vece_i$ denotes the $i$-th hyperedge of $M$, then $\widehat{g}(\vecu)=0$.
			\item $|\widehat{g}(\vecu)|\leq q^{-(k-1)m}$.
		\end{enumerate}
	\end{claim}
	\begin{proof}[Proof of~\autoref{claim:g hat combo}]
		Recall that 
		$q^{n}\widehat{g}(\vecu)=\sum_{\vecx}g(\vecx)\omega^{\vecu^\top \vecx}$.
		\begin{enumerate}
			\item           
			If $\supp(\vecu)\not\subseteq \supp(M)$, then there exists $i\in[n]$ such that $u_i\neq0$ but the $i$-th column of $A_{\vecc}$ is zero. For each $\vecx\in\Z_q^n$, for every $a\in\Z_q$ we have $g(\vecx)=g(\vecx+a\vecdelta_i)$, where $\vecdelta_i \in \Z_q^n$ denotes the coordinate vector in the $i$-th direction (i.e., $\vecdelta_i = 0^{i-1}10^{n-i}$). Also, note that $\sum_{a\in\Z_q}\omega^{\vecu^\top (\vecx+a\vecdelta_i)}=\omega^{\vecu^\top\vecx}\sum_{a\in\Z_q}\omega^{u_i\cdot a}=0$. This implies $\hat{g}(\vecu)=0$.
			
			\item Suppose $\langle \vecu, \vece_i \rangle \not\equiv0$ mod $q$. For each $\vecx\in\Z_q^n$ and $a\in\Z_q$, note that $g(\vecx)=g(\vecx+a\vece_i)$ because $a^k$ lies in the kernel of the folded matrix of this hyperedge. Second, since $\langle \vecu, \vece_i \rangle \not\equiv 0$ mod $q$, we have $\sum_{a\in\Z_q}\omega^{\vecu^\top \vecx+a\vece_i}=\omega^{\vecu^\top\vecx}\sum_{a\in\Z_q}\omega^{a\cdot\langle \vecu, \vece_i \rangle}=0$. This implies $\hat{g}(\vecu)=0$.
			
			\item By definition, we have $q^{n}\widehat{g}(\vecu)=\sum_{\vecx}\mathbf{1}_{A_{\vecc}\vecx=z_0}\omega^{\vecu^\top \vecx}$. Note that for fixed $M,\vecc,z_0$, there are at most $q^{n-(k-1)m}$ $\vecx$ such that $g(\vecx)=1$. Thus, we have $|\widehat{g}(\vecu)|\leq q^{-(k-1)m}$ as desired.
		\end{enumerate}
	\end{proof}
	
	Now, we can use~\autoref{claim:g hat combo} to further upper bound~\autoref{eq:diff} as follows. Recall that $\odot$ stands for the coordinate-wise product of vectors.
	\[
	\frac{q^{n+(k-1)m}}{|B|}\sum_{\substack{\vecu\in\Z_q^n\\\vecu\neq0^n}}\widehat{f}(\vecu)\widehat{g}(\vecu)\leq\frac{q^{n}}{|B|}\sum_{\substack{\vecu\in\Z_q^n\\\vecu\neq0^{n}\\\vecu\text{ is matched by }M\\ \|\vecu\odot\vece_i\|_1\equiv0\text{ mod }q\ \forall i\in[m]}}|\widehat{f}(\vecu)| \, .
	\]
	One key observation here is that the above bound is independent of $\vecz_0$ and therefore holds even if we take the maximum of the left hand side over all $\vecz_0\in\Z_q^{(k-1)m}$. We thus get, for every $M$ and $\vecc$:
	\[\max_{\vecz_0\in\Z_q^{(k-1)m}}\left|q^{(k-1)m}\Pr_{\substack{\vecx\sim \Unif(B)}}[A_{\vecc}\vecx=-\vecz_0]-1\right| \leq\frac{q^{n}}{|B|}\sum_{\substack{\vecu\in\Z_q^n\\\vecu\neq0^{n}\\\vecu\text{ is contained in }M\\ \|\vecu\odot\vece_i\|_1\equiv0\text{ mod }q\ \forall i\in[m]}}|\widehat{f}(\vecu)| \, 
	\]
	Finally, let us take the expectation of the above quantity over the randomness of $M$ and $\vecc$.
	
	\begin{align}
		&\Exp_{M,\vecc}\left[\max_{\vecz_0\in\Z_q^{(k-1)m}}\left|q^{(k-1)m}\Pr_{\substack{\vecx\sim \Unif(B)}}[A_{\vecc}\vecx=-\vecz_0]-1\right|\right] \leq\Exp_{M,\vecc}\left[\frac{q^{n}}{|B|}\sum_{\substack{\vecu\in\Z_q^n\\\vecu\neq0^{n}\\\vecu\text{ is contained in }M\\ \|\vecu\odot\vece_i\|_1\equiv0\text{ mod }q\ \forall i\in[m]}}|\widehat{f}(\vecu)|\right] \,.  \nonumber
		\intertext{Next, we partition the summation according to the $\ell_0$-norm of the Fourier coefficients.}
		&\leq\sum_{h=1}^n\Exp_{M,\vecc}\left[\frac{q^{n}}{|B|}\sum_{\substack{\vecu\in\Z_q^n\\\|\vecu\|_0=h\\\vecu\text{ is contained in }M\\\|\vecu\odot\vece_i\|_1\equiv0\text{ mod }q\ \forall i\in[m]}}|\widehat{f}(\vecu)|\right] \,  \nonumber
		\intertext{Observe that the event $\|\vecu\odot\vece_i\|_1\equiv 0\pmod{q}$ implies that either $\|\vecu\odot\vece_i\|_0=0$ or $\|\vecu\odot\vece_i\|_0\geq 2$ holds. Hence, the above summation can be replaced with a summation beginning at $h=2$, and the equation becomes}
		&\leq\sum_{h=2}^n\Exp_{M,\vecc}\left[\frac{q^{n}}{|B|}\sum_{\substack{\vecu\in\Z_q^n\\\|\vecu\|_0=h\\\vecu\text{ is contained in }M\\\|\vecu\odot\vece_i\|_0=0\text{ or }\|\vecu\odot\vece_i\|_0\geq2\ \forall i\in[m]}}|\widehat{f}(\vecu)|\right] \leq\sum_{h=2}^n p(h,k,m,n)\frac{q^{n}}{|B|}\sum_{\substack{\vecu\in\Z_q^n\\\|\vecu\|_0=h}}|\widehat{f}(\vecu)| \nonumber \,
		\intertext{When $B$ is $(C,s)$-bounded, we can further upper bound the above quantity as follows.}
		&\leq\sum_{h=2}^{n}p(h,k,m,n)\cdot U_{C,s}(h) \leq\delta^2, \nonumber 
	\end{align}
	where the last inequality is due to~\autoref{lem:boundedness implies uniform calculation}. Thus, when $\vecx\sim\Unif(B)$ and $U\sim\Unif(\Z_q^{(k-1)\alpha n})$, we have
	\[
	\Exp_{M,\vecc}\left[\max_{\vecz_0\in\Z_q^{(k-1)m}}\left|q^{(k-1)m}\Pr_{\substack{\vecx\sim \Unif(B)}}[A_{\vecc}\vecx=-\vecz_0]-1\right|\right]\leq\delta^2 \, . 
	\]
	By Markov's inequality, we have
	\[
	\max_{\vecz_0\in\Z_q^{(k-1)m}}\left|q^{(k-1)m}\Pr_{\substack{\vecx\sim \Unif(B)}}[A_{\vecc}\vecx=-\vecz_0]-1\right|\leq\delta
	\]
	with probability at least $1-\delta$. This yields the main part of the lemma. The consequences follow directly from the main part (since pointwise bounds on the distance between distributions imply total variation distance as well as expectation of a non-negative weight).
	
    This completes the proof of~\cref{lem:boundedness implies uniform}.
\end{proof}

\subsection{Proof of the ``induction step'' lemma}\label{sec:induction}
The goal of this section is to prove the ``induction step'' lemma. By Markov's inequality, it suffices to prove the following lemma which is the expectation version of~\cref{lem:induction step}. We first show how~\cref{lem:induction step exp} implies~\cref{lem:induction step} and then focus on proving the former in the rest of this subsection.

\begin{restatable}[Induction step in expectation]{lemma}{inductionstepexp}\label{lem:induction step exp}
	For every $q,k\in\N$ there exist $\alpha_0>0$ and $C_0>0$ such that for every $C > C_0$, there exist $\tau_0\in(0,1)$  and $C''>0$ such that the following holds:
	For every $n,m,s,h \in \N$ satisfying  $m \leq \alpha_0 n$ and  $0 <  s<\tau_0 n$ and $1 \leq h \leq s$, and every $B\subset\Z_q^n$ that 
	is $(C,s)$-strongly-bounded we have:
	\[
	\sum_{\substack{\vecu\in\Z_q^n}}\frac{q^n}{|B|}\left|\widehat{\mathbf{1}_B}(\vecu)\right|\Exp_M\left[\max_{B'} \left\{ \sum_{\substack{\vecu'\in\Z_q^n\\\|\vecu+\vecu'\|_0=h}}\frac{q^n}{|B'|}\left|\widehat{\mathbf{1}_{B'}}(\vecu')\right|\right\}\right]\leq W_{C'',s}(h) \, ,
	\]
	where the expectation is taken over a uniform random $k$-hypermatching $M$ on $m$ hyperedges, and the maximum is taken over all $B'$ that are $(M,C_0,s)$-reduced.
\end{restatable}

We first restate and prove \cref{lem:induction step} using \cref{lem:induction step exp}.
\leminductionstep*
\begin{proof}[Proof of \cref{lem:induction step}]
    Let $\alpha_0$, $C_0$ be as in \cref{lem:induction step exp}. Given $C$ and $\delta$, let $C''$ and $\tau_0$ be the constants given by 
    \cref{lem:induction step exp}. Let $C'$ be the constant from \cref{lem:weak-vs-strong} for $C = C''/\delta^2$ and $\epsilon_0 = \tau_0$. 
    We prove our lemma with these choices of parameters.
    
    For every matching $M$, fix a set $B'=B'(M)$ that is $(M,C_0,s)$-reduced and satisfies $|B'|\geq q^{n-b'}$ and $|B\cap B'|\geq(1-\delta)\cdot|B|\cdot|B'|/q^n \geq q^{n-s}$. We prove the lemma for every such fixing. (In particular $B'$ below is short for $B'(M)$.)
    
    Fix $s \leq \tau_0 n$.
	For every $h\in\{1,\dots,s\}$, by the convolution theorem (see \cref{lem:convthm}) for Fourier coefficients, we have
	\begin{align*}
		\sum_{\substack{\vecv\in\Z_q^n\\\|\vecv\|_0=h}}\frac{q^n}{|B\cap B'|}\left|\widehat{\mathbf{1}_{B\cap B'}}(\vecv)\right|
		& = \sum_{\substack{\vecv\in\Z_q^n\\\|\vecv\|_0=h}}\frac{q^n}{|B\cap B'|}\left|\sum_{\vecu\in\Z_q^n} \widehat{\mathbf{1}_B}(\vecu)\widehat{\mathbf{1}_{B'}}(\vecv-\vecu)\right| \\
		& \leq \sum_{\substack{\vecv\in\Z_q^n\\\|\vecv\|_0=h}}\frac{q^n}{|B\cap B'|}\sum_{\vecu\in\Z_q^n} \left|\widehat{\mathbf{1}_B}(\vecu)\right|\cdot\left|\widehat{\mathbf{1}_{B'}}(\vecv-\vecu)\right| \\
		& = \sum_{\vecu\in\Z_q^n} \sum_{\substack{\vecu'\in\Z_q^n\\\|\vecu'+\vecu\|_0=h}}\frac{q^n}{|B\cap B'|} \left|\widehat{\mathbf{1}_B}(\vecu)\right|\cdot\left|\widehat{\mathbf{1}_{B'}}(\vecu')\right| \\
		&= \frac{|B|\cdot|B'|}{q^n\cdot|B\cap B'|}\sum_{\substack{\vecu\in\Z_q^n}}\frac{q^n}{|B|}\left|\widehat{\mathbf{1}_B}(\vecu)\right|\sum_{\substack{\vecu'\in\Z_q^n\\\|\vecu+\vecu'\|_0=h}}\frac{q^n}{|B'|}\left|\widehat{\mathbf{1}_{B'}}(\vecu')\right|\\
		&\leq \frac{1}{1-\delta}\sum_{\substack{\vecu\in\Z_q^n}}\frac{q^n}{|B|}\left|\widehat{\mathbf{1}_B}(\vecu)\right|\sum_{\substack{\vecu'\in\Z_q^n\\\|\vecu+\vecu'\|_0=h}}\frac{q^n}{|B'|}\left|\widehat{\mathbf{1}_{B'}}(\vecu')\right| \, ,
	\end{align*}
where the first inequality above is from the triangle inequality and the second from the assumption in the lemma statement on the cardinality of $B \cap B'$.

	For $h \in [s]$, let $F(h)$ denote the event that the random matching $M$ is such that
	\[
	\frac{1}{1-\delta}\sum_{\substack{\vecu\in\Z_q^n}}\frac{q^n}{|B|}\left|\widehat{\mathbf{1}_B}(\vecu)\right|\sum_{\substack{\vecu'\in\Z_q^n\\\|\vecu+\vecu'\|_0=h}}\frac{q^n}{|B'|}\left|\widehat{\mathbf{1}_{B'}}(\vecu')\right|>\frac{1}{\delta^h} \cdot W_{C'',s}(h) =  W_{C''/\delta^2, s}(h).
	\]
    Further, for $h \in (s,\tau_0 n]$, let $F(h)$ denote the event that the random matching $M$ is such that 
	\begin{equation}
	\frac{1}{1-\delta}\sum_{\substack{\vecu\in\Z_q^n}}\frac{q^n}{|B|}\left|\widehat{\mathbf{1}_B}(\vecu)\right|\sum_{\substack{\vecu'\in\Z_q^n\\\|\vecu+\vecu'\|_0=h}}\frac{q^n}{|B'|}\left|\widehat{\mathbf{1}_{B'}}(\vecu')\right|>\frac{1}{\delta^h} \cdot W_{C'',h}(h) =  W_{C''/\delta^2, h}(h). \label{eq:Fh-high}	    
	\end{equation}

	Let $F = \cup_{h \in [\tau_0 n]} F(h)$ be the union of these events. Note that if $F$ does not hold, then, for every $s' \in [s,\tau_0n]$, $B \cap B'$ is $(C''/\delta^2,s')$-weakly-bounded, and so, by \cref{lem:weak-vs-strong}, $B\cap B'$ is $(C',s)$-strongly-bounded as desired. So we turn to bounding the probability of $F$.
	
	For $h \in [s]$, an application of Markov's inequality to \cref{lem:induction step exp} yields that 
\[
	\Pr[F(h)] \leq \frac{W_{C'',s}(h)}{(1-\delta)\cdot W_{C''/\delta^2,s}(h)} = \frac{W_{C'',s}(h)}{(1-\delta)\cdot \frac{1}{\delta^h}W_{C'',s}(h)} \leq \frac{\delta^h}{1-\delta}.
\]
	For $h \in (s,\tau_0n]$, we first note that since $B$ is $(C,s)$-bounded then it is also $(C,h)$-bounded (by Item (2) of \cref{lem:ucs-monotonic}). Similarly we also have that $B'$ is $(M,C_0,h)$-reduced. This allows us to invoke \cref{lem:induction step exp} with $s_{\cref{lem:induction step exp}} = h$ and then proceed as in the case above. Specifically for this choice of $s_{\cref{lem:induction step exp}}$ we get by \cref{lem:induction step exp} that the expected value of the LHS of \cref{eq:Fh-high} is at most $W_{C'',h}(h)/(1-\delta)$. Now an application of Markov's inequality yields:
	\[
	\Pr[F(h)] \leq \frac{W_{C'',h}(h)}{(1-\delta)\cdot W_{C''/\delta^2,h}(h)} = \frac{W_{C'',h}(h)}{(1-\delta)\cdot \frac{1}{\delta^h}W_{C'',h}(h)} \leq \frac{\delta^h}{1-\delta} \, .
	\]	
	We thus get
	$\Pr[F] \leq \sum_{h \in [\tau_0 n]} F(h) \leq \frac1{1-\delta}\sum_h \delta^h \leq 4\delta$ where the final step uses the fact that $\delta < 1/2$. We conclude that with probability at least $1-4\delta$ over the randomness of $M$, the event $F$ does not hold and $B \cap B'$ is $(C',s)$-bounded.
\end{proof}

Now we turn to proving \cref{lem:induction step exp}. The proof involves three steps. In the first step we partition the inner sum over $\vecu'$ based on a combinatorial structure that allows us to say how much the expected contribution of $\vecu'$ would be, based on a few parameters. In the second step we give bounds on these expected contributions in different cases and analyze the probability of each case. In the final step we then combine these different bounds to prove the lemma.

	\paragraph{Step 1: Partitioning the inner sum via a combinatorial structure.}
	We start by defining the following combinatorial quantity, based on intersection properties of a random $k$-hypermatching.
	
	\begin{definition} \label{def:pcomb}
		Let $n,q,k,u\in\mathbb{N}$ and $\alpha\in(0,1/k)$. Let $\vecu\in(\Z_q\backslash\{0\})^u\times0^{n-u}$ be a vector that is non-zero on exactly the first $u$ coordinates. For a $k$-hypermatching $M$ of size $m$, let $K_{\vecu}(M):=\{i\in[m]\, |\, \langle \vecu, \vece_i \rangle \not\equiv0 \pmod{q}\}$ be the set of edges with ``odd intersection'' (formally non-zero inner product mod $q$) with $\vecu$. 
		Let $E_{\vecu}(M):=\{j\in[n]\, \mid\, u_j\neq0,\, \exists i\not\in K_{\vecu}(M),\ j\in e_i\}$ denote the set of vertices in the support of $\vecu$ that are in ``even'' edges.\footnote{Informally we refer to edges as ``even'' (or ``odd'') which would be the right terminology if $q=2$. For $q \ne 2$ these words are formalized as having zero (or non-zero) inner product with $\vecu$.} Finally, let $O_{\vecu}(M):=\{j\in[n]\, \mid\, u_j\neq0,\, \exists i\in K_{\vecu}(M),\ j\in e_i\}$ be the vertices in the support of $\vecu$ from odd edges. For $o,\eta,\kappa\in\mathbb{N}$, we define 
		\begin{equation}
			p_{q,\alpha}(n,u,o,\eta,\kappa) := \max_{\vecu\in(\Z_q\backslash\{0\})^u\times0^{n-u}}\Pr_{M}\left[|K_{\vecu}(M)| =\kappa, |E_{\vecu}(M)|=\eta, |O_{\vecu}(M)|=o\right], \label{eq:matchprob}
		\end{equation}
		where $M$ is a uniformly random $k$-hypermatching of size $\alpha n$. (In other words $p_{q,\alpha}$ is the maximum probability of a vector $\vecu$ of support size $u$ having $\kappa$ odd edges, $\eta$ even vertices and $o$ odd vertices when the matching $M$ is drawn at random.)
	\end{definition}
	\cref{fig:induction} illustrates some of the parameters in the definition above.
	We remark that $p_{q,\alpha}(\cdots)$ should not be confused with the function $p(\cdots)$ defined in \cref{def:p-one}, which is a similar combinatorial quantity but not the same.

	Note that as each edge in $K_{\vecu}(M)$ contributes at least one element to $O_{\vecu}(M)$, we have $o\geq\kappa$.
    \begin{figure}[ht]
        \centering
        \includegraphics[width=10cm]{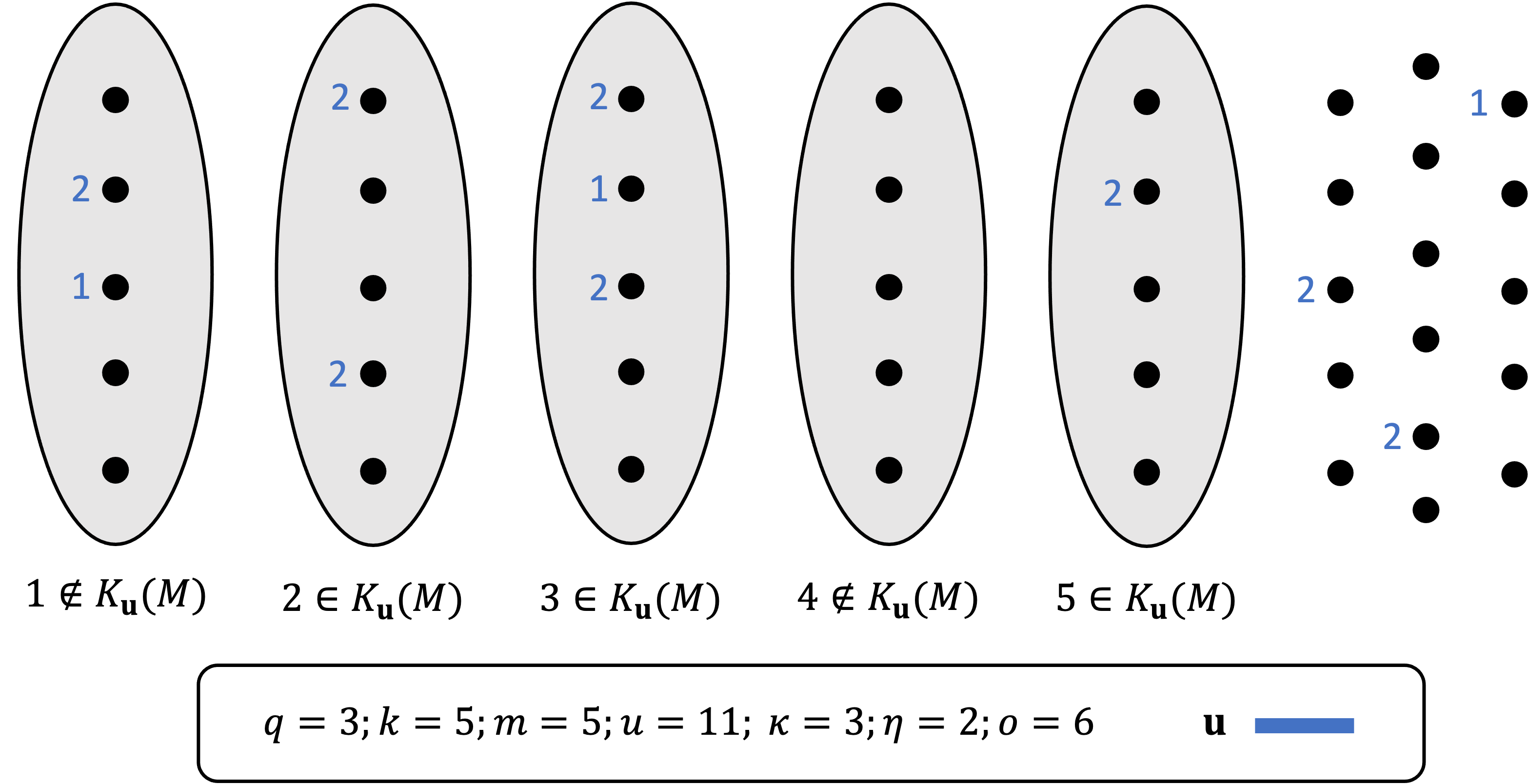}
        \caption{A graphical intuition for the parameters appeared in~\cref{def:pcomb}.}
        \label{fig:induction}
    \end{figure}
	
	We now show how to bound a certain expected value of the sum of Fourier coefficients of a fixed ``level'' from above in terms of the combinatorial quantity defined in \autoref{def:pcomb}.
	\begin{lemma}\label{lem:mcssum}
	    Let $n,q,k,u\in\N$, $\alpha\in(0,1/k)$, $0\leq s\leq n$, and $C>0$.
		For every $\vecu\in\Z_q^n$ with $u = |\supp(\vecu)|$ and $h\in[s]$, we have
		\begin{align*}
			\lefteqn{\Exp_M\left[\max_{\substack{B' \subset \Z_q^n\\ \text{$B'$ is $(M,C,s)$-reduced}}} \left\{ \sum_{\substack{\vecu'\in\Z_q^n\\\|\vecu+ \vecu'\|_0=h}} \frac{q^n}{|B'|}\left|\widehat{\one_{B'}}(\vecu')\right|\right\} \right]}\\
			& ~~~ \leq\ \sum_{o,\eta,\kappa}p_{q,\alpha}(n,u,o,\eta,\kappa)\cdot 
			(h+1) \cdot q^{k\kappa} \cdot U_{C,s}(h+o+\eta-(u+\kappa)) \, ,
		\end{align*}
	  where the summation is taken over $0 \leq o,\eta,\kappa \leq n$ satisfying conditions (1) $u \geq \eta+o$, (2) $\kappa\leq o \leq k \cdot \kappa$ and (3) $h+\eta+o - (u+\kappa) \geq 0$.
	\end{lemma}
	\begin{proof}
		As suggested by the right hand side, we consider the various possibilities for $o,\eta, \kappa$ and bound the left hand side conditioned on the event in \cref{eq:matchprob}, i.e., $|K_{\vecu}(M)| =\kappa, |E_{\vecu}(M)|=\eta$, and $|O_{\vecu}(M)|=o$.
		
		Let $u = |\supp(\vecu)|$. Consider a fixed matching $M = \{e_1,\ldots,e_m\}$ with $m = \alpha n$ and $|K_\vecu(M)| = \kappa$, $|E_\vecu(M)| = \eta$, and $|O_\vecu(M)| = o$ (see the relevant definitions in~\cref{def:pcomb}). Given $M$, let $A = \supp(\vecu) \setminus (E_\vecu(M) \cup O_\vecu(M))$ be the set of unmatched vertices of $\supp(\vecu)$. Furthermore, let $a = |A|$, so that $a = u-(\eta+o)$. For ease of notation, we drop the dependence on $\vecu$ and $M$ and simply write $E = E_\vecu(M)$ and $O = O_\vecu(M)$. We also abuse notation and often use $M$ to denote the subset of $[n]$ given by $\cup_{i \in [m]} e_i$. (The distinction is hopefully clear from context.)
		
        Note that since $|A| \geq 0$ we must have $u-(\eta+o) \geq 0$ for such a matching to exist. 
        This shows it suffices to restrict the summation to triples $(o,\eta,\kappa)$ satisfying condition (1). Note further that each edge in $K_\vecu(M)$ contributes at least one vertex,  and at most $k$ vertices, to $O$ and so $\kappa \leq o \leq k\cdot \kappa$ establishing the sufficiency of summing over triples satisfying condition (2). We now proceed to proving the rest of the lemma (and will prove sufficiency of condition (3) along the way).

		Let $B' \subset \Z_q^n$ be an $(M,C,s)$-reduced set. We give an upper bound on
		\[
		\sum_{\substack{\vecu'\in\Z_q^n\\\|\vecu+ \vecu'\|_0=h}} \frac{q^n}{|B'|}\left|\widehat{\one_{B'}}(\vecu')\right|
		\]
		in terms of the parameters $o,\eta,\kappa$, which will suffice to establish the lemma. We start by establishing some conditions that are necessary to get $\widehat{\one_B}(\vecu') \neq 0$.

We start with some more notation: For a set $S \subseteq [n]$,  we define the restriction of $\vecu$ with respect to $S$ to be the vector $\vecu|_S\in\Z_q^n$ where $(\vecu|_S)_j = \vecu_j$ if $j\in S$; otherwise $(\vecu|_S)_j = 0$. We define the closure of $S$ (with respect to the matching $M$) to be the set $\overline{S} = \cup_{\{i \in [m] \mid S \cap e_i \ne \emptyset\}} e_i$, i.e., $\overline{S}$ takes all the vertices that are contained in edges that touch $S$. (We only apply the notion of the closure to sets $S \subseteq M$.) 

		\begin{claim}\label{claim:mcssum structure}
		    For every $\vecu$ there exists a vector $\tilde{\vecu}$ such that for every vector $\vecu' \in \Z_q^n$ satisfying $\|\vecu+\vecu'\|_0 = h$, we have  $\widehat{\one_{B'}}(\vecu') \neq 0$ only if there exists $\vecz = \vecz(\vecu')\in\Z^n_q$ with $\supp(\vecz)\subseteq \overline{O}$ and 
		    $\tau = \tau(\vecu') \in [\kappa,h-a]$ such that $\|\vecu'+(\vecz+\tilde{\vecu})\|_0=h-a-\tau$. In particular, $\kappa\leq h-a$.
		\end{claim}
		
        Before proving the claim we note that the claim establishes that for there to exist $\vecu'$ such that $\|\vecu + \vecu'\|_0 = h$ and $\widehat{\one_{B'}}(\vecu') \neq 0$ we must have $h - a - \kappa \geq 0$. Combining with $a = u - \eta - o$, this allows us to restrict the summation in the RHS of \cref{lem:mcssum} to triples $(o,\eta,\kappa)$ satisfying $h + \eta + o - (u+\kappa) = h - a - \kappa \geq 0$, thereby establishing the sufficiency of condition (3).
        


        We now prove the claim.
		
		\begin{proof}
		We prove the claim for $\tilde{\vecu} := \vecu|_{[n]\setminus A}$. 
		Note that for every vector $\vecv$ and set $S \subseteq [n]$ we can write $\vecv = \vecv|_S + \vecv_{[n]\setminus S}$, and we also have $\|\vecv\|_0 =\|\vecv|_S\|_0 + \|\vecv_{[n]\setminus S}\|_0$. We use this to decompose $\vecu = \vecu|_A + \tilde{\vecu}$. 
		
		Now consider $\vecu'\in\Z_q^n$ such that $\widehat{\one_{B'}}(\vecu') \neq 0$ and $\|\vecu+\vecu'\|_0=h$. First, as $B'$ is $(M,C,s)$-reduced, by~\cref{lem:fourier coeff of set} we have $\supp(\vecu') \subseteq M$.
		Again we write $\vecu' = \vecu'|_A + \vecu'|_{[n]\setminus A}$. Since $A\cap M = \emptyset$ we must have $\vecu'|_A = 0$. Thus we get that $\tilde{\vecu}+\vecu' = \vecu|_{[n]\setminus A} + \vecu'|_{[n]\setminus A}$ and so 
		\[
		\|\tilde{\vecu}+\vecu'\|_0  = \|\vecu|_{[n]\setminus A} + \vecu'|_{[n]\setminus A}\|_0 = \|\vecu + \vecu'\|_0 - \|\vecu|_{A} + \vecu'|_{A}\|_0 = h - a,
		\]
		where the final equality uses $\|\vecu|_{A} + \vecu'|_{A}\|_0 = \|\vecu|_A\|_0$ which equals $a$ since $A \subseteq \supp(\vecu)$. 
		
        We show now that for $\vecz := - (\tilde{\vecu}+\vecu')|_{\overline{O}}$ 
		and $\tau := \|\tilde{\vecu}|_{\overline{O}}+\vecu'|_{\overline{O}}\|_0$,  we have $\|(\tilde{\vecu}+\vecz) + \vecu'\|_0 = h-a-\tau$. 
		
        Note that the definition of $\vecz$ is such that we have $(\tilde{\vecu}+\vecz+\vecu')|_{\overline{O}} = 0$. This ensures 
        \[\|(\tilde{\vecu}+\vecz) + \vecu'\|_0 = \|\tilde{\vecu}|_{[n]\setminus \overline{O}}+\vecu'|_{[n]\setminus \overline{O}}\|_0 = 
        \|\tilde{\vecu}+\vecu'\|_0 -  \|\tilde{\vecu}|_{ \overline{O}}+\vecu'|_{ \overline{O}}\|_0 = h-a-\tau\,.
        \]
        Finally, we would like to bound the range of possible values for $\tau$. For the upper bound, we have
        \[
        \tau := \|\tilde{\vecu}|_{\overline{O}}+\vecu'|_{\overline{O}}\|_0 \leq  \|\tilde{\vecu}+\vecu'\|_0 = h-a.
        \]
        For the lower bound we for claim that $\|\vecu|_{e_i}+\vecu'|_{e_i}\|_0 \geq 1$ for every edge $e_i$ with $i \in K$. This is so since $\langle\vecu,e_i\rangle \ne 0$ (definition of $K$) and $\langle\vecu',e_i\rangle = 0$ (since $\widehat{\one_{B'}}(\vecu') \neq 0$), and together they imply $\langle\vecu+\vecu',e_i\rangle \ne 0$ which can only happen if $(\vecu+\vecu')|_{e_i} \ne 0$, which in turn implies $\|\vecu|_{e_i}+\vecu'|_{e_i}\|_0 = \|(\vecu+\vecu')|_{e_i}\|_0 \geq 1$.
        
        From the above claim it follows that 
        \[ \|\tilde{\vecu}|_{\overline{O}}+\vecu'|_{\overline{O}}\|_0 = \|\vecu|_{\overline{O}}+\vecu'|_{\overline{O}}\|_0 = \sum_{i \in K} \|\vecu|_{e_i}+\vecu'|_{e_i}\|_0 \geq \sum_{i \in K} 1 = \kappa. \]

        This concludes the proof of the claim. 
		\end{proof}
		 		
		We now return to analyzing the summation in the LHS of the lemma statement. Let $\tilde{\vecu}$ be as given by \cref{claim:mcssum structure}. 
		We have:
		\begin{align}
			\sum_{\substack{\vecu'\in\Z_q^n\\\|\vecu+ \vecu'\|_0=h}} \frac{q^n}{|B'|}\left|\widehat{\one_{B'}}(\vecu')\right| 
			&= \sum_{\tau=\kappa}^{h-a} \sum_{\substack{\vecz \in \Z_q^n\\ \supp(\vecz) \subseteq \overline{O}}} \left[ \sum_{\substack{\vecu'\in\Z_q^n\\ \vecz(\vecu')=\vecz,\tau(\vecu')=\tau\\
			\|(\vecz + \tilde{\vecu})+\vecu'\|_0 = h-a-\tau}} \frac{q^n}{|B'|}|\widehat{\one_{B'}}(\vecu')| \right] \nonumber\\
	        & \mbox{ ~~~~~~~~~~(Using \cref{claim:mcssum structure})}\nonumber\\
	        &\leq \sum_{\tau=\kappa}^{h-a} \sum_{\substack{\vecz \in \Z_q^n\\ \supp(\vecz) \subseteq \overline{O}}} \left[ \sum_{\substack{\vecu'\in\Z_q^n\\ 
			\|(\vecz + \tilde{\vecu})+\vecu'\|_0 = h-a-\tau}} \frac{q^n}{|B'|}|\widehat{\one_{B'}}(\vecu')| \right] \nonumber\\
			&\leq \sum_{\tau=\kappa}^{h-a} \sum_{\substack{\vecz \in \Z_q^n\\ \supp(\vecz) \subseteq \overline{O}}} U_{C,s}(h-a-\tau) \nonumber\\
			& \mbox{ ~~~~~~~~~~(Using the $(C,s)$-reducedness of $B'$ with respect to the vector $\vecv := \tilde{\vecu}+\vecz$)} \nonumber\\
			&\leq \sum_{\tau=\kappa}^{h-a} q^{k\kappa} \cdot U_{C,s}(h-a-\tau) \nonumber\\
			&  \mbox{ ~~~~~~~~~~(Using $|\overline{O}|=k\kappa$ to get $|\{\vecz \, |\,  \supp(\vecz) \subseteq \overline{O}\}|\leq q^{k\kappa}$)}\nonumber\\
			&\leq \sum_{\tau=\kappa}^{h-a} q^{k\kappa} \cdot U_{C,s}(h-a-\kappa) \nonumber\\
			&  \mbox{ ~~~~~~~~~~(Using monotonicity of $U_{C,s}(h)$ when $h\in[s]$ by~\cref{lem:ucs-monotonic})}\nonumber\\
			&= (h-a-\kappa+1) \cdot q^{k\kappa} \cdot U_{C,s}(h-a-\kappa) \nonumber\\
			&\leq (h+1) \cdot q^{k\kappa}\cdot U_{C,s}(h+\eta+o-(u+\kappa)). \nonumber
		\end{align}
This proves the lemma.		
	\end{proof}

	\paragraph{Step 2: Useful inequalities about the boundedness parameters and the combinatorial structure.}
	
	In order to quantify the upper bound in~\autoref{lem:mcssum}, we need to obtain an upper bound for the combinatorial quantity $p_{q,\alpha}(n,u,o,\eta,\kappa)$.
	\begin{restatable}{lemma}{qbound}\label{lem:q-bound}
		For every $q, k\in\N$ there exists a constant $C$ such that for every $\alpha\in (0,1/k]$ and every $n,u,\kappa,o,\eta\in\N$ we have:
		\[
		p_{q,\alpha}(n,u,o,\eta,\kappa) \leq \alpha^{(o+\eta)/k}\cdot C^u \cdot (n/\kappa)^\kappa \cdot (u/\sqrt{n\eta})^\eta \cdot (u/n)^o \, . 
		\]
	\end{restatable}

\begin{proof} We prove the lemma for $C = 2qe^3k$.
We start by establishing some (significant amount of) notation for the proof. The proof consists of two steps: (i) upper bounding $p_{q,\alpha}(\cdots)$ by $\sum_{d=\eta/\kappa}^{\eta/2}N_q(u,d,o,\eta,\kappa)/\binom{n}{u}$ where $N_q(\cdots)$ is a certain well-defined combinatorial quantity and (ii) upper bounding $N_q(\cdots)$.

\paragraph{Step (i) of the proof for~\cref{lem:q-bound}.}
For $\vecu = (\vecu_1,\ldots,\vecu_n) \in \Z_q^n$, let $\supp(\vecu)\subseteq[n]$ denote the subset of non-zero coordinates of $\vecu$. Further, for $i\in\Z_q$, let $\supp_i(\vecu)$ denote the subset $\{j \in [n]\mid\ \vecu_j = i\}$. 
Now given non-negative integers $u_1,\ldots,u_{q-1}$ and $u = u_1 + \cdots +u_{q-1}$, let $S_{u_1,\ldots,u_{q-1}} = \{\vecu \in \Z_q^n \mid |\supp_i(\vecu)| = u_i \, \forall i \in [q-1]\}$ and let $S_u = \{\vecu\in \Z_q^n \mid |\supp(\vecu)| = u \}$. 

Given a vector $\vecu \in \Z_q^n$ and hypermatching $M$ containing $m = \alpha n$ hyperedges $e_1,\ldots,e_m$ where each $e_i$ is viewed as a subset of $[n]$ of size $k$, we define four associated sets below. Let: 
\begin{itemize}
		\item $K = \{i \in [m] \mid \langle \vecu, \vece_i \rangle \not\equiv 0 \pmod{q}\}$,
		\item $O = \{j\in\supp(\vecu) \mid \exists i\in K, j\in e_i\}$,  
		\item $E = \{j\in\supp(\vecu) \mid \exists i\in [m]\setminus K, j\in e_i \}$, and
		\item $D = \{i \in [m]: E \cap e_i \neq \emptyset\}.$
\end{itemize}	

Note that $p_{q,\alpha}(\cdots)$ bounds the maximum over $\vecu$ with $|\supp(\vecu)|=u$ of the probability, over a random hypermatching $M$, that $|K|=\kappa$, $|O|=o$ and $|E| = \eta$. By symmetry however we can fix the matching $M$ and consider the maximum, over $u_1,\ldots,u_{q-1}$ s.t. $u_1 + \cdots +u_{q-1} = u$, of the probability that $|K|=\kappa$, $|O|=o$ and $|E| = \eta$, when $\vecu$ is chosen uniformly from $S_{u_1,\ldots,u_{q-1}}$. In notation, we have
\[
	p_{q,\alpha}(n,u,o,\eta,\kappa)  = \max_{\{u_1,\ldots,u_{q-1} \mid u_1 + \cdots + u_{q-1} = u\}}\left\{\Pr_{\vecu\in S_{u_1,\ldots,u_{q-1}}} [\mathcal{E}(\vecu, o, \eta, \kappa)]\right\} 
\]
where $\mathcal{E}(\vecu, o, \eta, \kappa)$ is the event that $|K| = \kappa$, $|O| = o$, and $|E| = \eta$.
Now let $\cE_d(\vecu, o, \eta, \kappa)$ denote the event that $|K| = \kappa$, $|O|=0$, $|E|=\eta$, and $|D| = d$. Note that each hyperedge in $D$ contributes at least two elements to $E$ (since $\langle \vecu, \vece_i \rangle \equiv 0 \pmod{q}$ for $i\in D$). Hence, $d\leq\eta/2$. Moreover, as each edge in $D$ can contribute at most $k$ elements to $E$, we also have $d \geq \eta/k$. Thus we get:
\begin{align*}
		p_{q,\alpha}(n,u,o,\eta,\kappa) & = \max_{\{u_1,\ldots,u_{q-1} \mid u_1 + \cdots + u_{q-1} = u\}}\left\{\sum_{d=\eta/k}^{\eta/2} \Pr_{\vecu\in S_{u_1,\ldots,u_{q-1}}} [\mathcal{E}_d(\vecu,o,\eta,\kappa)]\right\}\\
		& \leq  \sum_{d=\eta/k}^{\eta/2} \left( \max_{\{u_1,\ldots,u_{q-1} \mid u_1 + \cdots + u_{q-1} = u\}}\left\{ \Pr_{\vecu\in S_{u_1,\ldots,u_{q-1}}} [\mathcal{E}_d(\vecu,o,\eta,\kappa)]\right\} \right)\,. \\
\end{align*}

Define $T^+((u_1,\ldots,u_{q-1}),d,o,\eta,\kappa)$ to be the set $\{\vecu_0 \in S_{u_1,\ldots,u_{q-1}} \mid |K| = \kappa, |O|=o, |E| = \eta, |D| = d\}$,
let $T_q(u,d,o,\eta,\kappa) = \cup_{\{u_1,\ldots,u_{q-1} \mid u_1 + \cdots + u_{q-1} = u\}}T^+((u_1,\ldots,u_{q-1}),d,o,\eta,\kappa)$. Intuitively, $T_q$ is the set that contains all the possible $\vecu_0$ in the event $\mathcal{E}_d$ while $T^+$ forms a partition for $T_q$. For every $u_1,\ldots,u_{q-1}$ we have 
\[
\Pr_{\vecu\in S_{u_1,\ldots,u_{q-1}}} [\mathcal{E}_d(\vecu,o,\eta,\kappa)] = \sum_{\vecu_0 \in T^+((u_1,\ldots,u_{q-1}),d,o,\eta,\kappa)} \Pr_{\vecu\in S_{u_1,\ldots,u_{q-1}}} [\vecu = \vecu_0].
\]
The final probability above $\Pr_{\vecu\in S_{u_1,\ldots,u_{q-1}}} [\vecu = \vecu_0]$ is upper bounded by $1/\binom{n}u$. ($\vecu$ is chosen by picking disjoint sets $U_1,\ldots,U_{q-1}$ uniformly subject to $|U_i|=u_i$. The event $\vecu = \vecu_0$ holds iff $U_i = \supp_i(\vecu_0)$ which in turn happens only if 
$\cup_i U_i = \supp(\vecu_0)$ which in turn happens with probability $1/\binom{n}u$.) Finally let $N_q(u,d,o,\eta,\kappa)=|T_q(u,d,o,\eta,\kappa)|$.
We thus have 
\begin{align}
		p_{q,\alpha}(n,u,o,\eta,\kappa) 
		& \leq  \sum_{d=\eta/k}^{\eta/2} \left( \max_{\{u_1,\ldots,u_{q-1} \mid u_1 + \cdots + u_{q-1} = u\}}\left\{ \Pr_{\vecu\in S_{u_1,\ldots,u_{q-1}}} [\mathcal{E}_d(\vecu,o,\eta,\kappa)]\right\} \right)\nonumber\\
		& \leq  \sum_{d=\eta/k}^{\eta/2} \left( \max_{\{u_1,\ldots,u_{q-1} \mid u_1 + \cdots + u_{q-1} = u\}}\left\{ \sum_{\vecu_0 \in T^+((u_1,\ldots,u_{q-1}),d,o,\eta,\kappa)} \Pr_{\vecu\in S_{u_1,\ldots,u_{q-1}}} [\vecu = \vecu_0].\right\} \right)\nonumber\\
		& \leq  \sum_{d=\eta/k}^{\eta/2} \left( \max_{\{u_1,\ldots,u_{q-1} \mid u_1 + \cdots + u_{q-1} = u\}}\left\{ |T^+((u_1,\ldots,u_{q-1}),d,o,\eta,\kappa)| \cdot \frac1{\binom{n} u} \right\} \right)\nonumber\\
		& \leq  \sum_{d=\eta/k}^{\eta/2} \left( \max_{\{u_1,\ldots,u_{q-1} \mid u_1 + \cdots + u_{q-1} = u\}}\left\{ N_q(u,d,o,\eta,\kappa) \cdot \frac1{\binom{n} u} \right\} \right)\nonumber\\
		& =  \sum_{d=\eta/k}^{\eta/2} \frac{N_q(u,d,o,\eta,\kappa)}{ \binom{n} u } \, . \label{eq:qcountingbound}
\end{align}
Thus to upper bound $p_{q,\alpha}(\cdots)$ it suffices to upper bound $N_q(\cdots)$. 

\paragraph{Step (ii) of the proof for~\cref{lem:q-bound}.}
A vector $\vecu \in T_q(u,d,o,\eta,\kappa)$ can be specified by specifying the sets $O$, $E$, $\supp(\vecu)- (O \cup E)$, and then by specifying $\vecu|_{\supp(\vecu)}$ i.e., the restriction of $\vecu$ to $\supp(\vecu)$. There are $(q-1)^u$ choices of $\vecu|_{\supp(\vecu)}$. So we turn to counting the number of possible $O$'s and $E$'s. $O$ may be specified by first specifying $K$ and then selecting $O$ from $\cup_{i \in K} e_i$. (There are further restrictions on the choices of $O$ which we will ignore to get an upper bound.) There are $\binom{m}{\kappa}$ choices of $K$ and at most $\binom{k\kappa}{o}$ choices of $O$ given $K$. Similarly for $E$ we have at most $\binom{m}{d}$ choices of $D$ and then at most $\binom{kd}{\eta}$ choices of $E$ given $D$. 
Finally, there are at most $\binom{n-km}{u-o-\eta}$ choices of $\supp(\vecu)\setminus (O\cup E)$, since they must be a set of $u-o-\eta$ vertices outside the $m$ edges of  our hypermatching. Putting all this together we get the following upper bound on $N_q(\cdots)$:
	\[
	N_q(u,d,o,\eta,\kappa) \leq \binom{\alpha n}{\kappa} \binom{\alpha n}{d} \binom{k\kappa}{o} \binom{kd}{\eta} \binom{n(1-\alpha k)}{u-o-\eta}(q-1)^u.
	\]
	Using the bounds $\left(\frac{a}{b}\right)^b \leq \binom{a}{b} \leq \left(\frac{ea}{b}\right)^b$, we have that
	\begin{align*}
		\frac{N_q(u,d,o,\eta,\kappa)}{\binom{n}{u}} 
		&\leq \left( \frac{e\alpha n}{\kappa}\right)^\kappa \left(\frac{e\alpha n}{d}\right)^d \left(\frac{ek\kappa}{o}\right)^o \left(\frac{ekd}{\eta}\right)^\eta \frac{(n(1-\alpha k))^{u-o-\eta}}{(u-o-\eta)!}(q-1)^u\cdot\left(\frac{u}{n}\right)^u\\
		&\leq n^{\kappa + d - \eta - o} \kappa^{-\kappa} d^{-d} \frac{u^u}{(u-o-\eta)!} \left(\alpha^{\kappa + d} e^{\kappa + d + o + \eta} k^{o+\eta} (q-1)^u  \left(\frac{\kappa}{o}\right)^o \left(\frac{d}{\eta}\right)^\eta (1-k\alpha)^{u-o-\eta} \right) \, .
	\end{align*}
	Recall from~\cref{lem:mcssum} and step (i) of the proof that $\kappa \leq o$, $2d \leq \eta$, and $o+\eta \leq u$. Hence, we have that $0\leq\kappa/o,d/\eta\leq1$ and $e^{\kappa+d+o+\eta} k^{o+\eta} (q-1)^u \leq e^{2u} \cdot ((q-1)k)^u \leq (q e^2 k)^u$. Moreover,
	\begin{align*}
		\frac{u^{u-o-\eta}}{(u-o-\eta)!} \leq e^u.
	\end{align*}
	Therefore, letting $C_k = q e^3 k$, we have 
	\begin{align*}
		\frac{N_q(u,d,o,\eta,\kappa)}{\binom{n}{u}} &\leq C_k^u \alpha^{\kappa+d} \cdot n^{\kappa + d - \eta - o} \kappa^{-\kappa} d^{-d} u^{o+\eta} = \left(\frac{\alpha n}{d}\right)^d \cdot C_k^u \alpha^\kappa \cdot n^{\kappa - \eta - o} \kappa^{-\kappa} u^{o+\eta} \, .
	\end{align*}
	Hence, by \eqref{eq:qcountingbound}, we have that for $C=2C_k = 2qe^3k$,
	\begin{align*}
		p_{q,\alpha}(n,u,o,\eta,\kappa) &\leq C_k^u \alpha^\kappa \cdot n^{\kappa - \eta - o} \kappa^{-\kappa} u^{o+\eta} \sum_{d=\eta/k}^{\eta/2} \left(\frac{\alpha n}{d}\right)^d \\
		&\leq C_k^u \alpha^{\kappa + \frac{\eta}{k}} \cdot n^{\kappa - \eta - o} \kappa^{-\kappa} u^{o+\eta} \sum_{d=\eta/k}^{\eta/2} \left(\frac{n}{d}\right)^d \\
		&\leq C_k^u \alpha^{\frac{\eta+o}{k}} \cdot n^{\kappa - \eta - o} \kappa^{-\kappa} u^{o+\eta} \cdot \frac{\eta}{2} \left(\frac{2n}{\eta}\right)^{\eta/2} \mbox{~~~(Using $o \leq k \cdot \kappa$ from \cref{lem:mcssum})}\\
		&\leq \alpha^{\frac{\eta+o}{k}} \cdot C^u \cdot (n/\kappa)^\kappa \cdot (u/\sqrt{n\eta})^\eta \cdot (u/n)^o,
	\end{align*}
	where the second-to-last inequality follows from the fact that $n/d \geq e$ and $x^{1/x}$ is a decreasing function of $x$ on $x \in (e,\infty)$. This completes the proof of~\autoref{lem:q-bound}. 
\end{proof}

Finally, we prove an additional inequality about the boundedness parameters. This will simplify the final proof of~\autoref{lem:induction step exp}.

\begin{restatable}{lemma}{qsimplify}\label{lem:q-simplify}
		For every $q,k\in\N$, there exists $\alpha_0\in(0,1/k)$ so that the following holds. For every  $C_1,C_2>0$ there exists $\epsilon_0 > 0$ and $C_3>0$ such that for every $\alpha\in (0,\alpha_0)$, $\epsilon\in(0,\epsilon_0)$ and $s,n,u,h,\eta,o,\kappa\in\N$ with $s =  \epsilon n \leq \epsilon_0 n$ and $h\in[s]$ and $u \in [n]$,  we have 
		\[
		U_{C_1,s}(u) \cdot p_{q,\alpha}(n,u,o,\eta,\kappa) \cdot h \cdot q^{k\kappa} \cdot U_{C_2,s}(h+\eta+o-(u+\kappa)) \leq 4^{-u-2} W_{C_3,s}(h),
		\]
        for every $0 \leq o,\eta,\kappa \leq n$ satisfying (1) $u \geq \eta+o$, (2) $\kappa\leq o \leq k \cdot \kappa$ and (3) $h+\eta+o - (u+\kappa) \geq 0$.
\end{restatable}

\begin{proof}
Given $q$ and $k$, 
let $C$ be the constant from \autoref{lem:q-bound}. 
Let $C_0 = \sqrt{2}\cdot 4e \cdot q\cdot C \cdot q^k$, and
let  $\alpha_0 = 1/(e^2C_0)^k$. 
Now given $C_1,C_2$, let $C_4 = 4\cdot \sqrt{C_1}\cdot C \cdot q^k$, $C_5 = \max\{1,2\sqrt{C_2},2C_2\}$ and 
$C_6 = e C_5$ 
(where $e$ is the base of the natural logarithm). 
Now let $\epsilon_0 = \min\left\{\frac{1}{(e^2C_4)^4},\frac{1}{(2 e^2C_4)^4},(e^{-2}/C_4)^{16}\right\}$ 
and $C_3 = \max\{(16 e^3 C_4 C_6)^2,(16C_6/\epsilon_0^{1/4})^2,(32C_6/\alpha_0^{1/k})^2,256C_6^2,(16 C_6/\alpha_0^{1/k})^2\}$.
We prove the lemma for this choice of $\alpha_0$, $\epsilon_0$ and $C_3$. 
Note in particular that this choice of $\alpha_0$ depends only on $q$ and $k$ but not on $C_1$ and $C_2$ (as required). 

Let $h'=h+\eta+o-(u+\kappa)$. By the conditions (1) and (3) in the lemma statement we have $0 \leq h' \leq h$. 

We divide the analysis into five cases depending on the choice of $u$. (The cases differ first because $U_{C_1,s}(u)$ differs in behavior depending on whether $u \leq s$ or not. Further differences arise in the analysis depending on the relationship between $u$ and $h$, as also how close $u$ is to $s$.)
The five cases are: (1) $1 \leq u \leq h$, (2) $h < u \leq s$, (3) $s < u \leq 16s$, (4) $16s < u \leq \sqrt{\epsilon}n$, and (5) $\sqrt{\epsilon}n < u \leq n$. 

\paragraph{Case 1: $1 \leq u \leq h$:} 

Expanding the definition of $U_{C_1,s}(u)$, $U_{C_2,s}(h')$, $W_{C_3,s}(h)$ and invoking the upper bound on $p_{q,\alpha}(n,u,o,\eta,\kappa)$ from \cref{lem:q-bound}, we have that it suffices to prove that:

\begin{align}
		& \left(C_1^{u/2} ((s n)/u^2)^{u/4}\right) \left(\alpha^{(o+\eta)/k} C^u (n/\kappa)^\kappa (u/\sqrt{n\eta})^\eta (u/n)^o\right) \left(C_2^{h'/2} (s n/h'^2)^{h'/4}\right) \cdot h \cdot q^{k\kappa} \nonumber\\ 
		& ~~~~~~~~~~ \leq 4^{-u-2}\cdot C_3^{h/2} \cdot (s n/h^2)^{h/4} = 4^{-u-2}\cdot U_{C_3,s}(h)\, , \label{eq:case1}
\end{align}
We multiply the LHS above by $4^{u+2} (h^2/sn)^{h/4}$ and show it is upper bounded by $C_3^{h/2}$:

\begin{align*}
L_1
& :=  16 (16C_1)^{u/2} (sn/u^2)^{u/4} \cdot \alpha^{(o+\eta)/k} C^u (n/\kappa)^\kappa (u/\sqrt{n\eta})^\eta (u/n)^o \cdot C_2^{h'/2} (sn/h'^2)^{h'/4}\cdot h \cdot q^{k\kappa} \cdot (h^2/(sn))^{h/4} \\
		& \leq  16 C_4^u C_5^{h} \cdot (sn/u^2)^{u/4} \cdot (n/\kappa)^\kappa \cdot (u/\sqrt{n\eta})^\eta \cdot (u/n)^o  \cdot  (sn/h'^2)^{h'/4}\cdot (h^2/(sn))^{h/4}\\
		&  \mbox{~~~~~~~~~~~~~~~(Using $\alpha \leq \alpha_0 \leq 1$, $h'\leq h$, $h \leq 2^h$, $\kappa \leq o \le u$, $C_4 \geq 4\sqrt{C_1}\cdot C \cdot q^k$, $C_5 \geq 2\sqrt{C_2} $)}\\
		& \leq  16 C_4^u C_6^{h}  \cdot (sn/u^2)^{u/4} \cdot (n/\kappa)^\kappa \cdot (u/\sqrt{n\eta})^\eta \cdot (u/n)^o  \cdot  (sn/h^2)^{-(h-h')/4}\\
		&  \mbox{ ~~~~~~~~~~~~~~~(Using $(h/h')^{h'/2} \leq e^h$ and $C_6 \geq eC_5$ )}\\
		& =  16 C_4^u C_6^{h}  \cdot (sn/u^2)^{u/4} \cdot (n/\kappa)^\kappa \cdot (u/\sqrt{n\eta})^\eta \cdot (u/n)^o  \cdot  (sn/h^2)^{-(u+\kappa-(\eta+o))/4}\\
		& =  16 C_4^u C_6^{h}  \cdot (h^2/u^2)^{u/4} \cdot (h^2n^3/(s\kappa^4))^{\kappa/4} \cdot (su^4/(n\eta^2h^2))^{\eta/4} \cdot (su^4/n^3h^2)^{o/4}\\
		& = 16 C_4^u C_6^{h}  \cdot (h^2/u^2)^{u/4} \cdot (h^2n^2/(\epsilon\kappa^4))^{\kappa/4} \cdot (\epsilon u^4/(\eta^2h^2))^{\eta/4} \cdot (\epsilon u^4/(n^2h^2))^{o/4}\\
		& =:  S_1.
\end{align*}

Thus far we have not used $u \leq h$. (We have only used $u \leq s$ and this was to establish our goal as \cref{eq:case1}.) We now use $u \leq h$ to analyze $S_1$. 
\begin{eqnarray*}
S_1 & = & 16 C_4^u C_6^{h}  \cdot (h^2/u^2)^{u/4} \cdot (h^2n^2/(\epsilon\kappa^4))^{\kappa/4} \cdot (\epsilon u^4/(\eta^2h^2))^{\eta/4} \cdot (\epsilon u^4/(n^2h^2))^{o/4}\\
& = & 16 C_4^u C_6^{h}  \cdot \epsilon^{(-\kappa+\eta+o)/4} \cdot (h^2/u^2)^{u/4} \cdot (h^2n^2/\kappa^4)^{\kappa/4} \cdot (u^4/(\eta^2h^2))^{\eta/4} \cdot (u^4/(n^2h^2))^{o/4}\\
& & \mbox{ ~~~~~~~~~~~~~~~(Collecting $\epsilon$ terms)}\\
& \leq  &  16 C_4^u C_6^{h}  \cdot (h^2/u^2)^{u/4} \cdot (h^2n^2/(\kappa^4))^{\kappa/4} \cdot (u^4/(\eta^2h^2))^{\eta/4} \cdot (u^4/(n^2h^2))^{o/4}  \\
& & \mbox{ ~~~~~~~~~~~~~~~(Using $\eta\geq 0$, $o \geq \kappa$ and $\epsilon \leq 1$)}\\
& \leq & 16 (e^2 C_4)^u C_6^{h}  \cdot (h/u)^{u/2} \cdot (h n/(u^2))^{\kappa/2} \cdot (u/h)^{\eta/2} \cdot (u^2/(nh))^{o/2}\\
& & \mbox{ ~~~~~~~~~~~~~~~(Using $(u/\kappa)^\kappa \leq e^u$ and $(u/\eta)^\eta \leq e^u$)}\\
& = & 16 (e^2 C_4)^u C_6^{h}  \cdot (h/u)^{u/2} \cdot (u/h)^{\eta/2} \cdot (u^2/(nh))^{(o-\kappa)/2}\\
& \leq & 16 (e^2 C_4)^u C_6^{h} \cdot (h/u)^{u/2} \cdot (u/h)^{\eta/2} \cdot (u/h)^{(o-\kappa)/2}\\
& & \mbox{ ~~~~~~~~~~~~~~~(Using $u \leq n$ and $o \geq \kappa$)}\\
& = & 16 (e^2 C_4)^u C_6^{h}  \cdot (h/u)^{(u-\eta-o+\kappa)/2} \\
& \leq & 16 (e^2 C_4 C_6)^{h}  \cdot (h/u)^{(u-\eta-o+\kappa)/2} \mbox{~~~(Using $u \leq h$ and $e^2 C_4 \geq C_4 \geq 1$)}\\
& \leq & 16 (e^2 C_4 C_6)^{h} \cdot (h/u)^{u/2} \\
& & \mbox{ ~~~~~~~~~~~~~~~(Using $h \geq u$ and $u\geq u-o+\kappa - \eta$ since $o \geq \kappa$ and $\eta\geq 0$)}\\
& \leq & 16 (e^2 C_4 C_6)^{h} \cdot e^h \\
& & \mbox{ ~~~~~~~~~~~~~~~(Using $(h/u)^{u/2} \leq e^{h/2} \leq e^{h}$.)}\\
& \leq & 16 (\frac1{256} C_3)^{h/2} \mbox{ ~~~(Using $C_3 \geq (16 e^3 C_4 C_6)^2$)} \\
& \leq & C_3^{h/2} \mbox{ ~~~(Using $h\geq 1$).}
\end{eqnarray*}
This yields \cref{eq:case1} in the range $u \in [h]$. 

\paragraph{Case 2: $h < u \leq s$:} Here again our goal is to prove \cref{eq:case1} and we still have $L_1 \leq S_1$. We proceed as follows: 
\begin{eqnarray*}
		S_1
		& = & 16 C_4^u C_6^{h}  \cdot (h^2/u^2)^{u/4} \cdot (h^2n^2/(\epsilon\kappa^4))^{\kappa/4} \cdot (\epsilon u^4/(\eta^2h^2))^{\eta/4} \cdot (\epsilon u^4/(n^2h^2))^{o/4}\\
		&  \leq & 16 C_4^u C_6^{h} \epsilon^{(u-h)/4}  \cdot (h^2/u^2)^{u/4} \cdot (h^2n^2/(\kappa^4))^{\kappa/4} \cdot (u^4/(\eta^2h^2))^{\eta/4} \cdot (u^4/(n^2h^2))^{o/4}\\
		& & \mbox{ ~~~~~~~~~~~~~~~(Collecting $\epsilon$ terms and using $\eta+o-\kappa \geq u - h$)}\\
		& \leq & 16 C_4^u C_6^{h} \epsilon_0^{(u-h)/4}  \cdot (h^2/u^2)^{u/4} \cdot (h^2n^2/(\kappa^4))^{\kappa/4} \cdot (u^4/(\eta^2h^2))^{\eta/4} \cdot (u^4/(n^2h^2))^{o/4}\\
		& & \mbox{ ~~~~~~~~~~~~~~~(Using $\epsilon \leq \epsilon_0$ and $h \leq u$)}\\
		& \leq & 16 (e^2\epsilon_0^{1/4}C_4)^u (C_6/\epsilon_0^{1/4})^{h}  \cdot (h/u)^{u/2} \cdot (h n/(u^2))^{\kappa/2} \cdot (u/h)^{\eta/2} \cdot (u^2/(nh))^{o/2}\\
		& & \mbox{ ~~~~~~~~~~~~~~~(Using $(u/\kappa)^\kappa \leq e^u$ and $(u/\eta)^\eta \leq e^u$)}\\
		& = & 16 (e^2\epsilon_0^{1/4}C_4)^u (C_6/\epsilon_0^{1/4})^{h}  \cdot (h/u)^{u/2} \cdot (u/h)^{\eta/2} \cdot (u^2/(nh))^{(o-\kappa)/2}\\
		& \leq & 16 (e^2\epsilon_0^{1/4}C_4)^u (C_6/\epsilon_0^{1/4})^{h}  \cdot (h/u)^{u/2} \cdot (u/h)^{\eta/2} \cdot (u/h)^{(o-\kappa)/2}\\
	    & & \mbox{ ~~~~~~~~~~~~~~~(Using $u \leq n$ and $o \geq \kappa$)}\\
	    & = & 16 (e^2\epsilon_0^{1/4}C_4)^u (C_6/\epsilon_0^{1/4})^{h}  \cdot (h/u)^{(u-\eta-o+\kappa)/2} \\
	    & \leq & 16 (e^2\epsilon_0^{1/4}C_4)^u  (C_6/\epsilon_0^{1/4})^h \mbox{ ~~~~(Since $h \leq u$ and $u-\eta-o+\kappa \geq u-\eta-o \geq 0$)}\\
	    & \leq & 16 (C_6/\epsilon_0^{1/4})^{h} \mbox{~~~(Using $\eps_0 \leq \frac{1}{(e^2C_4)^4}$)}\\
	    & \leq & C_3^{h/2} \mbox{ ~~~(Using $C_3\geq (16C_6/\epsilon_0^{1/4})^2$, $h \geq 1$)}\,.
\end{eqnarray*}
This concludes \cref{eq:case1} in Case 2. 

\paragraph{Case 3: $s < u \leq 16s$:} The form for $U_{C_1,s}(u)$ now changes and forces a change in our goal. Using $U_{C_1,s}(u) \leq C_1^{u/2} (n/u)^{u/4}$ our new goal becomes: 
\begin{align}
		& \left(C_1^{u/2} (n/u)^{u/4}\right) \left(\alpha^{(o+\eta)/k} C^u (n/\kappa)^\kappa (u/\sqrt{n\eta})^\eta (u/n)^o\right) \left(C_2^{h'/2} (s n/h'^2)^{h'/4}\right) \cdot h \cdot q^{k\kappa} \nonumber\\ 
		& ~~~~~~~~~~ \leq 4^{-u-2}\cdot C_3^{h/2} \cdot (s n/h^2)^{h/4} = 4^{-u-2}\cdot U_{C_3,s}(h)\, , \label{eq:case3}
\end{align}
Again multiplying the LHS by $4^{u+2} (h^2/sn)^{h/4}$ we get the quantity $L_3$ below which we show to be upper bounded by $C_3^{h/2}$. We have:

\begin{align*}
L_3
& :=  16 (16C_1)^{u/2} (n/u)^{u/4} \cdot \alpha^{(o+\eta)/k} C^u (n/\kappa)^\kappa (u/\sqrt{n\eta})^\eta (u/n)^o \cdot C_2^{h'/2} (sn/h'^2)^{h'/4}\cdot h \cdot q^{k\kappa} \cdot (h^2/(sn))^{h/4} \\
& = \text{\small $(u/s)^{u/4} 16 (16C_1)^{u/2} (sn/u^2)^{u/4} \cdot \alpha^{(o+\eta)/k} C^u (n/\kappa)^\kappa (u/\sqrt{n\eta})^\eta (u/n)^o \cdot C_2^{h'/2} (sn/h'^2)^{h'/4}\cdot h \cdot q^{k\kappa} \cdot (h^2/(sn))^{h/4}$} \\
& = (u/s)^{u/4} L_1 \\
& \leq 16^{u/4} L_1 
\end{align*}

We now use the fact that the inequality $L_1\leq S_1$ in Case 1, did not use $u \leq s$. We thus conclude $L_3 \leq 16^{u/4} L_1 \leq 16^{u/4} S_1$. Similarly the inequality $S_1 \leq 16 (e^2\epsilon_0^{1/4}C_4)^u  (C_6/\epsilon_0^{1/4})^h$ from Case 2 also did not use $u \leq s$. So we get $L_3 \leq 16^{u/4} S_1 \leq 16^{u/4} \cdot 16 (e^2\epsilon_0^{1/4}C_4)^u  (C_6/\epsilon_0^{1/4})^h$ which we simplify below. We have:
\begin{align*}
L_3 & \leq 16^{u/4} \cdot 16 (e^2\epsilon_0^{1/4}C_4)^u  (C_6/\epsilon_0^{1/4})^h \\
& = 16 (2 e^2 \epsilon_0^{1/4}C_4)^u (C_6/\epsilon_0^{1/4})^h \\
& \leq  16 (C_6/\epsilon_0^{1/4})^{h} \mbox{~~~(Using $\eps_0 \leq \frac{1}{(2 e^2C_4)^4}$)}\\
& \leq  C_3^{h/2} \mbox{ ~~~(Using $C_3\geq (16C_6/\epsilon_0^{1/4})^2$, $h \geq 1$)}\,.
\end{align*}
This concludes Case 3.

\paragraph{Case 4: $16s < u \leq \sqrt{\epsilon}n$:} Here again it suffices to prove \cref{eq:case3} which is equivalent to proving
$L_3 \leq C_3^{h/2}$. We have 
\begin{align*}
L_3
& =  \text {\small $16 (16C_1)^{u/2} (un/u^2)^{u/4} \cdot \alpha^{(o+\eta)/k} C^u (n/\kappa)^\kappa (u/\sqrt{n\eta})^\eta (u/n)^o \cdot C_2^{h'/2} (sn/h'^2)^{h'/4}\cdot h \cdot q^{k\kappa} \cdot (h^2/(sn))^{h/4}$} \\
& =  \text {\small $16 (16C_1)^{u/2} (u/s)^{u/4} (sn/u^2)^{u/4} \cdot \alpha^{(o+\eta)/k} C^u (n/\kappa)^\kappa (u/\sqrt{n\eta})^\eta (u/n)^o \cdot C_2^{h'/2} (sn/h'^2)^{h'/4}\cdot h \cdot q^{k\kappa} \cdot (h^2/(sn))^{h/4}$} \\
& \leq  \text {\small $16 (16C_1)^{u/2} \epsilon^{-u/8} (sn/u^2)^{u/4} \cdot \alpha^{(o+\eta)/k} C^u (n/\kappa)^\kappa (u/\sqrt{n\eta})^\eta (u/n)^o \cdot C_2^{h'/2} (sn/h'^2)^{h'/4}\cdot h \cdot q^{k\kappa} \cdot (h^2/(sn))^{h/4}$} \\
&  ~~~~~~~~~~~~~~~~~ \mbox{(Using $s=\epsilon n$ and $u \leq \sqrt{\epsilon}n$ yielding $u/s \leq \sqrt{\epsilon}/\epsilon = \epsilon^{-1/2}$)}\\
		& \leq   16 (C_4/\epsilon^{1/8})^u C_5^{h} \cdot (sn/u^2)^{u/4} \cdot (n/\kappa)^\kappa \cdot (u/\sqrt{n\eta})^\eta \cdot (u/n)^o  \cdot  (sn/h'^2)^{h'/4}\cdot (h^2/(sn))^{h/4}\\
		&  \mbox{ ~~~~~~~~~~~~~~~\small (Using $\alpha \leq \alpha_0 \leq 1$, $h'\leq h$, $h \leq 2^h$, $\kappa \leq o \le u$, $C_4 \geq 4\sqrt{C_1}\cdot C \cdot q^k$, $C_5 \geq \max\{1,2C_2,2\sqrt{C_2}\} $)}\\
		& \leq  16 (C_4/\epsilon^{1/8})^u C_6^{h}  \cdot (sn/u^2)^{u/4} \cdot (n/\kappa)^\kappa \cdot (u/\sqrt{n\eta})^\eta \cdot (u/n)^o  \cdot  (sn/h^2)^{-(h-h')/4}\\
		&  \mbox{ ~~~~~~~~~~~~~~~(Using $(h/h')^{h'/2} \leq e^h)$ and $C_6 \geq eC_5$ )}\\
		& =  16 (C_4/\epsilon^{1/8})^u C_6^{h}  \cdot (sn/u^2)^{u/4} \cdot (n/\kappa)^\kappa \cdot (u/\sqrt{n\eta})^\eta \cdot (u/n)^o  \cdot  (sn/h^2)^{-(u+\kappa-(\eta+o))/4}\\
		& =  16 (C_4/\epsilon^{1/8})^u C_6^{h}  \cdot (h^2/u^2)^{u/4} \cdot (h^2n^3/(s\kappa^4))^{\kappa/4} \cdot (su^4/(n\eta^2h^2))^{\eta/4} \cdot (su^4/n^3h^2)^{o/4}\\
		& =  16 (C_4/\epsilon^{1/8})^u C_6^{h}  \cdot (h^2/u^2)^{u/4} \cdot (h^2n^2/(\epsilon\kappa^4))^{\kappa/4} \cdot (\epsilon u^4/(\eta^2h^2))^{\eta/4} \cdot (\epsilon u^4/(n^2h^2))^{o/4}\\
		& = 16 C_4^u C_6^{h} \cdot \epsilon^{-u/8-\kappa/4+\eta/4+o/4  } \cdot (h^2/u^2)^{u/4} \cdot (h^2n^2/\kappa^4)^{\kappa/4} \cdot (u^4/(\eta^2h^2))^{\eta/4} \cdot (u^4/(n^2h^2))^{o/4}\\
		&  \mbox{ ~~~~~~~~~~~~~~~(Collecting $\epsilon$ terms)}\\
		&  \leq 16 C_4^u C_6^{h} \epsilon^{u/8-h}  \cdot (h^2/u^2)^{u/4} \cdot (h^2n^2/(\kappa^4))^{\kappa/4} \cdot (u^4/(\eta^2h^2))^{\eta/4} \cdot (u^4/(n^2h^2))^{o/4}\\
		&  \mbox{ ~~~~~~~~~~~~~~~(Using $\eta+o-\kappa \geq u - h$)}\\
		& \leq 16 C_4^u C_6^h \epsilon^{u/16} \cdot (h^2/u^2)^{u/4} \cdot (h^2n^2/(\kappa^4))^{\kappa/4} \cdot (u^4/(\eta^2h^2))^{\eta/4} \cdot (u^4/(n^2h^2))^{o/4}\\
		&  \mbox{ ~~~~~~~~~~~~~~~(Using $u \geq 16s \geq 16h$ in the form $h \leq u/16$ to conclude $u/8-h \geq u/16$.)}\\
		& \leq 16 C_4^u C_6^{h} \epsilon_0^{u/16}  \cdot (h^2/u^2)^{u/4} \cdot (h^2n^2/(\kappa^4))^{\kappa/4} \cdot (u^4/(\eta^2h^2))^{\eta/4} \cdot (u^4/(n^2h^2))^{o/4}\\
		&  \mbox{ ~~~~~~~~~~~~~~~(Using $\epsilon \leq \epsilon_0$)}\\
		& \leq 16 (e^2\epsilon_0^{1/16}C_4)^u C_6^{h}  \cdot (h/u)^{u/2} \cdot (h n/(u^2))^{\kappa/2} \cdot (u/h)^{\eta/2} \cdot (u^2/(nh))^{o/2}\\
		& \mbox{ ~~~~~~~~~~~~~~~(Using $(u/\kappa)^\kappa \leq e^u$ and $(u/\eta)^\eta \leq e^u$)}\\
        & = 16 (e^2\epsilon_0^{1/16}C_4)^u C_6^{h}  \cdot (h/u)^{(u+\kappa-\eta-o)/2} \cdot (u/n)^{(o-\kappa)/2}\\
        & \leq 16(e^2\epsilon_0^{1/16}C_4)^u C_6^{h}  \cdot (h/u)^{(u+\kappa-\eta-o)/2} \mbox{~~~~~~~~ (Using $u \leq n$)}\\
        &\leq 16 (e^2\epsilon_0^{1/16}C_4)^u C_6^{h} \mbox{~~~~~~(Using $h \leq u$ and $u+\kappa-\eta - o \geq \kappa \geq 0$)} \\
        & \leq 16 C_6^h \mbox{~~~~~~~~ (Using $\epsilon_0\leq (e^{-2}/C_4)^{16}$)}\\
        & \leq 16 (C_3^{1/2}/16)^h \mbox{~~~~~~~~ (Using $C_3 \geq 256 C_6^2$)}\\
        & \leq C_3^{h/2} \mbox{~~~~~~~~ (Using $h\geq 1$).}\\
\end{align*}
This establishes \cref{eq:case3} in Case 4. 

\paragraph{Case 5: $\sqrt{\epsilon}n < u \leq n$:} 
Here we use $U_{C_1,s}(u) \leq (2q^2e^2n/u)^{u/2}$. 
With this modification we need to prove:
\begin{align}
		& \left((2q^2e^2)^{u/2} (n/u)^{u/2}\right) \left(\alpha^{(o+\eta)/k} C^u (n/\kappa)^\kappa (u/\sqrt{n\eta})^\eta (u/n)^o\right) \left(C_2^{h'/2} (s n/h'^2)^{h'/4}\right) \cdot h \cdot q^{k\kappa} \nonumber\\ 
		& ~~~~~~~~~~ \leq 4^{-u-2}\cdot C_3^{h/2} \cdot (s n/h^2)^{h/4} = 4^{-u-2}\cdot U_{C_3,s}(h)\, , \label{eq:case5}
\end{align}
Multiplying the LHS by $4^{u+2} (h^2/sn)^{h/4}$ we get the term $L_5$ defined below which we wish to upper bound by $C_3^{h/2}$. 
\begin{align*}
	L_5 & := \text{\small $16 \cdot 4^u (2q^2 e^2)^{u/2} (n/u)^{u/2} \alpha^{(o+\eta)/k} C^u (n/\kappa)^\kappa (u/\sqrt{n\eta})^\eta (u/n)^o \cdot C_2^{h'/2} (sn/h'^2)^{h'/4} \cdot h \cdot q^{k\kappa} \cdot (h^2/(sn))^{h/4}$} \\
		& \leq  16 C_0^u C_5^{h} \alpha^{(o+\eta)/k} \cdot (n/u)^{u/2} \cdot (n/\kappa)^\kappa \cdot (u/\sqrt{n\eta})^\eta \cdot (u/n)^o  \cdot  (sn/h'^2)^{h'/4}\cdot (h^2/(sn))^{h/4}\\
		&  \mbox{ ~~~~~~~~~~~~~~~(Using $h'\leq h \leq 2^h$, $\kappa \leq u$, $C_0 \geq \sqrt{2}\cdot 4e \cdot q\cdot C \cdot q^k$, $C_5 \geq 2\sqrt{C_2} $)}\\
		& \leq 16  C_0^u C_6^{h} \alpha^{(o+\eta)/k} \cdot (n/u)^{u/2} \cdot (n/\kappa)^\kappa \cdot (u/\sqrt{n\eta})^\eta \cdot (u/n)^o  \cdot  (sn/h^2)^{-(h-h')/4}\\
		&  \mbox{ ~~~~~~~~~~~~~~~(Using $(h/h')^{h'/2} \leq e^h$ and $C_6 \geq eC_5$ )}\\
		& =  16 C_0^u C_6^{h} \alpha^{(o+\eta)/k} \cdot (n/u)^{u/2} \cdot (n/\kappa)^\kappa \cdot (u/\sqrt{n\eta})^\eta \cdot (u/n)^o  \cdot  (sn/h^2)^{-(u+\kappa-(\eta+o))/4}\\
		& = 16  C_0^u C_6^{h} \alpha^{(o+\eta)/k} \cdot (nh^2/(su^2))^{u/4} \cdot (h^2n^3/(s\kappa^4))^{\kappa/4} \cdot (su^4/(n\eta^2h^2))^{\eta/4} \cdot (su^4/n^3h^2)^{o/4}\\
		& = 16  C_0^u C_6^{h} \alpha^{(o+\eta)/k} \cdot (h^2/(\epsilon u^2))^{u/4} \cdot (h^2n^2/(\epsilon\kappa^4))^{\kappa/4} \cdot (\epsilon u^4/(\eta^2h^2))^{\eta/4} \cdot (\epsilon u^4/(n^2h^2))^{o/4}\\
		&= 16 C_0^u C_6^{h} \alpha^{(o+\eta)/k} \cdot \epsilon^{(-u-\kappa+\eta+o)/4} \cdot (h^2/( u^2))^{u/4} \cdot (h^2n^2/(\kappa^4))^{\kappa/4} ( u^4/(\eta^2h^2))^{\eta/4}  (u^4/(n^2h^2))^{o/4} \\
		& \leq 16 (e^2C_0)^u C_6^{h} \alpha^{(o+\eta)/k} \cdot \epsilon^{(-u-\kappa+\eta+o)/4} \cdot (h/u)^{u/2} \cdot (h n/(u^2))^{\kappa/2} \cdot (u/h)^{\eta/2} \cdot (u^2/(nh))^{o/2}\\
		& \mbox{ ~~~~~~~~~~~~~~~(Using $(u/\kappa)^\kappa \leq e^u$ and $(u/\eta)^\eta \leq e^u$)}\\
		& = 16 (e^2C_0)^u C_6^{h} \alpha^{(o+\eta)/k} \cdot \epsilon^{(-u-\kappa+\eta+o)/4} \cdot (h/u)^{u/2} \cdot (u/h)^{\eta/2} \cdot (u^2/(nh))^{(o-\kappa)/2}\\
		& \leq 16 (e^2C_0)^u C_6^{h} \alpha^{(o+\eta)/k} \cdot \epsilon^{(-u-\kappa+\eta+o)/4} \cdot (h/u)^{u/2} \cdot (u/h)^{\eta/2} (u/h)^{(o-\kappa)/2}\\
		& \mbox{ ~~~~~~~~~~~~~~~(Using $u \leq n$ and $o \geq \kappa$)}\\
		& \leq 16 (e^2C_0)^u C_6^{h} \alpha^{(u-h)/k} \cdot \epsilon^{(-u-\kappa+\eta+o)/4} \cdot (h/u)^{(u-\eta-o+\kappa)/2}\\
		& \mbox{ ~~~~~~~~~~~~~~~(Using $\alpha \leq 1$ and $u-h \leq \eta + o$)}\\
		& \leq 16 (e^2C_0)^u C_6^{h} \alpha_0^{(u-h)/k} \cdot \epsilon^{(-u-\kappa+\eta+o)/4} \cdot (h/u)^{(u-\eta-o+\kappa)/2}\\
		& \mbox{ ~~~~~~~~~~~~~~~(Using $\alpha \leq \alpha_0$ and $u-h \geq 0$)}\\
		& = 16 (\alpha_0^{1/k}e^2C_0)^u (C_6/\alpha_0^{1/k})^{h} (h^2/(\epsilon u^2))^{(u-\eta-o+\kappa)/4}\\
		& \leq 16 (\alpha_0^{1/k}e^2C_0)^u (C_6/\alpha_0^{1/k})^{h}\\
		& \mbox{~~~~~(Using $h \leq s = \epsilon n$ and $u \geq \sqrt{\epsilon}n$ to conclude $h^2/(\epsilon u^2)\leq 1$. Also using $u-\eta-o+\kappa \geq 0$)}\\
		& \leq 16 (C_6/\alpha_0^{1/k})^{h}  ~~~\mbox{(Using $\alpha_0 \leq 1/(e^2C_0)^k$)} \\
		& \leq 16 (C_3^{1/2}/16)^h~~~\mbox{(Using $C_3 \geq (16 C_6/\alpha_0^{1/k})^2$)} \\
		& \leq C_3^{h/2} ~~~\mbox{(Using $h \geq 1$).}\,
\end{align*}
This concludes the analysis of Case 5 and proves the lemma.
\end{proof}

	\paragraph{Step 3: Proof of~\cref{lem:induction step exp}.}

We are now ready to combine the ingredients from the previous steps to prove \cref{lem:induction step exp}. 
\begin{proof}[Proof of~\autoref{lem:induction step exp}]

	Let $\alpha_0$ be the as given by \cref{lem:q-simplify}. Let
	$\epsilon_0$ and $C_3$ be the parameters given by \cref{lem:q-simplify} for $C_1 = C$ and $C_2 = C_0$. 
	We prove the lemma for $C'' = C_3$ and $\tau_0 = \epsilon_0$.
	
	Let $\alpha \leq \alpha_0$ and $m=\alpha n$. For every $s \leq \epsilon_0 n$, we prove that the LHS in the lemma statement is upper bounded by $W_{C_3,s}(h)$ for every $h \in [s]$. 
	In the following, for every matching $M$ of size $m$, we fix a $B' = B'(M)$ that is $(M,C_0,s)$-reduced. (The inequalities hold for every such fixing.)
	We have 
	\begin{align}
		\lefteqn{\sum_{\substack{\vecu\in\Z_q^n}}\frac{q^n}{|B|}\left|\widehat{\mathbf{1}_B}(\vecu)\right|\Exp_M\left[\sum_{\substack{\vecu'\in\Z_q^n\\\|\vecu+\vecu'\|_0=h}}\frac{q^n}{|B'|}\left|\widehat{\mathbf{1}_{B'}}(\vecu')\right|\right]} \nonumber\\
		& =\ \sum_{u=0}^n \sum_{\substack{\vecu\in\Z_q^n\\ |\supp(\vecu)| = u}}\frac{q^n}{|B|}\left|\widehat{\mathbf{1}_B}(\vecu)\right|\Exp_M\left[\sum_{\substack{\vecu'\in\Z_q^n\\\|\vecu+\vecu'\|_0=h}}\frac{q^n}{|B'|}\left|\widehat{\mathbf{1}_{B'}}(\vecu')\right|\right] \nonumber\\
		&\leq\ \sum_{u=0}^n \sum_{\substack{\vecu\in\Z_q^n\\ |\supp(\vecu)| = u}}\frac{q^n}{|B|}\left|\widehat{\mathbf{1}_B}(\vecu)\right| \cdot \sum_{o,\eta,\kappa} p_{q,\alpha}(n,u,o,\eta,\kappa) \cdot h \cdot 2^{k\kappa} \cdot U_{C_0,s}(h+o+\eta-(u+\kappa)) \label{eq:qsummation}\\
		& \leq\ \sum_{u=0}^n \sum_{o,\eta,\kappa} U_{C,s}(u) \cdot p_{q,\alpha}(n,u,o,\eta,\kappa) \cdot h \cdot 2^{k\kappa} \cdot U_{C_0,s}(h+o+\eta-(u+\kappa)) \label{eq:bbd}\\
		& \leq\ \sum_{u=0}^n \sum_{o,\eta,\kappa} 4^{-u-2} W_{C_3,s}(h) \label{eq:qubound}\\
		& \leq\ \sum_{u=0}^n (u+1)^3 \cdot 4^{-u-2} \cdot W_{C_3,s} \nonumber\\
		& \leq\ W_{C_3,s}(h), \nonumber
	\end{align}
	where \eqref{eq:qsummation} follows from \autoref{lem:mcssum} and the fact that $B'$ is $(M,C_0,s)$-reduced, \eqref{eq:bbd} follows from the fact that $B$ is $(C,s)$-bounded, and \eqref{eq:qubound} follows from Lemma~\ref{lem:q-simplify} for $C_3$ as defined above. This proves the lemma.
\end{proof}

\section*{Acknowledgments}

Thanks to Raghuvansh Saxena for pointing out errors in previous versions of this paper including a significant error in a previous proof of \cref{lem:q-simplify}. Thanks to Noah Singer for many valuable comments on the paper including pointing out the use of inconsistent and ambiguous notation and some significant errors (that are hopefully fixed in this version). Thanks to an anonymous conference referee for pointing out some errors in the proof of \cref{lem:hyper}. Thanks to the anonymous SICOMP referees for the careful reading of the paper and for the helpful and detailed comments.

\bibliography{mybib}
\bibliographystyle{alpha}

\end{document}